\documentclass[12pt]{amsart}
\usepackage{mathrsfs}
\usepackage[colorlinks]{hyperref}
\usepackage{xcolor}
\usepackage{color}
\usepackage{kpfonts}
\usepackage[toc,page]{appendix}
\usepackage{ulem}
\usepackage{cancel}
\usepackage{dsfont}

\usepackage{mathtools}
\usepackage{todonotes}

\usepackage{fouridx}

\usepackage{scalerel}
\usepackage{stackengine,wasysym}

\def\Re{\mathop{ \rm Re }}
\newcommand\reallywidetilde[1]{\ThisStyle{%
  \setbox0=\hbox{$\SavedStyle#1$}%
  \stackengine{-.1\LMpt}{$\SavedStyle#1$}{%
    \stretchto{\scaleto{\SavedStyle\mkern.2mu\AC}{.5150\wd0}}{.6\ht0}%
  }{O}{c}{F}{T}{S}%
}}

\def\test#1{$%
  \reallywidetilde{#1}\,
$\par}
\newcommand\newtilde[1]{\mbox{\test{#1}}}

\usepackage{setspace}

\usepackage{amsthm,amssymb,mathrsfs,mathtools,empheq}
\usepackage[shortlabels]{enumitem}
\usepackage[T1]{fontenc}

\mathtoolsset{showonlyrefs}
\usepackage{color,stmaryrd,bbm}
\usepackage{fullpage}

\usepackage{mathscinet}
\usepackage{latexsym}
\usepackage{amsthm}

\usepackage{mathscinet}
\usepackage{latexsym}
\usepackage{amsthm}
\usepackage{amssymb}
\usepackage{amsfonts}
\usepackage{amsmath}
\usepackage{verbatim}

\newtheorem{theorem}{Theorem}[section]
\newtheorem{lemma}[theorem]{Lemma}
\newtheorem{proposition}[theorem]{Proposition}
\newtheorem{corollary}[theorem]{Corollary}

\newtheorem{convention}[theorem]{\emph{Convention}}

\theoremstyle{definition}
\newtheorem{definition}[theorem]{Definition}
\newtheorem{example}[theorem]{Example}
\theoremstyle{remark}
\newtheorem{remark}[theorem]{Remark}

\numberwithin{equation}{section}
\def\d{{\rm d}}
\newcommand{\e}{{\rm e}}

\newcommand{\Di}{D_i}
\newcommand{\Dj}{D_j}

\renewcommand{\D}{D}

\newcommand{\Djk}{D_{jk}}

\newcommand{\Dir}{\mathscr{Dir}}

\newcommand{\Djeden}{\frac{\d  }{\d x}}

\newcommand\embed{\hookrightarrow}

\newcommand{\eps}{\varepsilon}

\newcommand{\DeltaRn}{\Delta}
\newcommand\coma[1]{{\color{red}#1}}
\newcommand\adda[1]{{\color{blue}#1}}

\newcommand\dela[1]{}

\newcommand\lpo{L^p(\mathscr{O})}
\newcommand\wkp{W^{k,p}(\mathscr{O})}
\newcommand\supp{\mathop{ \rm supp }}
\newcommand\lpr{L^p(\mathbb{R} ^n)}
\newcommand\wzp{W^{2,p}(\mathbb{ R}^n \setminus \{ 0 \} )}
\newcommand\wzpo{W^{2,p}_0(\mathbb{R}^n \setminus \{ 0 \} )}
\newcommand\wzqo{W^{2,q}_0(\mathbb{R}^n \setminus \{ 0 \} )}

\newcommand\wdp{W^{\Delta,p}(\mathbb{ R}^n \setminus \{ 0 \} )}
\newcommand\wdpo{W^{\Delta,p}_0(\mathbb{R}^n \setminus \{ 0 \} )}

\newcommand\n{\hbox{\vrule width 1pt height 8ptdepth 2pt}}

\newcommand{\bdm}{\begin{displaymath}}
\newcommand{\edm}{\end{displaymath}}
\newcommand{\ds}{\displaystyle}

\newcommand{\range}{\mathrm{Range}}

\begin{document}
\title{Semigroups on generalized Sobolev spaces associated with Laplacians with  applications Stochastic PDEs with singular boundary conditions}
\author{Sergio Albeverio}
\address{ Institute for Applied Mathematics, Rheinische Friedrich-Wilhelms
Universit\"{a}t Bonn, and Hausdorff Center for Mathematics, Endenicher Allee
60, D-53115 Bonn, Germany}
\email{albeverio@iam.uni-bonn.de}

\author{Zdzis{\l}aw Brze{\'z}niak}

\address{Department of Mathematics, The University of York, Heslington, York, YO105DD,
UK}
\email{zdzislaw.brzezniak@york.ac.uk}

\author{Szymon Peszat}
\address{Institute of Mathematics, Jagiellonian University, {\L}ojasiewicza 6, 30--348 Krak\'ow, Poland.}
\email{napeszat@cyf-kr.edu.pl}

\maketitle


\begin{abstract}
Laplacians associated with domains with singular boundary conditions and are considered together with semigroups on generalized Sobolev spaces, they generate. Applications are given to stochastic PDEs with singular boundary conditions.
\end{abstract}

\section{Introduction}

The description of many dynamical systems in the material world, both by deterministic and stochastic methods, often has to take into account the natural presence of boundary conditions.

Boundary conditions may depend on the geometry of the underlying spaces and are often characterized by specifying values for the solutions and their derivatives on the boundary, such as Dirichlet-type, Neumann-type, or mixed conditions.

Sturm-Liouville problems provide examples of ordinary deterministic differential equations. Probabilistic methods offer new insights into these problems. See the example
 \cite{Fe54} and \cite{Mandl68}.

Vast extensions to PDEs and stochastic PDEs have been discussed, see, e.g., \cite{Ka44}, \cite{Dy82}, \cite{Sk61}, \cite{Tai88} and \cite{Fr67} for elliptic and parabolic systems and \cite{Ch} for hyperbolic systems.

In many areas of the applications of mathematics, it has turned out to be useful to consider other types of boundary conditions, namely singular ones, in the sense that they must be defined either by generalized functions supported on the boundary or, more generally, by suitable limits of regularized systems and boundary conditions. A simple example was encountered first in the modeling of very short-range strong interactions in physical systems, see \cite{Fermi36}, \cite{Breit47} and references in \cite{AGGH-K_2005}, \cite{AK00}, \cite{Gue23}, \cite{FiST24}, \cite{GaMi23}, \cite{Dell'An23}. The prototype is described by the Laplacian (possibly perturbed by regular and potential terms) on $\mathbb{R}^n$, 
with $n=1,2,3$, defined first on continuous functions with support strictly outside of the origin in $\mathbb{R}^n$.  Self-adjoint extensions of this  operator have been already indicated in 
\cite{BeFa61} and \cite{MiFa62}. They are related to the theory of extensions of symmetric operators on Hilbert space and correspond to a limit problem starting from a regularized weighted potential term \( \lambda_{\eps} \delta_{\eps}(\cdot) \). In the limit as \( \eps \to 0 \), with \( \delta_{\eps}(\cdot) \) being a regularized scaled Dirac delta distribution at the origin, and \( \lambda_{\eps} \) being suitable function radial function of \( \eps \), see \cite{AGGH-K_2005}, \cite{Ne95} and \cite{AFH-KL86}.
These extensions generate self-adjoint semigroups in \( L^2(\mathbb{R}^n) \). Some of them are Markovian and constitute transition semigroups of Markov diffusion processes with boundary conditions at the origin, which have been discussed in detail, see e.g. \cite{AGGH-K_2005}, \cite{AH-KS77}, \coma{\cite{AH-KS80}}, \cite{AFKS81} and \cite{FuOTa11}. A probabilistic description of some of these can be obtained through the theory of quadratic forms and associated diffusion processes, see \cite{AH-KS77}, and e.g. \cite{Alb03} and \cite{AGGH-K_2005}.
Another direct description is given in \cite{FM04} and \cite{FMV07}. The use of such singular boundary conditions has been vastly extended to many more models and situations, e.g.  moving point interactions on infinitely man discrete manifolds, or on sub-manifolds, or on random sets,  see e.g. references in \cite{AGGH-K_2005}, \cite{AK00}. 
The use of singular boundary conditions in probabilistic problems also comes from other sources, e.g., from the study of SPDEs (elliptic, parabolic, hyperbolic) in connection with problems in quantum field theory, hydrodynamics, or electromagnetism., see e.g.  \cite{Ne95}, \cite{AH-K84}, \cite{AR91} and \cite{J-LM85} for early works and \cite{GubH19}, \cite{Hai14},   \cite{ALR} and \cite{AK20} for recent developments. Here, the SPDEs exhibit interesting singularities and non-linearities due to physical constants, and boundary conditions naturally arise when considering limits with space-time regularization, which are later removed by a procedure defined by the  generic expression "renormalization". Problems also arise here, similar to the above-mentioned point interactions, which are modeled and play an important role in the development of models, see e.g.  \cite{Hepp69}, \cite{MiFa62}, \dela{\coma{\cite{BeMi61}}}, \cite{DiRa04}, \cite{AFig18}, and also the new references \cite{FiST24} and \cite{Dell'An23}.

Let us also  mention  the use of deterministic and stochastic Wentzell-type boundary conditions in SPDEs, which often arise in other contexts, see e.g. \cite{Mas_2010} and references therein. The semigroups associated with singular boundary conditions have also been studied in \( L^p \) spaces, see, e.g.  references  in \cite{ABD_1995-a}, \cite{ClTi86}\dela{\coma{\cite{Raci13}}} and \cite{ALR}. In the present paper, we systematically study singular boundary conditions in various spaces; e.g. Sobolev and Lebesgue \( L^p \) spaces, having mainly in mind further  applications to SPDEs.

\

The content of the paper is roughly as follows. Section 2 presents some notation and definitions. In Section 3, we study sub-left spaces of all the \( P \) associated with the Laplacian. The main result is given in Theorem 3.12, which provides a characterization of the elements of searchable spaces. The independence of the representation of an element with respect to the parameters is specified in Theorem 3.13 and its Corollaries 3.14, 3.15, and 3.16.

In particular, the relations between the spaces \( W^{\Delta, p} \) and the classical \( W^{2, p} \) spaces are discussed in Propositions 3.19 and 3.20. In Corollary 3.9, it is stressed that something specific happens when \( n = 2 \) (to be detailed later). Examples 3.21 and 3.22 further illustrate the case \( n = 1 \). In Remark 3.23, we point out the cases \( n = 2 \) and \( n = 3 \), where the space \( W^{2, p} \), for suitable \( p \), consists only of continuous functions.

A representation formula for the elements in the space \( W^{2, p} \) is provided in Theorems 3.12 and 3.13, and in the related corollaries.

In Section 4, a Green's formula on \( \mathbb{R}^N \) is established, as well as a Stokes-type operator for elements in the space \( W^{\Delta, p} \).

In Section 5, we prove a slightly more general result than the one in paper [27] for \( C_0 \)-semigroups on real Banach spaces. We then discuss the special case of the linear unbounded operator \( A_0 \), covering the Laplacian \( \Delta \) on the space \( C_c^\infty \). We consider the operator \( A_1 \), which coincides with \( \Delta \) on the space \( W^{2, p}(\mathbb{R}^N) \), and extends \( A_0 \) in the class of generators of \( C_0 \)-semigroups on \( L^p \)-spaces for \( n = 3 \) and \( p < \frac{n}{2} \). Other extensions are discussed for \( p > \frac{n}{2} \).

The discussion is further extended in Section 6.

In Section 7, symmetric extensions of the Laplacian operator on \( \mathbb{R}^N \) are discussed in detail.

Previous results are recovered in a systematic way and discussed in Sections 9 and 10, which are devoted to the application of the study of the heat equation evolution on semigroups with singular boundary conditions, especially in the case \( n = 2 \). The relations with the Gaussian Markov families in the \( L^p \) spaces are also pointed out.


Let $\Delta_{\mathbb{R}^n\setminus\{0\}}$ be the distributional Laplace operator on $\mathscr{D}'(\mathbb{R}^n\setminus\{0\})$. Given a function\footnote{In the whole paper we consider complex valued functions and distributions. We will not introduce  any  specific notation for these objects. } $u\in L^p(\mathbb{R}^n\setminus\{0\})=L^p(\mathbb{R}^n\setminus\{0\},\mathbb{C})$ such that $\Delta_{\mathbb{R}^n\setminus\{0\}}u $ belongs to $L^p(\mathbb{R}^n\setminus\{0\})$, we consider
$$
R (u):= \newtilde{\Delta_{\mathbb{R}^n\setminus\{0\}}u}- \Delta \newtilde{u},
$$
where for a $v\colon  \mathbb{R}^n\setminus\{0\} \to \mathbb{C}$,
\begin{equation*}
\newtilde{v}(x)
=
\begin{cases} v(x), \quad  &\text{if $x \notin  \mathbb{R}^n\setminus\{0\}$}, \\
0, \quad  &\text{if $x  =0$}.
\end{cases}
\end{equation*}

Our first result, see Lemma \ref{lem-3.3} characterizes  the difference $R(u)$. Next, see Theorem \ref{T3.10}, we estimate haw $W^{\Delta,p}(\mathbb{R}^n\setminus\{0\})$ differs from $W^{2,p}(\mathbb{R}^n)$. Then we are concerned with the operators $(A,D(A))$ on $L^p(\mathbb{R}^n\setminus\{0\})$ such that
\begin{align*}
D(A)\cap D(A^\ast) &\supset \mathscr{C}_0^\infty(\mathbb{R}^n\setminus \{0\}),\\
Au =  A^\ast u &=  \Delta_{\mathbb{R}^n\setminus\{0\}}u, \qquad u\in \mathscr{C}_0^\infty(\mathbb{R}^n\setminus \{0\}).
\end{align*}
To this end, crucial  roles are played by Theorem \ref{T3.10} and a Green formula on $\mathbb{R}^n\setminus\{0\}$, see Theorem \ref{T4.2}.  In Section \ref{S8}, and  also in Appendix  \ref{Section7}, we consider the important case of $p=2$ and self-adjoint extensions of $\Delta_{\mathbb{R}^n\setminus\{0\}}$. Finally we consider  stochastic  boundary problems on $\mathbb{R}^n\setminus\{0\}$.

\subsection{Notation}
The set of all natural numbers, including the zero number, will be denoted by $\mathbb{N}$. By $\mathbb{N}^\ast$ we will denote the subset of $\mathbb{N}$ consisting of all non-zero natural numbers.
By $\Di$, $i=1,\ldots, n$, we will denote the weak, that is the distributional, derivatives in open subsets of $\mathbb{R}^n$. Some authors prefer to use the notation $\frac{\partial  }{\partial  x_i}$. When $n=1$, the weak derivative $\D_1$ will be denoted by $\D$. \dela{some authors prefer to use notation $\Djeden$.} Next, for a multiindex $\alpha=(\alpha_1, \alpha_2,\ldots, \alpha_n)$ we 
set   $\vert \alpha \vert = \alpha_1+ \ldots +\alpha_n$ and  \[D^\alpha = D^{\alpha _1}_1\ldots D^{\alpha_n}_{n}.\]

\section{Preliminaries} \label{Section2}

Given an open set    $\mathscr{O}\subset \mathbb{R}^n$,  $k=0,1,\ldots$ and  $p \in (1,+\infty)$,  $W^{k,p} (\mathscr{O})$ is defined, see \cite[Section 3.2]{Adams+Fournier_2003} to be the space of all  functions $f \in \lpo $ such that the weak derivatives $D^\alpha f$, for $\vert  \alpha \vert \le k$, belong to $\lpo $ as well. It is known that $\wkp $ is a Banach space with the norm
\begin{equation}\label{E2.1}
\vert f \vert _{k,p} := \left\{ \sum_{\vert \alpha \vert \le k} \int_\mathscr{O} \left\vert D^\alpha  f(x)\right\vert^p \d x \right\} ^{1 \over p},
\end{equation}
where as usual we put $D^0f:=f$. Let us also recall, see \cite[Section 3.2]{Adams+Fournier_2003}, that $W^{k,p} (\mathscr{O})$ can be  defined equivalently as the completion of the space $\mathscr{C}^{k,p}:= \{f \in  \mathscr{C}^k(\mathscr{O})\colon  \vert f \vert _{k,p}<+\infty\}$,  see \cite{Meyers+Serrin_1964} and/or \cite[Theorem 3.17]{Adams+Fournier_2003}. Thus in particular, $\mathscr{C}^{k,p}$ is dense in   $W^{k,p} (\mathscr{O})$.

\begin{remark}\label{rem-czy to potrzebne}
Assume that  the boundary $\partial  \mathscr{O}$  is of a class $\mathscr{C}^k$ and  $\vert\alpha\vert <k-\frac 1p$.  Let  $B^{r,p}_q (\partial \mathscr{O} )$ be the Besov space, see \cite[Definition 3.6.1]{Triebel}. It is well known, see e.g. \cite[Theorem 4.7.1]{Triebel} or \cite{Agmon}, that   there exists a unique linear bounded mapping
\begin{equation}
\label{E2.2}
\gamma_\alpha \colon \wkp  \mapsto B^{k-{1 \over p} -\vert \alpha \vert , p}_p(\partial \mathscr{O} ) ,
\end{equation}
such that
\begin{equation}
\label{eqn-gamma}
\gamma_\alpha (u)= D^\alpha u \vert _{\partial \mathscr{O}}, \quad u\in \mathscr{ C}^k( \overline{\mathscr{O}}).
\end{equation}
One first defines a map $\newtilde{\gamma}_\alpha \colon \mathscr{ C}^k( \overline{\mathscr{O}}) \cap \mathscr{C}^{k,p} \ni u  \mapsto D^\alpha u \vert _{\partial \mathscr{O} }  \in \mathscr{ C}^k (\partial  \mathscr{O})$ and proves that  it is continuous  in the  norms  \eqref{E2.1}. Making  use of  the density  of $\mathscr{ C}^k( \overline{\mathscr{O}})\cap \mathscr{C}^{k,p}$ in the space $\wkp$ , see \cite[Theorem 3.2.5]{Triebel}, one proves that $\newtilde{\gamma}_\alpha$ has exactly one continuous extension which  is the  map $\gamma_\alpha$  from \eqref{E2.2}.  For this and related results we suggest  to see  \cite[Sections 2.9.1 and 4.7.1 ]{Triebel} or \cite{Agmon}. However, if the  boundary $\partial  \mathscr{O} $  is not  regular enough, the situation becomes more complicated. It is our aim to consider here an example of a domain $\mathscr{O} $ with irregular boundary, in a sense made precise below.
\end{remark}

Let    $\mathscr{ D} (\mathscr{O})$ be   the space of test functions, i.e. $ \mathscr{ D} (\mathscr{O})=\mathscr{C}_0^\infty(\mathscr{O})$ endowed with the locally convex topology as defined in \cite[Definition 6.3]{Rudin-FA_1973}. By  $\mathscr{ D}^\prime(\mathscr{O})$ we denote the space of distributions on  $\mathscr{O}$ in the sense of \cite[Chapter 6]{Rudin-FA_1973}. Let $B$  be a differential operator of order $\le m$ with $\mathscr{C}^\infty(\mathscr{O})$  coefficients,  i.e.
\begin{equation}\label{eqn-def-B}
 Bu =\sum_{\alpha\colon \vert \alpha \vert\le m} a_\alpha (\cdot) D^\alpha u \in  \mathscr{  D}^\prime(\mathscr{O}), \quad  u  \in  \mathscr{  D}^\prime(\mathscr{O}).
\end{equation}
In other words, $B$ is a linear operator in  the space $\mathscr{  D}^\prime(\mathscr{O})$ of distributions  on $\mathscr{O}$,  defined by
\begin{equation}\label{eqn-def-B-2}
\bigl( Bu, \varphi  \bigr):= \bigl( u, B^\ast \varphi  \bigr), \qquad  \varphi   \in  \mathscr{  D}(\mathscr{O}),
\end{equation}
where  $B^\ast$ is the formally adjoint operator of $B$   defined by
\begin{equation}\label{eqn-def-B-adjoint}
\begin{split}
 B^\ast \varphi  &=\sum_{\alpha\colon \vert \alpha \vert \le m} (-1)^{\vert \alpha \vert} D^\alpha \bigl(  a_\alpha (\cdot)  \varphi  \bigr)
 \\
 &=\sum_{\alpha\colon \vert \alpha \vert \le m} \sum_{\beta \leq \alpha} (-1)^{\vert \alpha \vert}\, \binom{\alpha}{\beta} \, D^{\alpha-\beta}  a_\alpha (\cdot)  D^{\beta} \varphi \in  \mathscr{  D}(\mathscr{O}), \quad \varphi   \in  \mathscr{  D}(\mathscr{O}).
 \end{split}
\end{equation}
We will say that $Bu\in\lpo$ if and only if there exists $v\in\lpo$ such that
\begin{equation}\label{eqn-def-B-2}
\bigl( v, \varphi  \bigr):= \bigl( u, B^\ast \varphi  \bigr), \qquad  \varphi   \in  \mathscr{  D}(\mathscr{O}).
\end{equation}

The next definition is a basis for our whole paper.
\begin{definition}\label{def-W^B,p}
Assume that $\mathcal{B}=\{ B_1, \ldots, B_N\}$ is a finite set of  differential operators of the form \eqref{eqn-def-B}. Let us  define   $W^{\mathcal{B},p}(\mathscr{O})$  to be the space of all $u\in L^p(\mathscr{O})$ such that $B_iu\in\lpo$, $i=1,\cdots,N$,  i.e.
\begin{equation}\label{eqn-W^B,p}
W^{\mathcal{B},p}(\mathscr{O}):= \bigl\{  u\in L^p(\mathscr{O})\colon  B_iu\in\lpo, \; i=1,\cdots,N \bigr\}.
\end{equation}
The space $W^{\mathcal{B},p}(\mathscr{O})$ is endowed with the following  norm
\begin{equation}\label{eqn-W^B,p-norm}
\Vert u \Vert _{\mathcal{B},p} := \left\{ \vert u \vert _{\lpo} ^p + \sum_{i=1}^N \vert B_iu \vert _{\lpo} ^p \right\} ^{1 \over p}.
\end{equation}
When the family $\mathcal{B}=\{ B\}$ is a singleton, we will simply use a simpler notation  $W^{B,p}(\mathscr{O})$. 
\end{definition}
\begin{example}\label{example-W^kp}
The space $W^{k,p}(\mathscr{O})$, where $k \in \mathbb{N}$, is an example of the space $W^{\mathcal{B},p}(\mathscr{O})$, with $\mathcal{B}=\{  D^\alpha \colon\alpha \in  \mathbb{N}^n\colon \vert \alpha \vert\le k \}$. The space $W^{\Delta,p}(\mathscr{O}):=W^{\{\Delta\},p}(\mathscr{O})$,  is another example of the space $W^{\mathcal{B},p}(\mathscr{O})$.
\end{example}

 It is standard to prove that $W^{\mathcal{B},p}(\mathscr{O})$ is a  Banach space. Moreover, if $\mathcal{B} \subset \mathcal{B}^\prime$, then
 $$
 W^{\mathcal{B}^\prime,p}(\mathscr{O}) \embed W^{\mathcal{B},p}(\mathscr{O})
 $$
 continuously.  It may happen that  $W^{\mathcal{B}^\prime,p}(\mathscr{O}) = W^{\mathcal{B},p}(\mathscr{O})$ with equivalent norms.

In many cases,   for  example if $B$  is a uniformly elliptic operator of  order $m=2l$ and $\mathscr{O}$ is  a bounded domain with a $\mathscr{C}^\infty$ boundary  or if  $\mathscr{O} =\mathbb{ R}^n$,  see \cite{Triebel},  the space   $W^{B,p}(\mathscr{O}):= W^{\{B\},p}(\mathscr{O})$  coincides  with  the classical Sobolev space $W^{m,p}  (\mathscr{O})$ and the norms are equivalent. For example,   we have
\begin{equation}\label{eqn-W^DElat,p=W^2,p}
W^{\Delta,p}(\mathbb{R}^n) = W^{2,p}(\mathbb{R}^n) \qquad \mbox{ (with equivalent norms)}.
\end{equation}

The last assertion does not take place in case of $\mathscr{O} =\mathbb{ R} ^n \setminus \{ a \}  $ for a fixed point $a \in \mathbb{ R}  ^n$. We will come back to this question later on. For simplicity of exposition we always take  $a=0$ below.  Before we state our first result, let us introduce some notation.

\begin{definition}\label{D2.1} Assume that $p \in [1,+\infty)$ and $\mathscr{O}$ is an open subset of $\mathbb{ R}^n$. If $u \in L^{p} (\mathscr{O}) $ then by  $\newtilde{u}$ we denote the equivalence class in $L^p(\mathbb{R}^n)$ of  a function  that is equal to $0$ outside  $\mathscr{O}$ and equal to $u$ in $\mathscr{O}$. Thus, with  a slight abuse of notation,
\begin{equation*}
\newtilde{u}(x)=\begin{cases} u(x), \quad  &\text{if $x \notin  \mathscr{O}$}, \\
0, \quad  &\text{if $x  \in \mathscr{O}$}. \end{cases}
\end{equation*}
\end{definition}
Let us observe that if $\mathscr{O}=\mathbb{R}^n\setminus\{0\}$, or generally if the closed set   $\mathbb{R}^n\setminus \mathscr{O}$ is of Lebesgue measure $0$, then  the map
$$
\texttildelow \colon  L^{p} (\mathscr{O}) \ni u \mapsto   \newtilde{u} \in  L^p(\mathbb{R}^n)
$$
is an isometric isomorphism.

In what follows, the Laplace  operator acting on the space of distributions $\mathscr{ D}^\prime (\mathbb{ R}^n)$ will be denoted by  $\Delta_{\mathbb{ R}^n}$ or simply by $\Delta$, i.e.
$$
\Delta_{\mathbb{ R}^n}=\Delta: \mathscr{ D}^\prime (\mathbb{ R}^n) \to \mathscr{ D}^\prime (\mathbb{ R}^n).
$$

\begin{definition}\label{D2.2}
Let  $T$ be a distribution on $\mathbb{ R}^n$, i.e.  $T \in \mathscr{ D}^\prime (\mathbb{ R}^n)$ and  $\mathscr{O}$ be an open subset of $\mathbb{ R}^n$.  By  $T_\mathscr{O}$ we denote  the distribution on $\mathscr{O}$ which is the  restriction of $T$  to $\mathscr{D}^\prime(\mathscr{O})$,  i.e. $T_\mathscr{O}\in  \mathscr{D}^\prime(\mathscr{O}) $ is defined by
\begin{align*}
 \left(T_\mathscr{O}, \varphi \right)&= \left( T, \newtilde{\varphi} \right) \mbox{ for  $\varphi\in \mathscr{ D} (\mathscr{O})$}.
\end{align*}
\end{definition}

It is obvious that the following restriction map
\begin{equation}\label{eqn-pi_O-Delta}
\pi \colon  W^{\Delta,p}(\mathbb{R}^n) = W^{2,p}(\mathbb{R}^n) \ni u \mapsto u|_{\mathbb{R}^n\setminus\{0\}} \in \wzp \embed \wdp
\end{equation}
is well defined, injective and continuous.

Let us point out that if  $i_\mathscr{O}\colon \mathscr{ D} (\mathscr{O}) \embed  \mathscr{ D} (\mathbb{ R}^n)$ is the natural embedding and  $i_\mathscr{O}^\prime \colon \mathscr{ D}^\prime (\mathbb{ R}^n) \to \mathscr{ D}^\prime (\mathscr{O}) $  is the adjoint map  defined by
 $$
\bigl(i_\mathscr{O}^\prime(T), \varphi\bigr):= \bigl(T, i_\mathscr{O}(\varphi)\bigr) , \quad \varphi \in \mathscr{ D} (\mathscr{O}),\; T\in \mathscr{ D}^\prime (\mathbb{R}^n),
 $$
 then $T_{\mathscr{O}}=i_\mathscr{O}^\prime (T)$.

Note that for $T\in \mathscr{ D}^\prime (\mathbb{R}^n)$,
$$
\Delta_{\mathscr{O}}(i_\mathscr{O}^\prime(T)) =i_\mathscr{O}^\prime(\Delta T).
$$
Note also that we also have the following trivial identity
\[
\Delta(i_\mathscr{O}(\varphi)) =i_\mathscr{O}(\Delta_\mathscr{O}  \varphi),\qquad  \varphi \in \mathscr{ D} (\mathscr{O}).
\]

Applying the above result to the particular case of $T=\newtilde{u}$ with  $u\in W^{\Delta ,p} (\mathscr{O})$,  we obtain
\begin{equation}\label{E2.3}
\Delta_{\mathscr{O}}u =(\Delta \newtilde{u})\vert_{\mathscr{O}},
\end{equation}
 i.e.
\begin{equation*}
 (\Delta_{\mathscr{O}} u,\varphi)= ((\Delta \newtilde{u} )\vert_{\mathscr{O}}, \varphi), \quad \varphi \in  \mathscr{ C}_0^\infty (\mathscr{O} ).
\end{equation*}
Indeed,   $\newtilde{u}\vert_{\mathscr{O}}=u$ in the $\mathscr{D}^\prime(\mathscr{O}) $ sense.

\begin{convention}\label{conv-01}
 If not  explicitly stated otherwise,  all  the functional  spaces treated  in this  paper are assumed  to consist  of complex  valued functions. By $\delta_0$ we  denote the Dirac delta distribution at the point $0 \in \mathbb{ R}^n$. By $\mathscr{ L}(X)$  we denote the space of  all bounded and linear operators from a Banach space $X$ into itself.
\end{convention}

\section{The Sobolev space $W^{\Delta,p}(\mathbb{R}^n\setminus\{0\})$} \label{Section3}

From now on we will consider the  case of  $\mathscr{O}=\mathbb{ R}^n \setminus \{ 0\}$. Let us recall that the spaces $\wdp$ and $\wzp$ have been defined in Definition \ref{def-W^B,p} and Example \ref{example-W^kp}.

We begin with the following two "local" properties  of the space $W^{\Delta,p}(\mathbb{R}^n\setminus\{0\})$.
\begin{proposition}\label{Prop-new-localization}
Let $\tilde{\mathscr{O}}$ be an open subset of $\mathbb{R}^n$ such that its closure $\overline{\tilde{\mathscr{O}}}\subset \mathbb{R}^n\setminus\{0\}$, and let  $u\in W^{\Delta,p}(\mathbb{R}^n\setminus\{0\})$. Then
$$
u \vert_{\tilde{\mathscr{O}}}\in W^{2,p}\left(\tilde{\mathscr{O}}\right).
$$
\end{proposition}
\begin{proof} Set
$$
B(0,r):=\{ x\in \mathbb{R}^n\colon \vert x\vert \le r\}.
$$
Let us denote by $B^\mathrm{c}(0,r)$ the complement of $B(0,r)$ in $\mathbb{R}^n$.  Then there is an $r>0$ such that  $\tilde{\mathscr{O}}\subset B^\mathrm{c}(0,r)$.  Then, since $\partial B^\mathrm{c}(0,r)$ is regular, we have $W^{\Delta,p}\left(B^\mathrm{c}(0,r)\right)= W^{2,p}\left(B^\mathrm{c}(0,r)\right)$, and the desired conclusion follows as
$$
v\in W^{2,p}\left(B^\mathrm{c}(0,r)\right)\Longrightarrow v\vert _{\tilde{\mathscr{O}}}\in W^{2,p}\left(\tilde{\mathscr{O}}\right).
$$
\end{proof}

\begin{proposition}
\label{prop-localization}
Assume that $\phi \in \mathscr{C}_0^\infty(\mathscr{O})$  is such that $\phi=1$ in a neighbourhood of the origin $0$. Assume that $u \in L^{p}(\mathbb{R}^n\setminus\{0\})$. Then the following two conditions are equivalent:
\begin{trivlist}
\item[(i)] $u \in W^{\Delta,p}(\mathbb{R}^n\setminus\{0\})$,
\item[(ii)] $\phi u \in W^{\Delta,p}(\mathbb{R}^n\setminus\{0\})$ and  $(1-\phi) \newtilde{u} \in W^{2,p}(\mathbb{R}^n)$.
\end{trivlist}
\end{proposition}
\begin{proof} The implication (ii) $\implies$ (i) is obvious, since the second condition in (ii) implies that the restriction of the function   $(1-\phi) \newtilde{u}$ to the open set $\mathbb{R}^n\setminus\{0\}$ equals to  $(1-\phi) u$ and thus belongs $W^{\Delta,p}(\mathbb{R}^n\setminus\{0\})$ and $W^{\Delta,p}(\mathbb{R}^n\setminus\{0\})$ is a vector space.

To prove the converse implication we assume that $u \in W^{\Delta,p}(\mathbb{R}^n\setminus\{0\})$. First, we will  show  that $(1-\phi) \newtilde u \in W^{2,p}(\mathbb{R}^n)= W^{\Delta,p}(\mathbb{R}^n)$. Since obviously $(1-\phi) \newtilde u \in L^{p}(\mathbb{R}^n)$, in view of \cite[Proposition III.1.3]{Stein_1970},  it is sufficient to prove that $\Delta\bigl( (1-\phi) \newtilde u  \bigr)\in L^{p}(\mathbb{R}^n)$.  It is easy to see that for any $w\in W^{2,p}(B^\mathrm{c}(0,r))$ and any $R>r$ there is a  $\newtilde w\in W^{2,p}(\mathbb{R}^n)$ such that $\newtilde w\vert_{B^\mathrm{c}(0,R)}= w\vert _{B^\mathrm{c}(0,R)}$.  Therefore there is a $v\in W^{2,p}(\mathbb{R}^n)$ such that
\begin{equation}\label{Exs}
(1-\phi)\newtilde u= (1-\phi)v.
\end{equation}
Then
\begin{align}\label{eqn-Delta of product}
\Delta\bigl( (1-\phi) \newtilde u  \bigr)=\Delta \bigl((1-\phi)v \bigr) =  v\Delta (1-\phi)  +(1-\phi) \Delta v +2\sum_{i=1}^n  D_i(1-\phi) D_iv
\end{align}
in the weak sense, and all terms  on the RHS of \eqref{eqn-Delta of product} belong to $L^{p}(\mathbb{R}^n)$.

Proposition \ref{Prop-new-localization} can be used in the proof that $\phi u\in W^{\Delta,p}(\mathbb{R}^n\setminus\{0\})$.  For,  there is an $r>0$ such that $\phi\vert _{B(0,r)}$ is constant and $u\vert _{B^\mathrm{c}(0,r)}\in W^{2,p}(B^\mathrm{c}(0,r))$. Then by a simple continuity argument
\begin{align*}
\Delta\vert_{\mathbb{R}^n\setminus \{0\}}\bigl( \phi u  \bigr)
&=  u\Delta \vert_{\mathbb{R}^n\setminus \{0\}}\phi  +\phi \Delta \vert_{\mathbb{R}^n\setminus \{0\}} u \\ &+2\times \mathds{1}_{B^\mathrm{c}(0,r)} \left(\sum_{i=1}^n   D_i\vert _{B^\mathrm{c}(0,r)} \phi\vert_{B^\mathrm{c}(0,r)}  D_i\vert_{B^\mathrm{c}(0,r)} u\vert_{B^\mathrm{c}(0,r)} \right)^{\sim} ,
\end{align*}
where $\mathds{1}_\Gamma$ denotes the characteristic function of a set $\Gamma$ and $f^{\sim}$ is an arbitrary extension to $\mathbb{R}^n\setminus\{0\}$ of a function $f$ defined on $B^\mathrm{c}(0,r)$.
\end{proof}

In what follows by $W^{-2,p}(\mathbb{R}^n)$ we denote the space
$$
W^{-2,p}(\mathbb{R}^n)=\left( W^{2,p^*}(\mathbb{R}^n)\right)^\prime,
$$
where $p^*$ is such that  $1/p+ 1/p^*=1$.  Note that for any $\psi \in L^p(\mathbb{R}^n)$, $\Delta \psi \in W^{-2,p}(\mathbb{R}^n)$.  Finally,  we understand that
\begin{align}
\frac{n}{n-2}=+\infty &\mbox{ if } n=2 \mbox{ and }  \frac{n}{n-1}=+\infty \mbox{ if } n=1.
\end{align}

\begin{lemma}\label{Lem}
Assume that  $n, N \in \mathbb{N}$ and    $p \in (1,+\infty )$.  Let
\begin{equation}\label{EdE}
\omega= \sum_{\vert \alpha\vert \le N} c_\alpha D^\alpha \delta_0.
\end{equation}
We have:
\begin{itemize}
\item[(i)]  if $p\ge n/(n-2)$, then $w\in W^{-2,p}(\mathbb{R}^n)$ if and only if  $c_\alpha =0$ for each $\alpha$,
\item[(ii)] if $n/(n-1)\le p< n/(n-2)$, then    $\omega \in W^{-2,p}(\mathbb{R}^n)$ if and only if  $c_\alpha =0$ for each $\alpha$ such that $\vert \alpha \vert \ge 1$,
\item[(iii)] if $1\le p< n/(n-1)$, then    $\omega \in W^{-2,p}(\mathbb{R}^n)$ if and only if  $c_\alpha =0$ for each $\alpha$ such that $\vert \alpha \vert \ge 2$.
\end{itemize}
\end{lemma}
\begin{proof} Let us recall the following Sobolev embeddings: $W^{2,q}(\mathbb{R}^n)\hookrightarrow \mathscr{C}^1(\mathbb{R}^n)$ if and only if $q>n$, $W^{2,q}(\mathbb{R}^n)\hookrightarrow \mathscr{C}(\mathbb{R}^n)$ if and only if $q>n/2$, and for $k>1$, there are no $q$ such that $W^{2,q}(\mathbb{R}^n)\hookrightarrow \mathscr{C}^k(\mathbb{R}^n)$. Therefore, if $D^\alpha \delta_0\in W^{-2,p}(\mathbb{R}^n)$ then  $\vert \alpha\vert \le 1$, and  $D^\alpha \delta_0\in W^{-2,p}(\mathbb{R}^n)$ for $\alpha \colon \vert \alpha \vert = 1$  if and only if $ p< n/(n-1)$. Finally $\delta_0\in W^{-2,p}(\mathbb{R}^n)$ if and only if $p< n/(n-2)$.

Therefore, the problem is only to show that if $\omega$ of the form \eqref{EdE} belongs to $W^{-2,p}(\mathbb{R}^n)$ then necessarily each component $c_\alpha D^\alpha \delta_0\in W^{-2,p}(\mathbb{R}^n)$.  To do this consider  $\phi\in \mathscr{C}_0^\infty(\mathbb{R}^n)$  such that $\phi (x)= 1$ a neighbourhood of  $0$.  Set $x^\alpha = x^{\alpha _1}_1\ldots x^{\alpha_n}_n$ and $\phi_\alpha(x)=x^\alpha \phi(x)$. Next, let $\alpha^{k}= (\alpha_1, \ldots, \alpha _k-1, \ldots \alpha_n)$ if $\alpha_k>0$.   Then for any $\psi \in \mathscr{C}_0^\infty(\mathbb{R}^n)$,
\begin{align*}
D^\alpha  \left(\phi_\alpha (x)\psi(x)\right)\vert_{x=0}&=\alpha_1! \cdots \alpha _n!\psi(0),\\
D^\alpha \left( \phi_{\alpha ^k}(x)\psi(x)\right)\vert_{x=0}&= \alpha_1!\cdots (\alpha_k-1)!\ldots \alpha_n D_k\psi(0) \qquad \text{if $\alpha_k>0$.}
\end{align*}
Let $\omega$ be of the form \eqref{EdE}, and let $\mathcal{I} = \left\{ \alpha\colon c_\alpha \not =0\right\}$.  We have to show that $\mathcal{I}=\emptyset$ if $p\ge n/(n-2)$. Assume that $\mathcal{I}\not=\emptyset$ and let $\alpha \in \mathcal{I}$ be such that $\vert\alpha\vert \ge \vert \beta\vert$ for any $\beta \in \mathcal{I}$.  Then the multiplication operator $\psi \mapsto \phi_\alpha \psi $ is continuous on $W^{2,p^*}(\mathbb{R}^n)$ and
$$
\omega\left( \phi_\alpha \psi\right)= c_\alpha \alpha_1!\ldots\alpha_n! \delta _0 \psi.
$$
which leads to contradiction as $\delta_0\not \in \left( W^{2,p^*}(\mathbb{R}^n)\right)^\prime= W^{-2,p}(\mathbb{R}^n)$. In the same way we can tread the case of $p\in [n/(n-1),n/(n-2)]$. We need to show that $\mathcal{I}\cap\{\alpha\colon \vert \alpha \vert \ge 1\}=\emptyset$. Assume that $\alpha\in \mathcal{I}$  is such that $\vert \alpha \vert \ge \vert \beta\vert$ for any $\beta \in \mathcal{I}$ and that $\vert \alpha \vert \ge 2$. Let $\alpha _k \ge 1$. Then there is a constant $c$ such that
$$
\omega\left(\phi_{\alpha^k}\psi\right)=  \alpha_1!\ldots (\alpha_k-1)!\ldots \alpha_n D_k\psi(0) + c\delta_0(\psi).
$$
Since $\delta_0\in W^{-2,p}(\mathbb{R}^n)$ this contradicts the fact that $D_k\delta_0\not \in W^{-2,p}(\mathbb{R}^n)$.

It remains to examine the case of $p<n/(n-1)$. In thus case we have to show that $\mathcal{I}\cap\left\{\alpha \colon \vert \alpha \vert \ge 2\right\}=\emptyset$. It is easy to show that $\mathcal{I}\cap\left\{\alpha \colon \vert \alpha \vert \ge 3\right\}=\emptyset$. For one can use the above arguments with function $\phi_{(\alpha^k)^j}$. It is more difficult to show that  $\mathcal{I}\cap\left\{\alpha \colon \vert \alpha \vert = 2\right\}=\emptyset$. Since $\delta_0$ and $D_k \delta_0$ are from $W^{-2,p}(\mathbb{R}^n)$ we can assume that
$$
\omega = \sum_{\alpha \colon \vert \alpha\vert =2}  c_\alpha D^\alpha \delta_0
$$
Changing variables we can assume that
$$
\omega= \sum_{j=1}^n a_j D^2_j \delta_0\not =0.
$$
We can assume that $a_1\not =0$. Let $\varphi \in \mathscr{C}_0^\infty(\mathbb{R}^{n-1})$  be such that each $D^2_k  \phi (0)= 0$ and $\phi(x)=1$  a neighbourhood of  $0$. Consider the following linear continuous mapping $\Phi\colon W^{2,p}(\mathbb{R})\mapsto W^{2,p}(\mathbb{R}^n)$;
$$
\Phi(\psi)= \psi\otimes \varphi, \qquad \psi\otimes\varphi(x_1,x_2,\ldots,x_n)=\psi(x_1)\varphi(x_2,\ldots,x_n).
$$
Then
$$
\omega(\Phi(\psi))= a_1 D^2 \delta _0(\psi),
$$
which leads to contradiction.
\end{proof}
The following result presents a characterization of the Sobolev spaces $W^{\Delta,p} (\mathbb{ R}^n \setminus \{ 0\}) $ depending on $p$ and $n$. Here by $\Delta$ we understand the map $\Delta_{\mathbb{ R}^n}$.

\begin{lemma} \label{lem-3.3}
Assume that  $n \in \mathbb{N}$ and    $p \in (1,+\infty )$. If  $u \in W^{\Delta,p} (\mathbb{ R}^n \setminus \{ 0\}) \subset L^{p} (\mathbb{ R}^n \setminus \{ 0\})$, $v:={\Delta_{\mathbb{ R}^n \setminus \{ 0\}}  u} \in L^p(\mathbb{ R}^n \setminus \{ 0\})$    and, see Definition \ref{D2.1} above, $\newtilde{u} \in L^p(\mathbb{ R}^n)$ and
$$
\newtilde{v} := \newtilde{\Delta_{\mathbb{ R}^n \setminus \{ 0\}}  u}\in L^p(\mathbb{R}^n),
$$
then there exists a finite set of   numbers $c_\alpha\in \mathbb{C}$ such that
 \begin{equation}
\label{E3.1}
-\Delta \newtilde{u} + \newtilde{v} =\begin{cases}
\displaystyle
c_0\delta _0 +\sum_{j=1}^n c_j D_j \delta _0, & \text{if  $1< p < \frac{n}{n-1}$,} \\
c_0\delta _0,       &   \text{if  $\frac{n}{n-1}\le p < \frac{n}{n-2}$,}  \\
0, & \text{if $p\geq \frac{n}{n-2}$.}
\end{cases}
\end{equation}
\end{lemma}
\begin{proof} The proof of this lemma will be preceded by a somehow more general result. Namely,  assume that $p \in [1,+\infty)$ and $\mathscr{O}$ is an open subset of $\mathbb{ R}^n$. If $u \in W^{\Delta,p} (\mathscr{O})$ and  $v:=\Delta_{\mathscr{O}}u\in L^{p} (\mathscr{O}) $, then
\begin{equation}\label{E3.2}
\supp (-\DeltaRn \newtilde{u}+ \newtilde{v})  \subset \partial(\mathscr{O}).
\end{equation}
To do this  let us point out that since $v \in L^{p} (\mathscr{O})$, the extension  $\newtilde{v}$ is a well defined element of  $L^p(\mathbb{ R}^n)$. Similarly $\newtilde{u} \in L^p(\mathbb{ R}^n)$ and so $\DeltaRn \newtilde{u} \in  W^{-2,p} (\mathbb{ R}^n )$. Hence also $-\DeltaRn \newtilde{u}+ \newtilde{v} \in  W^{-2,p} (\mathbb{ R}^n ) \subset \mathscr{ D}^\prime (\mathscr{O})$. To prove \eqref{E3.2}, let us first consider an arbitrary $\varphi \in \mathscr{C}_0^\infty(\mathbb{R}^n)$ with support in $\mathscr{O}$.  Then by \eqref{E2.3} we have
\begin{equation*}
(- \DeltaRn  \newtilde{u} , \varphi ) + \left( \newtilde{\Delta_{\mathscr{O}}  u},\varphi \right)=0.
\end{equation*}
Next let us  consider an arbitrary $\varphi \in  \mathscr{ C} _0^\infty (\mathbb{R}^n)$ with support in $\mathbb{ R}^n \setminus \bar{\mathscr{O}}$. Then
$$
(\DeltaRn  \newtilde{u} , \varphi )= (\newtilde{u} , \DeltaRn \varphi )=0
$$
because $\supp (\DeltaRn \varphi) \subset \mathbb{ R}^n \setminus \bar{\mathscr{O}}$ and $\newtilde{u}=0$ on $\mathbb{ R}^n \setminus \bar{\mathscr{O}}$.
For the same reason,
\[
\bigl( \newtilde{\Delta_{\mathscr{O}} u} ,\varphi \bigr) =( \newtilde{v},\varphi)=0.
\]
Thus the proof of \eqref{E3.2} is complete.

To see \eqref{E3.1}, let us choose and fix a $u \in W^{\Delta,p} (\mathbb{ R}^n \setminus \{ 0\}) $. Then  $v:=\Delta_{\mathbb{ R}^n \setminus \{ 0\}}  u \in L^p(\mathbb{ R}^n \setminus \{ 0\})$, and  with our convention concerning the  map  $\texttildelow$, by \eqref{E3.2} we have
$$
\supp (-\DeltaRn \newtilde{u}+ \newtilde{v} ) \subset \{ 0 \}  .
$$
Then by  \cite[Theorem 6.25]{Rudin-FA_1973}, there exist an $N\in \mathbb{N}$  and a  set of complex numbers $c_\alpha$ with $\vert \alpha\vert \le N $  such that
$$
-\DeltaRn \newtilde{u}+ \newtilde{v}   =  \sum _{\alpha\colon \vert \alpha\vert \le N }c_\alpha D^\alpha \delta _0 .
$$
Since $\newtilde {u}\in L^p(\mathbb{R}^n)$,
$$
 \newtilde{v}- \sum_{\alpha\colon \vert \alpha \vert \le N} c_\alpha D^\alpha \delta_0\in W^{-2,p}(\mathbb{R}^n),
$$
and the desired conclusion follows from Lemma \ref{Lem}.
\end{proof}

\begin{lemma}\label{L3.2}
In the framework of Lemma \ref{lem-3.3}, the representation (\ref{E3.1}) is unique. In other words, for every  $u \in W^{\Delta,p}(\mathbb{R}^n\setminus\{0\})$ there exist exactly one $\newtilde{v} \in \lpr $ and, respectively on the parameters $n$ and $p$, exactly one system of constants $c_\alpha$ such that (\ref{E3.1}) holds.
\end{lemma}
\begin{proof}
Our claim is a direct consequence of  the following simple observation.  If $f \in \lpr $, $N  \in \mathbb{ N}$ and $c_\alpha \in \mathbb{C}$ for  $|\alpha |\le N$ are such that	
$$
f-\ds\sum_{|\alpha |\le N}  c_\alpha D^\alpha  \delta _0 =0  \quad \text{in}\quad   \mathscr{D}^\prime(\mathbb{R}^n),
$$
then $f=0$ and $c_\alpha =0$ for all $\alpha$.

\end{proof}

Define the  \emph{Bessel potential}  $G_n$ on $\mathbb{R}^n$ to be a function $\mathbb{R}^n \to [0,+\infty]$ given by
\begin{equation}\label{eqn-G_n}
\begin{aligned}
G_{n}(x)&:=\int_0^{+\infty} \e^{-\frac{\vert x\vert ^2}{4t}-t} \left(4\pi t\right)^{-\frac{n}{2}} \d t = (4\pi )^{-\frac{n}{2}} 2^{\frac{n}{2}}\vert x \vert^{1-\frac{n}{2}}   K_{\frac{n}{2}-1}( \vert x \vert)
\\
&= (2\pi )^{-\frac{n}{2}} \vert x \vert^{1-\frac{n}{2}}   K_{\frac{n}{2}-1}( \vert x \vert),
 \\&\dela{=\frac{1}{4\pi } \int_0^{+\infty} \e^{-\frac{\pi \vert x\vert ^2 }{t}- \frac{t}{4\pi}}\,t^{-\frac{n}{2}}\d t=\kappa_n(\vert x \vert), \qquad  x\in \mathbb{R}^n,}
\end{aligned}
\end{equation}
where  $ K_{\frac{n}{2}-1}$ is the Macdonald function, or modified Bessel function of the
third kind, see \cite[entries \textbf{8.432}.1 and \textbf{8.432}.6]{Grad+Ryzh_2007} and \cite[Theorem 7.2.1]{Moll_2016},  defined by
\begin{equation}\label{eqn-Macdonald function}
\begin{aligned}
K_\nu(z)&:=
 \int_{0}^{+\infty}
\e^{-z \cosh u } \cosh (\nu u) \, \d u
\\
&=\frac{1}{2} \bigl( \frac{z}{2}\bigr)^{\nu} \int_0^{+\infty} t^{-\nu-1} \, \e^{-t-\frac{z^2}{4t}}\; \d t, \;\;\;\;  \vert \arg z \vert < \frac{\pi}{2} \mbox { or } \Re z =0 \mbox{ and }\nu=0.
\end{aligned}
\end{equation}
\dela{and $\kappa_n\colon [0,+\infty) \to [0,+\infty]$ is  defined by
\begin{equation}\label{eqn-kappa_n}
\kappa_{n}(r):= \frac{1}{4\pi } \int_0^{+\infty} \e^{-\frac{\pi r^2 }{t}- \frac{t}{4\pi}}\,t^{-\frac{n}{2}}\d t,  \qquad  r \in [0,+\infty).
\end{equation}}
\dela{In particular,
\begin{equation}\label{eqn-kappa_2}
\kappa_{2}(r):= \frac{1}{4\pi } \int_0^{+\infty} \e^{-\frac{\pi r^2 }{t}- \frac{t}{4\pi}}\,t^{-1}\d t=
\frac{1}{2\pi }\frac12 \int_0^{+\infty} \e^{-r \cosh u}\,\d u
,  \qquad  r \in [0,+\infty).
\end{equation}}
Note that if $n \geq 2$, then $\lim_{x \to 0} G_n(x)=+\infty$. 

The functions $K_\nu$ for $\nu=-\frac12$ and $\nu=\frac12$ can be explicitly calculated.  We have
\begin{proposition}\label{prop-K_nu}
One has 
$$
  K_{-\frac12}(z) = \left(\frac{\pi}{2z}\right)^{\frac 12} \e^{-z} = K_{\frac12}(z).
$$
\end{proposition}
\begin{proof}
 Using the change of variables $t=\sigma^2$ in   the following identity, see  e.g. \cite{Yosida}, pp. 233--234,
\begin{equation}\label{eqn-Yosida identity}
\int_0^{+\infty} \e^{-(\sigma^2 + c\sigma^{-2})}\d \sigma = \frac{\sqrt{\pi}}{2}\e^{-2\sqrt{c}},\quad c>0,
\end{equation}
with $\sqrt{c}= \frac{z}{2}$,  we get
\begin{align}\label{eqn-K_-12-1}
  K_{-\frac12}(z) &=
  \frac{1}{2} \bigl( \frac{z}{2}\bigr)^{-\frac12} \int_0^{+\infty} t^{-\frac12} \, \e^{-t-\frac{z^2}{4t}}\; \d t
  \\
  &= 2 \times \frac{1}{2} \bigl( \frac{2}{z}\bigr)^{\frac12} \int_0^{+\infty}  \, \e^{-(\sigma^2+\frac{z^2}{4\sigma^2})}\; \d \sigma
  = \bigl( \frac{2}{z}\bigr)^{\frac12} \frac{\sqrt{\pi}}{2}\e^{-2\frac{z}{2}}=
   \bigl( \frac{\pi}{2z}\bigr)^{\frac12} \e^{-z}.
\end{align}
 Using the change of variables $s=t^{-\frac 12}$  in \eqref{eqn-Yosida identity}, we get
\begin{align}\label{eqn-K_12-1}
  K_{\frac12}(z) &=
  \frac{1}{2} \bigl( \frac{z}{2}\bigr)^{\frac12} \int_0^{+\infty} t^{-\frac32} \, \e^{-t-\frac{z^2}{4t}}\; \d t=
    2 \times \frac{1}{2} \bigl( \frac{z}{2}\bigr)^{\frac12} \int_0^{+\infty}  \, \e^{-(s^{-2}+s^2\frac{z^2}{4})}\; \d s
  \\&=  \bigl( \frac{z}{2}\bigr)^{\frac12} \frac{2}{z} \int_0^{+\infty}  \, \e^{-(\sigma^{2}+\sigma^{-2}\frac{z^2}{4})}\; \d\sigma
=\bigl( \frac{1}{2z}\bigr)^{\frac12} 2 \frac{\sqrt{\pi}}{2}\e^{-z}=\bigl( \frac{\pi}{2z}\bigr)^{\frac12}  \e^{-z}.
\end{align}
\end{proof}

Moreover, the function $G_n$ induces a unique  tempered distribution on $\mathbb{R}^n$ and,  as  such, it satisfies
\begin{equation}\label{eqn-G_n-2}
 G_n-\Delta G_n=\delta_0 \mbox{ in } \mathscr{D}^\prime(\mathbb{R}^n).
 \end{equation}
It follows from  \eqref{eqn-G_n-2} that $\Delta G_n=G_n$  in  $\mathscr{D}^\prime(\mathbb{R}^n\setminus\{0\})$. Therefore we have the following direct consequence of \eqref{eqn-G_n-2}.
\begin{corollary}\label{cor-G_n}
Assume that  $n \in \mathbb{N}$ and $p \in [1,+\infty)$. If $G_n$ denotes the restriction of the previous $G_n$ to the open set $\mathbb{R}^n\setminus\{0\}$, then the following implication holds
$$
G_{n} \in L^{p}(\mathbb{R}^n\setminus\{0\}) \implies  G_{n} \in  W^{\Delta,p}(\mathbb{R}^n\setminus\{0\}).
$$
Obviously, by  \eqref{eqn-W^B,p}, the converse implication holds as well.
\end{corollary}

Recall that  by $\Di=\frac{\partial }{\partial x_i}$ we mean the weak partial derivative. Finally, in the  four items below, by $G_n$ we understand  the restriction of the previous $G_n$ to the open set $\mathbb{R}^n\setminus\{0\}$.
\begin{lemma} \label{L3.3}
For any $n \in \mathbb{N}$,
\begin{equation}\label{E3.3}
G_n= \Delta_{\mathbb{R}^n\setminus\{0\}}G_{n} \qquad \text{pointwise on $\mathbb{R}^n\setminus\{0\}$ and in  $\mathscr{D}^\prime(\mathbb{R}^n\setminus\{0\})$.}
\end{equation}
Moreover, if  $p\in (1,+\infty)$, $i=1,\ldots n$ and  $j,k=1,\ldots n$, then
\begin{equation}\label{E3.4}
\begin{split}
G_{n} \in L^{p}(\mathbb{R}^n\setminus\{0\}) \;
&\mbox{ if and only if $n=1,2$ or  $n\ge 3$ and $p< \frac{n}{n-2}$ },
\end{split}
\end{equation}

\begin{equation}\label{E3.5}
\Di G_{n}\in W^{\Delta,p}(\mathbb{R}^n\setminus\{0\})\quad \text{ if and only if $n=1$ or  $n\ge 2$ and  $p< \frac{n}{n-1}$.}
\end{equation}

\begin{equation}\label{E3.5c}
\Djk G_{n}\in W^{\Delta,p}(\mathbb{R}^n\setminus\{0\})\quad \text{ if and only if $n=1$.}
\end{equation}
\end{lemma}
\begin{proof} By  definition $(I-\Delta)G_n$ is a distribution with the support $\{0\}$. On the other hand $(I-\Delta)G_n$ restricted to $\mathbb{R}^n\setminus\{0\}$ is a $\mathscr{C}^\infty$-function.  Therefore, $(I-\Delta)G_n(x)=0$ for $x\not =0$ and hence \eqref{E3.3} follows  as required.

Having shown \eqref{E3.3}, claim  \eqref{E3.4}  follow easily from the well known asymptotics of $G_{n}$. Indeed, for $n=1$, $G_1$ is bounded and continuous on $\mathbb{R}$, and for  $n \geq 1$, $G_n$ decays exponentially at $\infty$. Finally,   one has
\begin{align}
\label{eqn-G_n-asymptotics}
G_{n}(x)=
\begin{cases}
\frac{\Gamma\left(\frac{n-2}{2}\right)}{4 \pi\vert x\vert ^{n-2}}\left(1+\text{o}(1)\right)&\text{if $n \geq 3$,}\\
\frac{1}{2\pi}\log \frac{1}{\vert x\vert}\left(1+\text{o}(1)\right)&\text{if $n=2$,}
\end{cases} \mbox{ as $\vert x\vert \to 0$.}
\end{align}
\dela{i.e.
\begin{align*}
\label{eqn-kappa_n-asymptotics}
\kappa_{n}(r)=
\begin{cases}
\frac{\Gamma\left(\frac{n-2}{2}\right)}{4 \pi r^{n-2}}\left(1+\text{o}(1)\right)&\text{if $n \geq 3$,}\\
\frac{1}{2\pi}\log \frac{1}{r}\left(1+\text{o}(1)\right)&\text{if $n=2$,}
\end{cases} \mbox{ as $r\to 0$.}
\end{align*}
}

Since distributional operators $\Delta$ and $\Di$ commute and in view of the earlier exact formula for $G_n$,  to prove assertion \eqref{E3.5} it is enough to  show that  for $n\geq 2$ and $i=1,\ldots n$,
\begin{equation}\label{E3.5b}
\Di G_{n}\in L^{p}(\mathbb{R}^n\setminus\{0\})\quad \text{  if and only if   $p< \frac{n}{n-1}$.}
\end{equation}
From the definition of the function $G_n$ it follows  that
\begin{equation}\label{Pol}
\vert \nabla G_{n}(x)\vert =  \frac{\vert x \vert }{2\pi } \int_0^{+\infty} \e^{-\frac{\pi \vert x\vert ^2 }{t}- \frac{t}{4\pi}}t^{-\frac{n+2}{2}}\d t, \quad x \in \mathbb{R}^n,
\end{equation}
and the result follows by applying the asymptotic formula above once again.
\end{proof}

The following result  follows from identity \eqref{eqn-G_n-2} and Lemma \ref{L3.3}.
\begin{corollary}\label{cor-G_n}
If $n\in \mathbb{N}$ and $p \in (1,+\infty)$, then $G_n \notin W^{\Delta,p}(\mathbb{R}^n)$. Moreover, if $n \geq 2$ and $1<p<\frac{n}{n-2}$, then  the spaces $W^{\Delta,p}(\mathbb{R}^n\setminus\{0\})$ and $W^{2,p}(\mathbb{R}^n\setminus\{0\})$ do not coincide.
\end{corollary}
\begin{proof} $G_n\in W^{\Delta,p}(\mathbb{R}^n)$ means that $G_n, \Delta G_n \in L^p(\mathbb{R}^n)$. This contradicts equality \eqref{eqn-G_n-2}.  Assume that $n \geq 2$ and $1<p<\frac{n}{n-2}$. It is enough to observe that by  Lemma \ref{L3.3}, $G_n \in W^{\Delta,p}(\mathbb{R}^n\setminus\{0\})$, whereas  $G_n \not \in W^{2,p}(\mathbb{R}^n\setminus\{0\})$ by  the last part of the lemma.
\end{proof}

In the particular cases of $n=1$ or $n=3$, the potential $G_n$ is given in an explicit form, see below. Moreover,  all the calculations from the proof of Lemma \ref{L3.3} can be performed in a direct manner.
\begin{lemma}\label{L3.4}
We have
\begin{align}\label{eqn_G_1}
G_1(x)=\frac12 \e^{-\vert x\vert}, \qquad x\in \mathbb{R},
\end{align}
and
\begin{equation}\label{eqn_G_3}
G_3(x)= \frac{1}{4\pi \vert x\vert} \e^{-\vert x\vert},\qquad  x \in \mathbb{R}^3,
\end{equation}
$G_3$ is of $\mathscr{C}^\infty$-class on $\mathbb{R}^3\setminus\{0\}$,  and
\begin{equation}\label{eqn_G_3-b}
G_3(x)= \Delta_{\mathbb{R}^3\setminus\{0\}}G_{3}(x), \qquad   x \in \mathbb{R}^3\setminus\{0\}.
\end{equation}
Moreover, \\
(i) $G_3\in W^{\Delta,p}( \mathbb{R}^3\setminus\{0\})$, or equivalently $G_3\in L^p(\mathbb{R}^3)$,  if and only if $p\in [1,3)$,
 \\
 and \\
(ii)
$\Di G_3\in  W^{\Delta,p}( \mathbb{R}^3\setminus\{0\})$, or equivalently $\Di G_3\in  L^p( \mathbb{R}^3)$,  if and only if $p\in [1,\frac{3}{2})$. Consequently, if $p\in [\frac{3}{2},3)$, then $G_3\in W^{\Delta,p}( \mathbb{R}^3\setminus\{0\})$ but $G_3\not \in W^{1,p}( \mathbb{R}^3\setminus\{0\})$.
\end{lemma}
\begin{proof}
Using the second  identity of Proposition \ref{prop-K_nu} and \eqref{eqn-G_n} for $n=3$ we get \eqref{eqn_G_3}. Using the first identity of Proposition \ref{prop-K_nu} and \eqref{eqn-G_n} for $n=1$ we get \eqref{eqn_G_1}.

Note that
\begin{equation}\label{E3.7}
\begin{aligned}
\Di G_{3}\left(x\right) &=-\frac{x_i}{\vert x\vert ^2}  G_{3}\left( x\right)- \frac{x_i}{\vert x\vert} G_3\left(x \right)\\
&= -x_i \left[ \frac{1}{\vert x\vert ^2} + \frac{1}{\vert x\vert}\right] G_3\left(x \right).
\end{aligned}
\end{equation}
Therefore
$$
\Delta_{\mathbb{R}^3\setminus\{0\}}G_{3}(x)= G_3(x)\left[ -3\left( \frac{1}{\vert x\vert ^2} + \frac{1}{\vert x\vert}\right)+ \left( \frac{1}{\vert x\vert} +1\right)^2 + \frac{2}{\vert x\vert^2}+ \frac{1}{\vert x\vert} \right] =G_3(x),
$$
which gives the first claim. The remaining parts follow easily  from the first one and \eqref{eqn_G_3}\dela{{E3.6}} and \eqref{E3.7}.
\end{proof}
We finish this subsection with the following result used later. It can be proved by direct calculations if $n=3$, and, for $n=2,3$, by the Plancherel formula and direct calculations.
\begin{lemma}
\label{lem-L^2-norms of G_2 and G_3}
We have
\begin{align}\label{eqn-L^2-norms of G_2}
\Vert G_2 \Vert_{L^2}^2     & =\frac{1}{2}, \\
      \\
   \Vert G_3 \Vert_{L^2}^2     & =\frac{1}{8\pi}.
   \label{eqn-L^2-norms of G_3}
\end{align}
\end{lemma}

\subsection{Representation Theorem}\label{subsec-representation theorem}

Let us assume that  $u \in W^{\Delta,p} (\mathbb{R}^n\setminus\{0\}) $  and let us put
$$
v:=\Delta_{\mathbb{R}^n\setminus\{0\}} u \in L^p(\mathbb{R}^n\setminus\{0\}).
$$
Clearly, $\newtilde{u}-\newtilde{v}\in L^p(\mathbb{R}^n)$. Applying the operator $(1-\DeltaRn  )^{-1} $, which   is an isomorphism of the space $\mathscr{  S}^\prime  (\mathbb{ R}^n)$ of tempered distributions, to the following version of equality \eqref{E3.1},
$$
\newtilde{u}-\DeltaRn \newtilde{u}=c_0\delta _0 +\sum_{j=1}^n c_j D_j \delta _0 \dela{\sum_{\alpha\colon \vert \alpha \vert \le 1} c_\alpha D^\alpha \delta_0}+\newtilde{u}-\newtilde{v},
$$
we obtain
$$
\newtilde{u}=\dela{\sum_{\alpha \colon \vert \alpha \vert \le 1} c_\alpha D^\alpha  (1-\Delta )^{-1} \delta _0}
c_0(1-\Delta )^{-1} \delta _0 +\sum_{j=1}^n c_j D_j (1-\Delta )^{-1}\delta _0
 +f
$$
with $f:=(1-\Delta )^{-1} (\newtilde{u}-\newtilde{v}) $.

Since ${\newtilde{u}} - \newtilde{v} \in L^p(\mathbb{ R}^n)$, the lift property of  the Sobolev spaces, see \cite[Theorem 2.3.4]{Triebel}, yields that $f \in W^{2,p}(\mathbb{ R}^n)$. Moreover,  in the sense of tempered distributions,
$$
(1-\Delta )^{-1} \delta _0 = G_n.
$$

Clearly $\frac{p}{p-1}>n$ if and only if $p< \frac{n}{n-1}$, and $2\frac{p}{p-1}>n$ if and only if $p<\frac{n}{n-2}$. Therefore we have the following consequence of Lemmata  \ref{lem-3.3}, \ref{L3.2} and \ref{L3.3}.

\begin{theorem} \label{T3.10} Assume that $n\in \mathbb{N}$ and $ p \in (1,+\infty)$. \\
${\rm (i)}$ If  $n\ge 1$ and $1<p<\frac{n}{n-1}$,
then every  function $u \in W^{\Delta,p}(\mathbb{R}^n\setminus\{0\})$ can be uniquely written in the following form
\begin{equation}\label{eqn-representation-general}
u= \sum_{\alpha\colon \vert \alpha \vert \le 1}c_\alpha D^\alpha G_n +f,
\end{equation}
where $c_\alpha  \in \mathbb{ C}$ and $f \in W^{2,p} (\mathbb{R} ^n)$. Conversely, if $u \in \mathscr{S}^\prime(\mathbb{R}^n)$ has representation   \eqref{eqn-representation-general} with  $c_\alpha  \in \mathbb{ C}$ and $f \in W^{2,p} (\mathbb{R} ^n)$, then
$u \in W^{\Delta,p}(\mathbb{R}^n\setminus\{0\})$. Moreover, the mapping
$$
W^{\Delta,p}(\mathbb{R}^n\setminus\{0\}) \ni u\mapsto \left( (c_j)_{j=0}^{n}; f\right)\in \mathbb{C}^{1+n}\times W^{2,p}(\mathbb{R}^n)
$$
is an isomorphism (and  hence a bijection).
\\ \noindent
${\rm (ii)}$ If  $n\ge 2$ and $\frac{n}{n-1}\le p< \frac{n}{n-2}$, then every  function $u \in W^{\Delta,p}(\mathbb{R}^n\setminus\{0\})$ can be uniquely written in the following form
\begin{equation}\label{eqn-representation-general-middle range}
u= c_0  G_n +f,
\end{equation}
where $c_0\in \mathbb{ C}$ and $f \in W^{2,p} (\mathbb{R} ^n)$. Conversely, if $u \in \mathscr{S}^\prime(\mathbb{R}^n)$ has representation   \eqref{eqn-representation-general-middle range} with  $c_0  \in \mathbb{ C}$ and $f \in W^{2,p} (\mathbb{R} ^n)$, then
$u \in W^{\Delta,p}(\mathbb{R}^n\setminus\{0\})$. Moreover, the mapping
\begin{equation}\label{2.11}
W^{\Delta,p}(\mathbb{R}^n\setminus\{0\}) \ni u\mapsto \left( c_0, f\right)\in \mathbb{C}\times W^{2,p}(\mathbb{R}^n)
\end{equation}
is  an isomorphism.
\\
\noindent  ${\rm (iii)}$  If  $n \geq 3$ and $p\ge \frac{n}{n-2}$,  then  $u \in W^{\Delta,p}(\mathbb{R}^n\setminus \{0\})$ if and only if  $\newtilde{u} \in W^{2,p} (\mathbb{ R}^n)$, i.e.  the map
\begin{equation}\label{eqn-2.12}
W^{\Delta,p}(\mathbb{R}^n\setminus \{0\}) \ni u \mapsto \newtilde{u} \in  W^{2,p} (\mathbb{ R} ^n)
\end{equation}
is an isomorphism. If we replace the space $W^{2,p} (\mathbb{ R} ^n)$ on the RHS of \eqref{eqn-2.12} by  the space $W^{\Delta,p} (\mathbb{ R} ^n)$, the map above will be an isometric isomorphism.
\end{theorem}

In some sense the representation of $u\in W^{\Delta,p}(\mathbb{R}^n\setminus\{0\})$ does not depend on $p$. In fact we have the following result.
\begin{theorem}\label{Thre}
Assume that $p, \tilde p\in (1,+\infty)$. Let $u\in W^{\Delta,p}(\mathbb{R}^n\setminus \{0\})\cap  W^{\Delta,\tilde p}(\mathbb{R}^n\setminus \{0\})$.  Let
\begin{align*}
u&= c_0G_n+ \sum_{k=1}^n c_k D_kG_n + v,\\
u&= \tilde c_0G_n+ \sum_{k=1}^n \tilde c_k D_kG_n + \tilde v,
\end{align*}
where $v\in W^{2,p}(\mathbb{R}^n)$ and $\tilde v \in W^{2,\tilde p}(\mathbb{R}^n)$. Then
$$
c_0=\tilde c_0, \quad c_k=\tilde c_k, \qquad v=\tilde v.
$$
\end{theorem}
\begin{proof}
Assume that $\tilde p <p$. Let $B(0,R)$ be the ball in $\mathbb{R}^n$ with center at $0$ and radius  $R$ and let $B(0,R)^c$ be its complement.  Clearly $W^{2,p}(B(0,R))\hookrightarrow W^{2,\tilde p}(B(0,R))$ for any finite $R>0$. We have to show that $v\vert _{B(0,R)^\mathrm{c}} \in W^{2,\tilde p}(B(0,R)^\mathrm{c})$. Since $G_n$ decays at infinity, $v-\tilde v$ decays at infinity. Therefore $v\vert _{B(0,R)^\mathrm{c}} \in W^{2,\tilde p}(B(0,R)^\mathrm{c})$ as $\tilde v\vert _{B(0,R)^\mathrm{c}} \in W^{2,\tilde p}(B(0,R)^\mathrm{c})$.
\end{proof}

\begin{corollary}
Assume that $p\ge \frac{n}{n-2}$ and $\tilde p>1$. If $u\in W^{\Delta,p}(\mathbb{R}^n\setminus \{0\})\cap  W^{\Delta,\tilde p}(\mathbb{R}^n\setminus \{0\})$, then $\widetilde u\in W^{2,p}(\mathbb{R}^n)\cap  W^{2,\tilde p}(\mathbb{R}^n)$.
\end{corollary}

\begin{corollary}
\label{cor-3.4}
If $n>1$ and $1<p<\frac{n}{n-2}$, then, see Lemma \ref{L3.3}, neither $G_n$ nor $D_kG_n$ belong to $W^{2,p}(\mathbb{R}^n\setminus\{0\})$. Therefore,   if $u \in W^{\Delta,p}(\mathbb{R}^n\setminus\{0\})$ , then the following two conditions are equivalent:
 \begin{trivlist}
 \item[(i)] $u \in W^{2,p}(\mathbb{R}^n\setminus\{0\})$,  \item[(ii)] $\newtilde{u} \in W^{2,p}(\mathbb{R}^n)$.
  \end{trivlist}
\end{corollary}

Note that $2<\frac{n}{n-1}$  if and only if $n=1$, and, for $n> 2$, $2<\frac{n}{n-2}$  if and only if $n\le 3$. Therefore we have the following consequence of Theorem \ref{T3.10}.
\begin{corollary}\label{cor-representation}
${\rm (i)}$ If $p \in (1,+\infty)$, then every  function $u \in W^{\Delta,p}(\mathbb{R}\setminus\{0\})$ can be uniquely written in the following form
\begin{align}\label{eqn-representation-n=1}
u= c_0G_1 + c_1 \D G_1+ f,
\end{align}
where $c_0,c_1  \in \mathbb{ C}$ and $f \in W^{2,p} (\mathbb{R})$. Moreover, the mapping
$$
W^{\Delta,p}(\mathbb{R}\setminus\{0\}) \ni u\mapsto \left( c_0,c_1, f\right)\in \mathbb{C}^{2}\times W^{2,p}(\mathbb{R})
$$
is an isomorphism.
\\ \noindent
${\rm (ii)}$ Assume that $n=2,3$.   Then every  function $u \in W^{\Delta,2}(\mathbb{R}^n\setminus\{0\})$ can be uniquely written in the following form
$$
u= c_0  G_n +f,
$$
where $c_0\in \mathbb{ C}$ and $f \in W^{2,2} (\mathbb{R} ^n)$. Moreover, the mapping
$$
W^{\Delta,2}(\mathbb{R}^n\setminus\{0\}) \ni u\mapsto \left( c_0, f\right)\in \mathbb{C}\times W^{2,2}(\mathbb{R}^n)
$$
is  an isomorphism.
\\
\noindent  ${\rm (iii)}$  Assume that $n \geq 4$. Then  $W^{\Delta,2}(\mathbb{R}^n\setminus \{0\}) \equiv W^{2,2} (\mathbb{ R} ^n)$.
\end{corollary}

In the definitions below
\begin{equation}\label{eqn-n-p}
        \begin{tabular}
        { |c| |c | c | c| }
        \hline
         Case   & $n$ & $p$ & relevant formula \\
        \hline\hline
            & & & \\
         (i)   & $n\ge 1$ & $1<p<\frac{n}{n-1}$ & \eqref{eqn-representation-general} \\
           && & \\
            \hline  \hline
           & & &\\
          (ii)  & $n\ge 2$ & $\frac{n}{n-1}\le p< \frac{n}{n-2}$ &\eqref{eqn-representation-general-middle range} \\
           & &&  \\
            \hline
            &&&\\
          (iii)  & $n>2$ & $p \geq \frac{n}{n-2}$ &  \eqref{eqn-2.12}\\
           & & &\\
            \hline
            \end{tabular}
\end{equation}
The above results lead us to the following definitions.

\begin{definition}\label{def-Pi_p} Assume that $n \in \mathbb{N}$,  $p \in (1,+\infty)$.  Let  $u \in W^{\Delta,p}(\mathbb{R}^n\setminus \{0\})$. In cases (i) and (ii),  we denote  $\Pi_pu$ the unique function $f\in W^{2,p}(\mathbb{R}^n)$  from the appropriate  identity from Theorem \ref{T3.10}; either \eqref{eqn-representation-general} or \eqref{eqn-representation-general-middle range}. In case (iii), by  $\Pi_pu$ we denote the unique function $\newtilde{u} \in  W^{2,p} (\mathbb{ R} ^n)$  from \eqref{eqn-2.12}. We will call $\Pi_p$ the \emph{natural projection from $W^{\Delta,p}(\mathbb{R}^n\setminus \{0\})$ onto $W^{2,p}(\mathbb{R}^n)$}.
\end{definition}

\begin{definition}\label{def-xi_p} Assume that $n \in \mathbb{N}$,  $p \in (1,+\infty)$. In  case (ii) we define the following map
\begin{equation}\label{eqn-xi_0,p}
\xi_{p} \colon  W^{\Delta,p}(\mathbb{R}^n\setminus\{0\}) \to \mathbb{C},  \quad  \xi_{p}(u)=c_0 \qquad \text{if }\  c_0G_n=u-\Pi_p(u)\vert_{\mathbb{R}^n\setminus\{0\}}.
\end{equation}

In  case (i) we define the following map
\begin{equation}\label{eqn-xi_p}
\xi_{p}  \colon  W^{\Delta,p}(\mathbb{R}^n\setminus\{0\}) \to \mathbb{C}^{1+n},  \quad \xi_{p}(u)=(c_0,c_1,\ldots,c_n) \mbox{ if } c_0G_n+\sum_{j=1}^n c_j \Dj G_n=u-\Pi_p(u)\vert_{\mathbb{R}^n\setminus\{0\}}.
\end{equation}
Next let put  \[\xi_{p}(u)=(\xi_{0,p}(u),\xi_{1,p}(u),\cdots,\xi_{n,p}(u)).\]

In case (iii), $\xi_{p}$ is the map into the zero-dimensional space $\mathbb{C}^0:=\{0\}$, i.e.
$$
\xi_{p}  \colon  W^{\Delta,p}(\mathbb{R}^n\setminus\{0\})  \ni u \mapsto 0 \in \mathbb{C}^0.
$$
\end{definition}

The following straightforward result  follows from Theorem \ref{T3.10}.
\begin{proposition}\label{prop-Pi_p}
The function $\Pi_p$ is a linear and bounded map from  $W^{\Delta,p}(\mathbb{R}^n\setminus \{0\})$ to
$W^{2,p}(\mathbb{R}^n)$.  The function $\xi_{p}$ is a linear and bounded map on   $W^{\Delta,p}(\mathbb{R}^n\setminus \{0\})$. Moreover, for every $f \in \range(\Pi_p)$,
\begin{equation}\label{eqn-Range Pi_p} \Delta_{\mathbb{R}^n\setminus\{0\}} f=\Delta f.
\end{equation}
Finally, for every $u \in W^{\Delta,p}(\mathbb{R}^n\setminus \{0\})$,
\begin{equation}
\label{eqn-Pi_p-kernel}
u-\Pi_pu \in \mathscr{Y}_{n,p} := \begin{cases}
\textrm{linspan}\left\{G_n, \frac{\partial G_n}{\partial x_1}, \ldots , \frac{\partial G_n}{\partial x_n}\right\},
 &  \mbox{ if  $n\ge 1$ and $1<p<\frac{n}{n-1}$},\\
\textrm{linspan}\left\{G_n \right\}
 &  \mbox{ if  $n\ge 2$ and $\frac{n}{n-1}\le p< \frac{n}{n-2}$},\\
 \left\{0 \right\}
 &  \mbox{ if  $n> 2$ and $ p \geq \frac{n}{n-2}$},
\end{cases}
\end{equation}
where the difference $u-\Pi_pu $ is understood in the space $L^p(\mathbb{R}^n)$  to which both $u$ and $\Pi_pu$ belong.
\end{proposition}
\begin{proof} We only need to prove \eqref{eqn-Range Pi_p} which, follows from the inclusion
$$
\range(\Pi_p) \subset W^{2,p}(\mathbb{R}^n).
$$
\end{proof}

\begin{example}\label{ex-n=1}
Assume that $n=1$ and $p \in (1,+\infty)$. Then every function $u \in W^{\Delta,p}(\mathbb{R}\setminus\{0\})$ can be identified with an element of $W^{2,p}(-\infty,0)\oplus
W^{2,p}(0,+\infty)$ and therefore determines four numbers $u(0^{\pm})$ and $u^\prime(0^{\pm})$. The numbers $c_0$, $c_1$ and the values $f_0:=f(0)$, $f_1:=f^\prime(0)$   in the representation  \eqref{eqn-representation-n=1} of $u$ can be found by solving the following non-degenerate $4\times 4$ system of linear equations:
\begin{equation}\label{eqn-4x4 system}
\left\{
\begin{array}{lllll}
c_0&-c_1&+2f_0&&=2u(0^+)\cr
c_0&+c_1&+2f_0&&=2u(0^-)\cr
-c_0&+c_1&&+2f_1&=2u^\prime(0^+)\cr
c_0&+c_1&&+2f_1&=2u^\prime(0^-).
\end{array}
\right.
\end{equation}
This system can be solved giving
\begin{equation}\label{eqn-4x4 solution}
\left\{
\begin{array}{ll}
c_0&= \bigl[ u^\prime(0^-)-u^\prime(0^+)\bigr]\cr
c_1&= \bigl[ u(0^-)-u(0^+)\bigr]\cr
f_0&=\frac{1}{2} \bigl[u(0^+)+u(0^-)+u^\prime(0^+)-u^\prime(0^-) \bigr]\\
f_1&=\frac{1}{2} \bigl[u(0^+)-u(0^-)+u^\prime(0^+)+u^\prime(0^-) \bigr].
\end{array}
\right.
\end{equation}
\end{example}
\begin{example}\label{ExG}
Let again $n=1$ and let  us consider $u\colon u(x)=G_1(\beta x)$ for some $\beta \not=0$, where $G_1$ has been defined in \ref{eqn_G_1}. Then
$$
u(0^+)=\frac{1}{2}= u(0^-), \qquad  u^\prime(0^+)=-\frac{\beta}{2}, \quad u^\prime(0^-)=\frac{\beta}{2}.
$$
Thus, the constants $c_0$, $c_1$ and the values $f_0:=f(0)$, $f_1:=f^\prime(0)$   in the representation  \eqref{eqn-representation-n=1} of $u$ are
\begin{equation}\label{eqn-4x4 solution-beta}
\left\{
\begin{array}{rclrcl}
c_0&=&\beta u^\prime(0^-)-u^\prime(0^+)
,\quad &c_1&=0
\cr
f_0&=&\frac{1}{2} \bigl( 1-\beta
\bigr) 
,\quad &f_1&=0.
\end{array}
\right.
\end{equation}
Hence
\begin{align*}
u(x)=G_1(\beta x)=\beta G_1(x)+f(x),
\end{align*}
where
\begin{align*}
f(x)=G_1(\beta x)-\beta G_1( x)=\frac12 \bigl[ \e^{-\beta \vert x\vert}-\beta \e^{-\vert x\vert }  \bigr].
\end{align*}
Note that $f \in W^{p,2}(-\infty,+\infty)$,  $f(0)=\frac12  (1-\beta)$,  and $f^\prime(0)=0$, what agrees with \eqref{eqn-4x4 solution-beta}.
\end{example}

\begin{remark}\label{R3.9}
Assume that $n=3$ and  $p\in \bigl(\frac{3}{2},3\bigr)=\bigl(\frac{n}{n-1}, \frac{n}{n-2}\bigr)$ or $n=2$ and $p \in [ 2,+\infty)= \bigl[\frac{n}{n-1},+\infty\bigr)$.
Then $2-\frac{n}{p}>0$ and therefore by the Sobolev embedding theorem   every  $f \in W^{2,p}(\mathbb{ R}^n)$ is continuous.

 Let $u\in W^{\Delta,p}(\mathbb{R}^n\setminus\{0\})$ has the representation
 \begin{align}\label{E3.8-a}
u=c_0G_n+f,\qquad  c_0 \in \mathbb{R}, \  f\in W^{2,p}(\mathbb{ R}^n).
\end{align}
 Then since $G_n(x)\to+\infty$ as $x \to 0$ and $f$ is continuous at $0$, we infer that
\begin{equation} \label{E3.8}
c_0=\lim_{x \to 0} {u(x) \over G_n(x)} .
\end{equation}
Moreover,
\begin{equation}
\label{E3.9}
f(0)= \lim_{x \to 0} [u(x)- c_0G_n(x)]	.
\end{equation}
Assume now that  $n \geq 4$  and $p< \frac{n}{n-2}$.  Then,  $2-\frac{n}{p}<2-n\frac{n-2}{n}=4-n\le 0$, the function $f$ in the representation \eqref{E3.8-a} above
\dela{$u=c_0G_n+f$, $f\in W^{2,p}(\mathbb{ R}^n)$, } of the function $u$, is in general neither   bounded nor continuous. In some cases, however,  we are able to calculate $c_0$, see the proof of Lemma \ref{L3.25} below.
\end{remark}


\begin{remark} \label{R3.8}
If $n \geq 2$ and $1<p<\frac{n}{n-2}$, then  the spaces $W^{\Delta,p}(\mathbb{R}^n\setminus\{0\})$ and $W^{2,p}(\mathbb{R}^n\setminus\{0\})$ do not coincide, see  Corollary \ref{cor-G_n}.  \dela{Indeed,  $G_n \in W^{\Delta,p}(\mathbb{R}^n\setminus\{0\})$ but $G_n \not \in W^{2,p}(\mathbb{R}^n\setminus\{0\})$.} On the other hand, if $n \geq 3$ and $p \geq \frac{n}{n-2}$, then by part (iii) of Theorem \ref{T3.10} we have
$$
W^{2,p} (\mathbb{ R} ^n)\equiv  W^{\Delta ,p} (\mathbb{ R} ^n\setminus\{0\}).
$$
Since obviously, $W^{2,p} (\mathbb{ R} ^n)\subset W^{2,p} (\mathbb{R}^n\setminus\{0\})$ and $W^{2,p} (\mathbb{R}^n\setminus\{0\})\subset  W^{\Delta ,p} (\mathbb{ R} ^n\setminus\{0\})$ we infer that
\begin{equation}
W^{2,p} (\mathbb{ R} ^n)\equiv W^{2,p} (\mathbb{R}^n\setminus\{0\})\equiv  W^{\Delta ,p} (\mathbb{ R} ^n\setminus\{0\}).
\end{equation}
Note that  if  $n=1$ and $p \ge 1$,  then
\begin{align*}
W^{\Delta,p}(\mathbb{R}\setminus\{0\}) &\equiv  W^{\Delta ,p} (0,+\infty ) \oplus W^{\Delta ,p} (-\infty , 0) \\
     &\equiv  W^{2 ,p} (0,+\infty ) \oplus W^{2 ,p} (-\infty , 0) \\
     &\equiv  W^{2,p}( \mathbb{R}\setminus \{0\}).
\end{align*}
\end{remark}
\begin{remark}\label{R3.21}
{\rm  For an arbitrary open set	$\mathscr{O} \subset  \mathbb{ R}^n$ let $W^{2,p}(\mathscr{O})$ be defined  as  in the introduction   while  $\widetilde{W}{}^{2,p}(\mathscr{O})$ be defined as in \cite[Definition  4.2.1/1]{Triebel}, i.e.  $u \in  \widetilde{W}{}^{2,p}(\mathscr{O})$ if and only if   there exists $v\in W^{2,p}(\mathbb{ R}^n)$ such that
$u=v\vert_{  \mathscr{O}}$. Let us recall that the  norm in $\widetilde{W}{}^{2,p}(\mathscr{O})$ is defined as
$$
\n \,  u \, \n\,_{2,p}=\inf \, \bigl\{\vert  v \vert_{W^{2,p}(\mathbb{ R}^n)}\colon  v \in W^{2,p}(\mathbb{ R}^n), \;u=v\vert_{  \mathscr{O}}  \bigr\}.
$$

As we  have mentioned before, it is well known that if $\mathscr{O}$ is  a bounded domain in $\mathbb{ R}^n $  satisfying  the cone condition, then the two spaces coincide, see e.g. \cite{Triebel}. However, if $\mathscr{O}=\mathbb{ R} ^n\setminus\{0\}$, then only  $\widetilde{W}{}^{2,p}(\mathscr{O}) \subset W^{2,p}(\mathscr{O})$. Indeed, if $n=1$, then as in Example \ref{ex-n=1} the space  $W^{\Delta,p}(\mathbb{R}\setminus\{0\})=W^{\Delta,p}(\mathbb{R}\setminus\{0\})$ can be identified with an element of $W^{2,p}(-\infty,0)\oplus W^{2,p}(0,+\infty)$ and hence it is strictly larger than $\widetilde{W}{}^{2,p}(\mathbb{R}\setminus\{0\})$.
}
\end{remark}
\begin{remark}\label{R3.10}
{\rm Let $\beta\in \mathbb{R}\setminus \{0\}$. Then obviously for every $u\in W^{\Delta, p}(\mathbb{R}^n\setminus\{0\})$, the function
$$
u_\beta\colon  \mathbb{R}^n\setminus\{0\} \ni x\mapsto u(\beta x)
$$
belongs to the space $W^{\Delta, p}(\mathbb{R}^n\setminus\{0\})$ as well. }
\end{remark}
Given $\lambda >0$ set
\begin{equation}\label{EGbeta}
G_{n,\lambda}(x)= G_n(\sqrt{\lambda} x), \qquad x\in \mathbb{R}^n\setminus\{0\}.
\end{equation}
\begin{remark} Note that only in dimension $n=2$, $G_{n,\lambda}$ is the  solution of equation $(\lambda-\Delta)G_{n,\lambda}=\delta_0$ in $\mathbb{R}^n$. If $n=3$, the solution of this equation in the distributional sense is  associated to the  function $(4\pi \vert x\vert)^{-1}\e^{-\sqrt{\lambda}\vert x\vert}$.
\end{remark}

\begin{lemma} \label{L3.25}
(i) $G_{1,\lambda}= \sqrt{\lambda} G_1+ f$, where $f= G_{1,\lambda}-\sqrt{\lambda} G_1\in W^{2,p}(\mathbb{R})$ for any $1<p<+\infty$. \\
(ii) $G_{2,\lambda}= c_0 G_2+ f$, where
$$
c_0= \lim_{x\to 0} \frac{G_{2,\lambda}(x)}{G_2(x)} = 1,
$$
and $f= G_{2,\lambda}-c_0 G_2= G_{2,\lambda}-G_2 \in W^{2,p}(\mathbb{R}^2)$ for any $p\in (1,+\infty)$. \\
(iii) If $n=3$, then $G_{n,\lambda}= c_0 G_n+ f$, where
$$
c_0= \lim_{x\to 0} \frac{G_{n,\lambda}(x)}{G_n(x)} = \lambda ^{\frac{n-2}{2}}
$$
and $f= G_{n,\lambda}-c_0 G_n\in W^{2,p}(\mathbb{R}^n)$ for any $p\in \left(1,\frac{n}{n-2}\right)$.\\
(iv) Let $n\ge 4$. Then $G_{n,\lambda}= c_0 G_n+ f$, where
$$
c_0= \lim_{x\to 0} \frac{G_{n,\lambda}(x)}{G_n(x)} = \lim_{x\to 0} \frac{D_1G_{n,\lambda}(x)}{D_1G_n(x)} =\lambda ^{\frac{2-n}{2}}
$$
and $f= G_{n,\lambda}-c_0 G_n\in W^{2,p}(\mathbb{R}^n)$ for any $p\in \left(1,\frac{n}{n-2}\right)$.
\end{lemma}
\begin{proof} The lemma follows from Theorem \ref{Thre} and Remark \ref{R3.9}.  Namely,  $G_n$ with all derivatives decay at infinity. Therefore, $D^\alpha G_n \in L^p(\mathbb{R}^n)$ implies that $D^\alpha G_n \in L^r(\mathbb{R}^n)$ for any $r<p$.  In particular if for a certain constant $c$,
$$
G_{n,\lambda}- c G_n\in W^{2,p}(\mathbb{R}^n),
$$
then
$$
G_{n,\lambda}- c G_n\in W^{2,r}(\mathbb{R}^n)\qquad \text{for any $r<p$.}
$$
This argument works in cases (i) to (iii), as $f\in W^{2,p}(\mathbb{R}^n)$ for  $2-\frac{n}{p}>0$ is continuous, see Remark \ref{R3.9}. The case (iv) is more complicated as $p\in \left[ \frac{n}{n-1},\frac{n}{n-2}\right)$ and $n\ge 4$ implies  $2-\frac{n}{p}\le 0$. Therefore assume that
$$
G_n(\sqrt{\lambda} x)= c_0G_n(x)+ f(x), \qquad x\not =0,
$$
where $f\in W^{2,p}(\mathbb{R})$. Then
$$
\sqrt{\lambda}(D_1 G_n)(\sqrt{\lambda} x)= c_0D_1 G_n(x)+ D_1f(x).
$$
We knew that $D_1f\in L^p(\mathbb{R}^n)$ and if $p\in \left[ \frac{n}{n-1},\frac{n}{n-2}\right)$ then $D_1G_n\not \in L^p(\mathbb{R}^n)$, see \eqref{E3.5}. On the other hand, by \eqref{E3.5b} and \eqref{Pol},
\begin{equation}\label{eqn-pytanie}
\lim_{x\to  0} \frac{\sqrt{\lambda} (D_1 G_n)(\sqrt{\lambda} x)}{D_1G_n(x)} = \lambda ^{\frac{2-n}{2}}.
\end{equation}
Hence, by \eqref{eqn-pytanie},   the following  limit exists 
$$
\lim_{x\to 0} \frac{D_1 f(x)}{D_1G_n(x)}= \lambda ^{\frac{2-n}{2}}-c_0.
$$
If $c_0\not = \lambda ^{\frac{2-n}{2}}$, then there are $r>0$ and $\delta >0$ such that
$$
\vert D_1f(x)\vert \ge \delta \vert D_1G_n(x)\vert, \qquad x\colon \vert x\vert \le r.
$$
This is in contradiction with  $p$ integrability of $D_1f$.
\end{proof}

\section{A Green formula (or Stokes Theorem) on $\mathbb{R}^n\setminus\{0\}$}\label{sec-Green formula}
Let us now make a small detour. We have defined $G_n$ as the unique  element of $\mathscr{S}^\prime(\mathbb{R}^n)$ satisfying
\begin{equation}\label{eqn-4.100}
\left(  G_n, (1-\Delta ) \varphi\right)= \varphi(0)\qquad  \text{for every  $\varphi \in \mathscr{S}(\mathbb{R}^n)$.}
\end{equation}
We will now show that the above identity is valid for a larger class of test functions.
\begin{lemma}\label{L4.1} Assume that $p,q \in (1,+\infty)$  are conjugate exponents, i.e. satisfy
\begin{equation}\label{eqn-conjugate}
\frac1p+\frac1q=1,
  \end{equation}
\begin{trivlist}
\item[a)] If  $n=1,2$ or $n \geq 3$ and $1<q<\frac{n}{n-2}$, i.e. equivalently  $p>\frac{n}{2}$, then $G_n\in L^q(\mathbb{R}^n)$, $W^{2,p}(\mathbb{R}^n)\embed \mathscr{C}(\mathbb{R}^n)$ and
\begin{equation}\label{eqn-4.101}
\left( G_n,(1-\Delta ) \varphi\right)= \varphi(0), \qquad \text{$\varphi \in W^{2,p}(\mathbb{R}^n)$.}
\end{equation}
\item[b)] If  $n=1$, then $\D  G_1\in L^q(\mathbb{R})$, $W^{2,p}(\mathbb{R})\embed \mathscr{C}^1(\mathbb{R})$ and
\begin{equation}\label{eqn-4.102}
\left( \D   G_1,(1-\Delta ) \varphi\right)= -\D  \varphi(0),\qquad \text{$\varphi \in W^{2,p}(\mathbb{R})$.}
\end{equation}
\item[c)] If  $2\leq n < p$ (what implies that $q<\frac{n}{n-1}$), then $\Di G_n\in L^q(\mathbb{R}^n)$, $i=1,\ldots,n$, $W^{2,p}(\mathbb{R}^n)\embed \mathscr{C}^1(\mathbb{R}^n)$ and
\begin{equation}\label{eqn-4.103}
\left( D_i G_n,(1-\Delta ) \varphi\right)= -D_i\varphi(0),\qquad \text{$\varphi \in W^{2,p}(\mathbb{R}^n)$.}
\end{equation}
\end{trivlist}
\end{lemma}
\begin{proof} Let us choose and fix $p,q \in (1,+\infty)$  such that $\frac{1}{p}+\frac{1}{q}=1$. We are showing the first part of the lemma.  Let $p>\max\left\{1,\frac{n}2\right\}$. Then  $G_n\in L^q(\mathbb{R}^n)$ by \adda{\eqref{E3.4}} in Lemma \ref{L3.3}. Therefore, the map
\begin{equation}\label{eqn-4.104}
W^{2,p}(\mathbb{R}^n)\ni \varphi \mapsto \left( G_n,(1-\Delta ) \varphi\right)\in \mathbb{C}
\end{equation}
is continuous and  linear. Since  $W^{2,p}(\mathbb{R}^n)\embed \mathscr{C}(\mathbb{R}^n)$ by the Sobolev embedding theorem, the map
$$
W^{2,p}(\mathbb{R}^n)\ni \varphi \mapsto \varphi(0)\in \mathbb{C}
$$
is also continuous and linear. Since by \eqref{eqn-4.100} these  maps concide  on a dense subset $\mathscr{S}(\mathbb{R}^d)$,   equality \eqref{eqn-4.101} follows.

Assume additionally that $n=1$. Then,  by equality \eqref{E3.5} in Lemma \ref{L3.3},  $\D  G_1\in L^q(\mathbb{R})$. Moreover, since $2-\frac{1}{p}>1$, $W^{2,p}(\mathbb{R}) \embed \mathscr{C}^1(\mathbb{R})$ as required. Therefotre, equality \eqref{eqn-4.102} can be proved in the exactly the same manner as we proved earlier equality \eqref{eqn-4.101}.

We are showing now the last part of the lemma.  If $ 2\leq n < p$,  then   $q<\frac{n}{n-1}$, and  again by equality \eqref{E3.5} in Lemma \ref{L3.3}, the function $\Di G_n\in L^q(\mathbb{R}^n)$. Since by assumptions,  $2-\frac{n}{p}>1$,  we infer that  $W^{2,p}(\mathbb{R}^n) \embed \mathscr{C}^1(\mathbb{R}^n)$ as required. Therefor, equality \eqref{eqn-4.103} can be proved in the exactly the same manner as we proved earlier equality \eqref{eqn-4.101}.
\end{proof}

Our objective in this section is to prove the following Stokes theorem (or Green formula) in the domain $\mathscr{O}=\mathbb{ R}^n \setminus \{0 \}$. As above  let $p,q>1$ be such that $\frac 1p+\frac 1q=1$. Let us define the following  bilinear form
\begin{align}\label{eqn-form E}
\mathscr{E}(u,v)&:= \left(\Delta_{\mathbb{R}^n\setminus\{0\}}   u,v\right)-\left(u,   \Delta_{\mathbb{R}^n\setminus\{0\}}   v\right),
\end{align}
for $u\in W^{\Delta,p}(\mathbb{R}^n)\setminus \{0\})$ and $v\in W^{\Delta,q}(\mathbb{R}^n)\setminus \{0\})$,
where, by definition   $\Delta_{\mathbb{R}^n\setminus\{0\}} u  \in L^p (\mathbb{R}^n\setminus\{0\})$,  $\Delta_{\mathbb{R}^n\setminus\{0\}} v  \in L^q (\mathbb{R}^n\setminus\{0\})$, and $(\cdot, \cdot)$ is  the  pairing between $L^p(\mathbb{R}^n\setminus\{0\})$ and $L^q(\mathbb{R}^n\setminus\{0\})$;
\begin{align}\label{eqn-Lp-Lq-duality}
(f,g)&=\int_{\mathbb{R}^n\setminus\{0\}}f(x)\overline{g(x)}\d x,\quad   \text{$f\in L^p(\mathbb{R}^n\setminus\{0\})$, $g \in L^q(\mathbb{R}^n\setminus\{0\})$.}
\end{align}
Assume that $u$ and $v$ have the following representation
\begin{align}\label{eqn-representation of u and v}
u=a_0G_n+ \sum_{i=1}^n a_i\Di G_n + f, \quad v=b_0G_n+ \sum_{i=1}^n b_i\Di G_n + g,
\end{align}
where $f\in W^{2,p}(\mathbb{R}^n)$ and $g\in W^{2,q}(\mathbb{R}^n)$ and $a_i,b_i$ are real numbers.

Since $\Delta_{\mathbb{R}^n\setminus\{0\}}G_n=G_n$, we have
$$
\mathscr{E}(G_n,G_n)= \left( \Delta_{\mathbb{R}^n\setminus\{0\}}G_n, G_n\right)- \left( G_n, \Delta_{\mathbb{R}^n\setminus\{0\}}G_n\right)= 0.
$$
Next, for $f\in W^{2,p}(\mathbb{R}^n)$ and $g\in W^{2,q}(\mathbb{R}^n)$, we have
\begin{align*}
\mathscr{E}(f,g)&= \left( \Delta_{\mathbb{R}^n\setminus\{0\}}f, g\right)- \left( f, \Delta_{\mathbb{R}^n\setminus\{0\}}g\right)= \left( \Delta f, g\right)- \left( f, \Delta g\right) \\
&= \int_{\mathbb{R}^n}\left[ \Delta f(x)\overline{g(x)}- f(x)\overline{\Delta g(x)}\right]\d x =0.
\end{align*}
Since
$$
\Delta_{\mathbb{R}\adda{^n}\setminus \{0\}}\Di G_n = \Di G_n,
$$
 as soon as it is well-defined, we have
\begin{align}
\mathscr{E}\left( \Di G_n, G_n\right)&=
\left( \Delta_{\mathbb{R}^n\setminus\{0\}}\Di G_n, G_n\right)- \left( \Di G_n, \Delta_{\mathbb{R}^n\setminus\{0\}}G_n\right)
\\
&=\left( \Di G_n, G_n\right)- \left( \Di G_n, G_n\right) =0.
\end{align}
Similarly we can prove that
\begin{align}
\mathscr{E}\left( G_n, \Dj G_n\right)= 0, \\
\mathscr{E}\left( \Di G_n, \Dj G_n\right)= 0.
\end{align}

Therefore, we deduce the following identity
\begin{align}\label{eqn-form E-2}
\mathscr{E}(u,v)= a_0\overline{g(0)} -\sum_{i=1}^n a_1 \overline{\Di g(0)} - f(0)\overline{b_0} +\sum_{i=1}^n \Di f(0)\overline{b_i},
\end{align}
if $u$ and $v$ are of the form \eqref{eqn-representation of u and v}.

This identity implies  the following Stokes-type theorem in which  we consider all possible cases.
\begin{theorem}\label{T4.2}
a) Assume that $n=1$. Then, by Theorem \ref{T3.10}, $u$ and $v$ have the unique representation
$$
u= a_0G_1+a_1\D   G_1+ f, \qquad v= b_0G_1+b_1\D   G_1+ g,
$$
and
$$
\mathscr{E}(u,v)=  a_0\overline{g(0)} - a_1 \overline{\D  g(0)} - f(0)\overline{b_0} +\D  f(0)\overline{b_1}.
$$
b) Let $n=2$, and let $1<p<\frac{n}{n-1}=2$. Then $q\in (2,+\infty)= \left(\frac{n}{n-1},\frac{n}{n-2}\right)$. Consequently, by Theorem \ref{T3.10},
$$
u= a_0G_2 + a_{1,0} \frac{\partial }{\partial x_1}G_2+a_{0,1} \frac{\partial }{\partial x_2}G_2+ f, \qquad v= b_0G_2+ g,
$$
and
$$
\mathscr{E}(u,v)=  a_0\overline{g(0)} - a_{1,0} \overline{\frac{\partial }{\partial x_1}g(0)} - a_{0,1} \overline{\frac{\partial }{\partial x_2}g(0)}  - f(0)\overline{b_0}.
$$
c) Let $n=2$, and let $p=2=q$. Then, by Theorem \ref{T3.10},
$$
u= a_0G_2 + f, \qquad v= b_0G_2+ g,
$$
and
$$
\mathscr{E}(u,v)= a_0\overline{g(0)}-f(0)\overline{b_0}.
$$
d) Let $n=3$, and let $1<p<\frac{n}{n-1}=\frac{3}{2}$. Then $\frac{n}{n-2}= 3<q$. Consequently, by Theorem \ref{T3.10},
$$
u= a_0G_3 + \sum_{i=1}^3  a_{i} \Di G_3+ f, \qquad v= g,
$$
and
$$
\mathscr{E}(u,v)= a_0 \overline{g(0)} - \sum_{i=1}^3 a_i \overline{\Di g(0)}.
$$
e) Let $n=3$, and let $\frac{n}{n-1}= \frac{3}{2}= p$. Then $3=q$. Consequently, by Theorem \ref{T3.10},
$$
u= a_0G_3 + f, \qquad v= g,
$$
and
$$
\mathscr{E}(u,v)= a_0\overline{g(0)}.
$$
f) Let $n=3$, and let $\frac{n}{n-1}= \frac{3}{2}< p<3=\frac{n}{n-2}$. Then $\frac{n}{n-1}= \frac{3}{2}<q<3$. Consequently, by Theorem \ref{T3.10},
$$
u= a_0G_3+ f, \qquad v= b_0G_3+g,
$$
and
$$
\mathscr{E}(u,v)= a_0\overline{g(0)}- f(0)\overline {b_0}.
$$
g) Let $n>3$, and let $\frac{n}{n-1}\le p<\frac{n}{n-2}$. Then $q\ge n/2 \ge \frac{n}{n-2}$. Consequently, by Theorem \ref{T3.10},
$$
u= a_0G_n +  f, \qquad v= g,
$$
and
$$
\mathscr{E}(u,v)= a_0\overline{g(0)}.
$$
h) Let $n>3$, and let $1<p< \frac{n}{n-1}$. Then $q>n>\frac{n}{n-2}$. Consequently, by Theorem \ref{T3.10},
$$
u= a_0G_n + \sum_{i=1}^n  a_{i} \Di G_n+ f, \qquad v= g,
$$
and
$$
\mathscr{E}(u,v)= a_0\overline{g(0)}- \sum_{i=1}^n a_i\overline{\Di g(0)}.
$$
\end{theorem}

\section{Non-uniqueness of extensions of a certain class}

Below we  reformulate Theorem 1.33, p. 46   from \cite{Nagel_1986} in  a slightly more general fashion. Moreover, we provide its full proof. Let us recall the \emph{resolvent set} $\rho(A)$ of   a densely defined linear operator $(A,D(A))$   on a Banach space $X$ is the set of all $\lambda \in \mathbb{R}$ such that $(\lambda I -A )^{-1} \in \mathscr{ L}(X)$. Finally $D_0\subset D(A)$ is the \emph{core} of $A$ if $D_0$ is dense in $D(A)$ endowed with the graph norm \begin{theorem} \label{thm-Nagel}
 Assume that $X$ is a real Banach space and $(A,D(A))$ is a closed densely defined operator in $X$ such that  $\rho(A)\not= \emptyset$. Assume that $D_0 \subset D(A)$ is a subspace and let us put
\begin{equation}\label{eqn-def-A_0}
D(A_0):=D_0 \; \mbox{ and }\; A_0u =Au \quad \text{ for  $u \in D_0$.}
\end{equation}
\begin{itemize}
\item[(i)]
If $D_0$ is a core for $A$,  then $A$ is the only extension $B$ of the operator $(A_0,D_0)$ such that
\begin{equation}\label{eqn-resolvent-B-A}
\rho(B) \cap \rho(A)\not= \emptyset.
\end{equation}
\item[(ii)]
If $D_0$ is not a core for $A$, then there exist infinitely many extensions $B$ of  the operator $A_0$ which satisfy condition \eqref{eqn-resolvent-B-A}.
\end{itemize}
\end{theorem}
\begin{proof} Assume that $D_0$ is  a core for $A$. Suppose that $B$ is extension of the operator $(A_0,D_0)$ for which condition \eqref{eqn-resolvent-B-A} holds. Then by \cite[Proposition 2.14]{Rudin-FA_1973} the graph of the operator $(\lambda I -B )^{-1}$ is closed (in $X^2$) and so also the graph of the operator $B$ is closed and hence $B$ is a closed operator\footnote{This standard argument implies that if an operator $(A,D(A))$ has non-empty resolvent set, then it is closed.}.  On the other hand, since $D_0$ is a core for $A$, we deduce that $A$ is the closure of $A_0$. Hence  $A \subset B$. Since $B$ satisfies condition \eqref{eqn-resolvent-B-A} we can choose and fix $\lambda \in \rho(B) \cap \rho(A)$. Then $(\lambda I -A )\colon  D(A) \to X$ and $(\lambda I -B )\colon  D(B) \to X$  are bijections. Since $A\subset B$ we have  $\lambda I -A  \subset \lambda I -B $. Since they are bijections we have   $\lambda I -A  =\lambda I -B $ and so $A=B$.  The proof of (i) is complete.

Our  proof of the second part differs from the proof provided in  \cite{Nagel_1986}. Let us choose and fix $\lambda \in  \rho(A)$. Then $\lambda I -A \colon  D(A) \to X$  is a  linear bijection. We equip $D(A)$ with the norm induced by the map $\lambda I-A$. Therefore $\lambda I-A\colon D(A)\mapsto X$ is a linear isomorphism.

Assume that $D_0$ is not a core for $A$. Then by the Hahn--Banach theorem, see e.g. \cite[Theorem 3.2]{Rudin-FA_1973}, there exists a linear and bounded map
\begin{equation}\label{eqn-HBT}
\phi\colon  D(A) \to \mathbb{R} \text{  such that }\ \phi_{|D_0}=0 \mbox{ and } \Vert \phi  \Vert_{\mathscr{L}(D(A),\mathbb{R})} = 1.
\end{equation}
Let $f \in D(A)$ be such that
\begin{equation}\label{eqn-f}
0\not = \phi(f)=   \vert f \vert_{X} < \Vert (\lambda I -A)^{-1} \Vert^{-1}_{\mathscr{L}(X,D(A))}.
\end{equation}
Define a linear operator
\begin{equation}\label{eqn-C}
C_f\colon D(A) \ni u \mapsto  \phi(u) f \in D(A) \subset X.
\end{equation}
Let us note that $ C_f$ is linear and bounded  (with respect to the graph-norm)  operator and
$$
\Vert C_f \Vert_{\mathscr{L}(D(A),X)} \leq \vert f \vert_{X}  \Vert \phi  \Vert_{\mathscr{L}(D(A),\mathbb{R})}< \Vert (\lambda I -A)^{-1} \Vert^{-1}_{\mathscr{L}(X,D(A))}.
$$
Therefore $\lambda I -A-C_f$ is invertible. Indeed, since  the operator $I-C_f(\lambda I -A)^{-1}$ invertible we have
\begin{align}
\left( \lambda I -A-C_f \right)^{-1}&=
\left[ \left( I -C_f (\lambda I -A)^{-1}\right)(\lambda I -A)  \right]^{-1} = (\lambda I -A)^{-1}\left( I -C_f (\lambda I -A)^{-1}\right)^{-1}.
\end{align}

Put $B_f=C_f+A$, $D(B_f)=D(A)$. Then it follows from the above that
\begin{trivlist}
\item[(i)]    $B_f(u)=u$ if $u \in D_0$ so that $A_0 \subset B_f$;
\item[(ii)]    $B_f \not=A$ since  $f\not=0$;
\item[(iii)]  $\lambda \in \rho(B_f)$ so that $B_f$ satisfies condition \eqref{eqn-resolvent-B-A}.
\end{trivlist}
\end{proof}

\begin{remark}\label{rem-core}
Let us observe that if $D_0$ is not a core for $A$, then $D_0$ may still be dense in $X$. For example assume that $p >\frac{n}{2}$, $X=L^{p}(\mathbb{R}^n)$, $D_0=\mathscr{C}^\infty_0(\mathbb{R}^n\setminus\{0\})$, and $D(A)= W^{2,p}(\mathbb{R}^n)$, $A:=\Delta$. Then by Theorem \ref{thm-density} it follows that  the set  $D_0$ is not dense in the Banach space $W^{2,p}(\mathbb{R}^n)$ but obviously it is dense in the Banach space $L^{p}(\mathbb{R}^n)$.
\end{remark}

We think it will be enlightening and useful to formulate Theorem 1.33 from \cite{Nagel_1986} in a form closer to the needs of the present paper.

\begin{theorem}
\label{thm-Nagel-original}
Assume that $X$ is a real Banach space and $(A,D(A))$ is  an infinitesimal generator of a $C_0$-semigroup on  $X$.
Assume that $D_0 \subset D(A)$ is a subspace and let us define operator $A_0$ by  \eqref{eqn-def-A_0}.
\begin{itemize}
\item[(i)]
If $D_0$ is a core for $A$,  then $(A,D(A))$ is the only extension $(B,D(B))$ of the operator $(A_0,D_0)$ which is
an infinitesimal generator of a $C_0$-semigroup on  $X$.
\item[(ii)]
If $D_0$ is not a core for $A$, then there exists infinitely many extensions $(B,D(B))$ of  the operator $(A_0,D_0)$ each of which
being an   infinitesimal generator of a $C_0$-semigroup on  $X$.
\end{itemize}
 \end{theorem}

Let us define a linear unbounded operator $A_0$ by
\begin{equation}\label{eqn-A_0}
D(A_0)=\mathscr{ C}_0^\infty (\mathbb{R}^n\setminus\{0\})=:D_0, \qquad  A_0u:=\Delta u\quad    \mbox{ for } u \in D(A_0).
\end{equation}
Our next theorem  deals with the extensions of the  operator $A_0$.  In fact the result is a consequence of Theorem \ref{thm-Nagel}. Note that $A_0$ is the usual Laplace operator while in the paper \cite{ABD_1995-a} the authors denoted by $A_0$ the minus Laplacian.
\begin{theorem}\label{thm-extensions for general p}
Let $n \in \mathbb{N}$ and  $p\in (1,+\infty)$.
 Assume that $A_1=A_{1,p}$ is the Friedrichs extension of the operator $A_0$ in the Banach space $L^{p}(\mathbb{R}^n)$, i.e.
\begin{equation}\label{eqn-A_1}
D(A_1)= W^{2,p}(\mathbb{R}^n),\qquad  A_1u:=\Delta u\quad  \mbox{ for } u \in D(A_1).
\end{equation}
 Then
\begin{itemize}
\item If $n \geq 3$ and $p \leq  \frac{n}{2}$, then $A_1$ is the  unique extension of $A_0$ in the class of generators of $C_0$-semigroups  on the Banach space $L^{p}(\mathbb{R}^n)$.
\item If  $p >  \frac{n}{2}$, then,  in the class of generators of $C_0$-semigroups  on the Banach space $L^{p}(\mathbb{R}^n)$,  there exist  infinitely many operators $A$ which are
 extensions of  the operator $A_0$. In particular, if $n=1,2$, then  $A_{1}$ is always a non-unique extension of $A_0$ within the class above.
 \end{itemize}
\end{theorem}
\begin{proof}
Let us begin the proof by an observation that not only $D(A_{1,p})= W^{2,p}(\mathbb{R}^n)$ but also  the standard norm on  $ W^{2,p}(\mathbb{R}^n)$ and  the graph norm on  $D(A_{1,p})$ are equivalent.

Assume that $p >  \frac{n}{2}$. Then by Theorem \ref{thm-density} we infer that   the set  $D_0$ is not dense in the Banach space $W^{2,p}(\mathbb{R}^n)$. Thus in view of the definition \eqref{eqn-A_1} of the operator $A_{1,p}$ we infer that the set  $D_0$ is not dense in the Banach space $D(A_{1,p})$ and therefore  the set  $D_0$ is not a core of $A_{1,p}$. Hence, it follows from the second part of Theorem \ref{thm-Nagel} there exist infinitely many extensions of $A_0$,  of the form $A_{1,p}+C$, where $C\colon D(A_{1,p}) \to D(A_{1,p})$ is bounded w.r.t. the graph norm.

Assume next that $p \leq  \frac{n}{2}$.  Then by Theorem \ref{thm-density} it follows that  the set  $D_0$ is  dense in the Banach space $W^{2,p}(\mathbb{R}^n)$. Thus in view of the definition \eqref{eqn-A_1} of the operator $A_{1,p}$ we infer that the set  $D_0$ is  a core of $A_{1,p}$, and hence, the  desired conclusion follows directly from the first  part of Theorem \ref{thm-Nagel}.
\end{proof}

\begin{remark}\label{rem-extensions}
If we are interested in the existence of nontrivial extensions of the operator $A_0$ in both spaces $L^{p}(\mathbb{R}^n)$ and $L^{q}(\mathbb{R}^n)$, where $p \in (1,+\infty)$ and $q\in (1,+\infty)$ are conjugate exponents, i.e. satisfy condition  \eqref{eqn-conjugate}, then by the above Theorem \ref{thm-extensions for general p} we would have that
\begin{equation}\label{ytr}
\mbox{ $\frac1p<\frac2n$ and $\frac1q<\frac2n$.}
\end{equation}
Therefore, in view of \eqref{eqn-conjugate}, the dimension $n$ would have  to satisfy $1< \frac4n$, i.e. $n\le 3$.
\end{remark}
\begin{remark}
In the notation of Theorem \ref{thm-extensions for general p}, assume now that $n\le 3$ and $p$ and $q$ are conjugate and satisfying \eqref{ytr}. A natural question arrises;  can we find a non-trivial bounded operator from $C\colon D(A_{1,p})= W^{2,p}(\mathbb{R}^n)  \to W^{2,p}(\mathbb{R}^n)$ and $C\colon D(A_{2,q})= W^{2,q}(\mathbb{R}^n)  \to W^{2,q}(\mathbb{R}^n)$ such that $A_{1,p}+C$ and $A_{2,q}+ C$ are  extensions of $A_0$ in  $L^{p}(\mathbb{R}^n)$ and $L^{q}(\mathbb{R}^n)$, respectively?
\end{remark}
\begin{remark}
A natural question is whether the second part of Theorem \ref{thm-extensions for general p} true when the r\^ole of the operator $A_{1,p}$ is played by other, the  non-trivial, i.e. the non-Friechrichs,  extensions of the operator $A_0$?
\end{remark}

\section{Some density results.}

Let us recall that the Banach spaces $\wdp$ and $W^{\Delta,p}(\mathbb{R}^n)= W^{2,p}(\mathbb{R}^n)$ have been defined in Definition \ref{def-W^B,p} and that, see \eqref{eqn-pi_O-Delta} that  the following restriction map
\[
\pi\colon  W^{\Delta,p}(\mathbb{R}^n)= \ni u \mapsto u|_{\mathbb{R}^n\setminus \{0\}} \in \wdp
\]
is a well defined and bounded linear map. Although there are domains $\mathscr{O}$ such that  the map
$$
\pi_{\mathscr{O}}\colon  W^{\Delta,p}(\mathbb{R}^n) \ni u \mapsto u|_{\mathscr{O}} \in W^{\Delta,p}(\mathscr{O})
$$
 is not injective, the  map $\pi=\pi_{\mathbb{R}^n\setminus \{0\}}$ is injective. In fact, $\pi_{\mathscr{O}}$ is  injective if the Lebesgue measure of the complement set $\mathbb{R}^n \setminus \mathscr{O}$ is equal to zero. The injectivity of the map $\pi$ can be interpreted as a statement that $W^{\Delta,p}(\mathbb{R}^n)  \subset \wdp$. In fact, by part (iii) of Theorem \ref{T3.10}, $W^{\Delta,p}(\mathbb{R}^n)  = \wdp$ if $n \geq 3$ and $p\ge \frac{n}{n-2}$.  Moreover, by parts (i) and (ii)  of Theorem \ref{T3.10}, we can identify the space $W^{\Delta,p}(\mathbb{R}^n)$ with a closed  subspace $\wdp$ of co-dimension $1$ or $n+1$, depending on whether  $p \in [ \frac{n}{n-1},  \frac{n}{n-2})$ or $p \in (1, \frac{n}{n-1})$.

\begin{definition}\label{def-wzpo}
Let respectively $\wzpo$ and $\wdpo$ be the closure of the space $\mathscr{ C}_0^\infty (\mathbb{R}^n\setminus\{0\})$ in the space $\wzp$ and respectively in  $\wdp$.
\end{definition}
\begin{remark}\label{R5.2}
(i) Since as justified above, $W^{\Delta,p}(\mathbb{R}^n)$ is  a closed  subspace of $\wdp$ and $\mathscr{ C}_0^\infty (\mathbb{R}^n\setminus\{0\})$ is a subset of the former space, the closure of $\mathscr{ C}_0^\infty (\mathbb{R}^n\setminus\{0\})$  in  $W^{\Delta,p}(\mathbb{R}^n)$ is equal to the closure of $\mathscr{ C}_0^\infty (\mathbb{R}^n\setminus\{0\})$  in $W^{2,p}(\mathbb{R}^n)$. Hence
\begin{align*}
\wdpo &= \mbox{the closure of }\mathscr{ C}_0^\infty (\mathbb{R}^n\setminus\{0\}) \mbox{  in the space } W^{\Delta,p}(\mathbb{R}^n)\\
&= \mbox{the closure  of }\mathscr{ C}_0^\infty (\mathbb{R}^n\setminus\{0\}) \mbox{  in the space }W^{2,p}(\mathbb{R}^n)\\
&= \wzpo.
\end{align*}
\end{remark}

We have the following explicit characterization of the spaces/space introduced above.
\begin{proposition}\label{prop-density}
$(i)$ If  $2-\frac{n}{p}\le 0$, i.e. $p \leq  \frac{n}{2}$, then  $\wzpo= W^{2, p}(\mathbb{R}^n)$.
\\
$(ii)$ If  $2-\frac{n}{p}\in (0,1]$, i.e. $ \frac{n}{2}< p \leq n$, then
$$
\wzpo = \left\{ f\in  W^{2, p}(\mathbb{R}^n)\colon f(0)=0\right\}.
$$
$(iii)$ If $2-\frac{n}{p}\in (1,2)$, i.e. $  p > n$, then
$$
\wzpo= \left\{ f\in  W^{2, p}(\mathbb{R}^n)\colon f(0)=0=D_if(0),\quad i=1,\ldots,n\right\}.
$$
\end{proposition}
\begin{proof} Write
$$
\mathcal{H}:=
\begin{cases} W^{2,p}(\mathbb{R}^n), & \mbox{ if } p \leq \frac{n}{2}, \cr
\bigl\{ u\in W^{2,p}(\mathbb{R}^n): u(0)=0\bigr\}, &\mbox{ if } p \in ( \frac{n}{2},n], \cr
\bigl\{ u\in W^{2,p}(\mathbb{R}^n): u(0)=0, \Di u(0)=0,\, i=1, \cdots, n\bigr\},&\mbox{ if } p >n .
\end{cases}
$$
We are showing that $\mathcal{H}= \wzpo$. Obviously $\mathscr{ C}_0^\infty (\mathbb{R}^n\setminus\{0\})\subset \mathcal{H}$ and by Sobolev embeddings,  $\mathcal{H}$ is a closed subspace of  $W^{2, p}(\mathbb{R}^n)$.  Suppose by contradiction that the set  $\mathscr{ C}_0^\infty (\mathbb{R}^n\setminus\{0\})$ is not dense in  $\mathcal{H}$. Thus by the Hahn--Banach theorem there exists a non-zero continuous linear functional $\omega \colon  W^{2,p}(\mathbb{R}^n) \to \mathbb{R}$ such that $\omega =0$ on $\mathcal{H}$. In other words, $\omega \in [W^{2,p}(\mathbb{R}^n)]^\prime \subset \mathscr{S}^\prime(\mathbb{R}^n)\subset \mathscr{D}^\prime(\mathbb{R}^n)$ and, since $\omega=0$ on $\mathscr{ C}_0^\infty (\mathbb{R}^n\setminus\{0\})$,  $\supp(\omega) = \{0\}$.  Hence, by  \cite[Theorems 6.24(c) and 6.25]{Rudin-FA_1973}, the distribution there exist a finite set $A$ on $n$-tuples  and family of  numbers $c_\alpha \in \mathbb{ R}$,  $ \alpha \in A$,  such that
$$
\omega = \sum _{ \alpha \in A }c_\alpha D^\alpha \delta _0.
$$
We have to show that  $\omega \equiv0$ or $\omega \not = 0$ on $\mathcal{H}$. This  follows from Lemma \ref{Lem}  as $[W^{2,p}(\mathbb{R}^n)]^\prime = W^{-2,q}(\mathbb{R}^n)$, where $p$ and $q$ are conjugate exponents.  Indeed,  $p\leq \frac{n}{2}$ implies that $q\ge \frac{n}{n-2}$, and $\omega \in W^{-2,q}(\mathbb{R}^n)$ implies $c_\alpha =0$ for any $\alpha$. Next,  $p \in ( \frac{n}{2},n]$ implies $q \in [ n/(n-1), n/(n-2)]$ and $\omega  \in W^{-2,q}(\mathbb{R}^n)$ implies $c_\alpha =0$  for any  $\alpha\colon \vert \alpha \vert \ge 1$. Hence either $\omega =0$ or $\omega \not =0$ on $\mathcal{H}$. Finally, if $p>n$, then $q\in (1,n/(n-1))$, and  $\omega   \in W^{-2,q}(\mathbb{R}^n)$ implies $c_\alpha =0$  for any  $\alpha\colon \vert \alpha \vert > 1$. Hence either $\omega \equiv0$ or $\omega \not  =0$ on $\mathcal{H}$.
\end{proof}

\section{Symmetric extensions of the Laplace operator on $\mathbb{R}^n\setminus \{0\}$}
\label{Section5}

Recall that $\Delta_{\mathbb{R}^n\setminus\{0\}}$ denotes   the distributional Laplacian;
\begin{equation}\label{eqn-Delta}
\Delta_{\mathbb{R}^n\setminus\{0\}}\colon  \mathscr{D}^\prime({\mathbb{R}^n\setminus\{0\}}) \to \mathscr{ D}^\prime({\mathbb{R}^n\setminus\{0\}}).
\end{equation}
Let us observe that in view of our definition  of the formally adjoint differential operator,   see equality  \eqref{eqn-def-B-adjoint},
the adjoint of $\Delta_{\mathbb{R}^n\setminus\{0\}}$ is equal to $\Delta_{\mathbb{R}^n\setminus\{0\}} \in \mathscr{L}(\mathscr{ C} _0^\infty (\mathbb{R}^n\setminus\{0\}))$. Here and below, for a normed vector space $Z$, by 
$ \mathscr{L}(Z)$ we denote the normed vector space of all bounded linear operators from $Z$ to $Z$.

In what follows we assume that $p>1$ and $q>1$ are   conjugate exponents, i.e. satisfying  condition\eqref{eqn-conjugate}. 

Given a densely defined linear  operator $A$, with domain $D(A)$,  on $L^p(\mathbb{R}^n)$ with domain $D(A)$ we denote by $A^\ast$,  with domain  $D(A^\ast)$, the adjoint operator acting on $L^q(\mathbb{R}^n)$.  Since the spaces $L^p(\mathbb{R}^n)$ and resp. $L^q(\mathbb{R}^n)$ are isometrically isomorphic to the  spaces  $L^p(\mathbb{R}^n\setminus\{0\})$ and resp. $L^q(\mathbb{R}^n\setminus\{0\})$ we may as well consider the operators $A$ and $A^\ast$ on these latter spaces.

The aim  of this section is to characterize all \dela{densely defined and closed} linear operators $\left(A,D(A)\right)$ on $L^p(\mathbb{R}^n)$ such that:
\begin{equation}\label{E5.1}
\mathscr{ C}_0^\infty (\mathbb{R}^n\setminus\{0\}) \subset D(A) \cap D(A^\ast),
\end{equation}
\begin{equation}\label{E5.2}
Au =A^\ast u=- \Delta_{\mathbb{R}^n\setminus\{0\}} u \mbox{ for } u \in \mathscr{ C} _0^\infty (\mathbb{R}^n\setminus\{0\}),
\end{equation}
\begin{equation}\label{E5.3}
A=\left(A^\ast\right)^\ast,
\end{equation}
\begin{equation}\label{E5.4}
\exists \lambda >0 \colon (\lambda +A )^{-1} \in \mathscr{ L}(L^p(\mathbb{ R}^n)).
\end{equation}
Clearly, if an operator  $\left(A,D(A)\right)$ satisfies conditions \eqref{E5.1} to \eqref{E5.3}, then it is closed and densely defined.

We start with the following result.
\begin{proposition}\label{P5.4}
Let us assume that  $A$ is  a closed  and densely defined operator in the Banach space $\lpr $  satisfying conditions \eqref{E5.1} and \eqref{E5.2}. Then
\begin{equation}
\label{E5.7}
\wzpo \subset D(A) \subset W^{\Delta,p}(\mathbb{R}^n\setminus\{0\}).
\end{equation}
and
\begin{equation}
\label{E5.5}
Au =-\newtilde{\Delta_{\mathbb{R}^n\setminus\{0\}}u} \qquad \text{ for }u\in D(A).
\end{equation}
Moreover,
\begin{equation}
\label{eqn-7.17}
\wzqo \subset D(A^\ast) \subset W^{\Delta,q}(\mathbb{R}^n\setminus\{0\})
\end{equation}
and
\begin{equation}
\label{E5.6}
A^\ast v =-\newtilde{\Delta_{\mathbb{R}^n\setminus\{0\}}v} \qquad \text{ for }v\in D(A^\ast).
\end{equation}
\end{proposition}
\begin{proof} We begin with a proof of the first   inclusion in  \eqref{E5.7}.   For this aim let us choose and fix  $u \in \wzpo$.  Then by Definition \ref{def-wzpo} there exists a sequence $\left\{ u_k\right\}_{k=1}^\infty \subset \mathscr{ C}_0^\infty  (\mathbb{R}^n\setminus\{0\})$, such that  $u_k  \to  u  $	in  $W^{\Delta,p}(\mathbb{R}^n\setminus\{0\}) $ as $k \to \infty$.  
This  implies  that $\Delta_{\mathbb{R}^n\setminus\{0\}} u_k \to \Delta_{\mathbb{R}^n\setminus\{0\}}u$ in $L^p(\mathbb{ R}^n)$. On the other hand,  $-\Delta_{\mathbb{R}^n\setminus\{0\}}u_k =Au_k$,  for every $k \in \mathbb{N}$, by  \eqref{E5.2}.  Therefore,  since $A$ is  a closed operator, it follows that $u \in D(A)$ and $Au=-\Delta _{\mathbb{R}^n\setminus\{0\}}u$.

Now we will prove the second  inclusion in  \eqref{E5.7}. For this aim let us choose and fix $u \in D(A)$ and $\varphi \in \mathscr{	C}_0^\infty (\mathbb{R}^n\setminus\{0\})$. Thus $u \in L^p(\mathbb{ R}^n)$  and so  $u \in \mathscr{D}^\prime({\mathbb{R}^n\setminus\{0\}})$. Therefore $\Delta_{\mathbb{R}^n\setminus\{0\}} u \in \mathscr{D}^\prime({\mathbb{R}^n\setminus\{0\}})$ and  by Definition \ref{D2.2},
\begin{align}\label{eqn-7.18}
-(\Delta_{\mathbb{R}^n\setminus\{0\}}u,\varphi)&= (u, -\Delta_{\mathbb{R}^n\setminus\{0\}} \varphi).
\end{align}
Since $\varphi \in \mathscr{C}_0^\infty (\mathbb{R}^n\setminus\{0\})$,  by Assumptions \eqref{E5.1} and \eqref{E5.2} we infer that
\begin{align}\label{eqn-7.19}
 (u, -\Delta_{\mathbb{R}^n\setminus\{0\}} \varphi)&=(u, A^\ast \varphi),
\end{align}
where on the RHS we have the duality between spaces $L^p(\mathbb{ R}^n)$ and $L^q(\mathbb{ R}^n)$.
Next, by the definition of the adjoint operator,   we infer that $(u, A^\ast \varphi)= ( Au,\varphi)$. Therefore,
\begin{align}\label{eqn-7.20}
(\Delta_{\mathbb{R}^n\setminus\{0\}}u,\varphi)&= ( -Au,\varphi).
\end{align}
Therefore, in view of Definition \ref{def-W^B,p} we deduce that $u \in W^{\Delta,p}(\mathbb{R}^n\setminus\{0\})$, as claimed. Next we observe that  equality  \eqref{E5.5} follows directly from just proved identity \eqref{eqn-7.20}.

\dela{Let us choose and fix $u\in D(A)$ and $v \in L^q(\mathbb{ R}^n)$. We can find  a   sequence $\{v_k\}_k \subset \mathscr{	C}_0^\infty (\mathbb{R}^n\setminus\{0\})$ such that $v_k \to v$ in $L^q(\mathbb{ R}^n)$. Then, in view of \eqref{E5.1} and \eqref{E5.2}, by Theorem \ref{T4.2},  we have that}
\dela{First we need to show that $\Delta_{\mathbb{R}^n\setminus\{0\}}u \in L^p(\mathbb{ R}^n)$ so that  the below makes sense. By Definition \ref{D2.2}, Assumption \eqref{E5.2} and the definition of the adjoint operator,   we infer that for any $\varphi \in \mathscr{	C}_0^\infty (\mathbb{R}^n\setminus\{0\})$,
\begin{align}\label{E5.2-2-1}
-(\Delta_{\mathbb{R}^n\setminus\{0\}}u,\varphi)&= (u, \Delta_{\mathbb{R}^n\setminus\{0\}} \varphi) = (u, A^\ast \varphi)\\
&= ( Au,\varphi).
\end{align}
Thus $\Delta_{\mathbb{R}^n\setminus\{0\}}u=Au$ in  $\mathscr{D}^\prime({\mathbb{R}^n\setminus\{0\}})$ and since by assumptions $Au \in L^p(\mathbb{ R}^n)$  we infer that also $\Delta_{\mathbb{R}^n\setminus\{0\}}u \in L^p(\mathbb{ R}^n)$.}

Let us now provide proofs of   \eqref{eqn-7.17} and  \eqref{E5.6}. Since for  $u \in D(A)$ and $\varphi \in \mathscr{	C}_0^\infty (\mathbb{R}^n\setminus\{0\})$ we have
\begin{align}\label{eqn-7.28}
(Au, \varphi)&=(-\Delta_{\mathbb{R}^n\setminus\{0\}}u,\varphi)= (u, -\Delta_{\mathbb{R}^n\setminus\{0\}} \varphi)
\end{align}
we infer that $\mathscr{C}_0^\infty (\mathbb{R}^n\setminus\{0\}) \subset D(A^\ast)$ and $A^\ast \varphi=  -\Delta_{\mathbb{R}^n\setminus\{0\}} \varphi$. This proves
the first inclusion in  \eqref{eqn-7.17}.  To prove the second, let us choose and fix $v \in D(A^\ast)$. Hence the map
\begin{equation}\label{eqn-A-adjoint}
  D(A) \ni u \mapsto (Au, v) \in \mathbb{R}
\end{equation}
is continuous w.r.t. $L^p(\mathbb{R}^n)$-topology.  Thus, by the just proved \eqref{E5.5} and \eqref{E5.7} the map
\begin{equation}\label{eqn-A-adjoint-2}
  \wzpo \ni \varphi  \mapsto (-\Delta_{\mathbb{R}^n\setminus\{0\}} \varphi  , v) \in \mathbb{R}
\end{equation}
is continuous w.r.t. $L^p(\mathbb{R}^n)$-topology. Thus, there exists a unique $z \in L^q(\mathbb{R}^n)$ such that
\begin{equation}\label{eqn-A-adjoint-3}
(-\Delta_{\mathbb{R}^n\setminus\{0\}} \varphi  , v)=(\varphi, z), \;\; \mbox{ for every } \varphi \in  \wzpo.
\end{equation}
This, according to Definition \ref{def-W^B,p} implies that $v \in  W^{\Delta,q}(\mathbb{R}^n\setminus\{0\})$ and $-\Delta_{\mathbb{R}^n\setminus\{0\}} v =z$. This proves the second inclusion in
 \eqref{eqn-7.17}.

In order to prove identity \eqref{E5.6} we proceed as follows. We choose and fix  $v \in D(A^\ast)$. Then, since $\wzpo \subset D(A)$,  we have the following train of equalities
\begin{align}
(\varphi, A^\ast v)&= (A\varphi, v)=(-\Delta_{\mathbb{R}^n\setminus\{0\}} \varphi  , v)=(\varphi,z), \;\; \varphi \in \wzpo,
\end{align}
where the last equality is a consequence of \eqref{eqn-A-adjoint-3}. Therefore $A^\ast v=z$ and since by an earlier proven identity
$-\Delta_{\mathbb{R}^n\setminus\{0\}} v =z$ we deduce \eqref{E5.6}. The proof is complete.
\end{proof}

\begin{remark} Let us observe that equalities   \eqref{E5.5} are  natural generalizations of assumption \eqref{E5.2}.
\end{remark}

Recall our notation  $G_{n,\lambda }(x)= G_n(\sqrt{\lambda} x)$, $\lambda >0$.  Let
$$
\mathscr{Y}_{n,p,\lambda} := \begin{cases}
\textrm{linspan}\left\{G_{n, \lambda}, \frac{\partial G_{n, {\lambda}}}{\partial x_1}, \ldots , \frac{\partial G_{n, {\lambda}}}{\partial x_n}\right\},
 &  \mbox{ if  $n\ge 1$ and $1<p<\frac{n}{n-1}$},\\
\textrm{linspan}\left\{G_{n, {\lambda}} \right\}
 &  \mbox{ if  $n\ge 2$ and $\frac{n}{n-1}\le p< \frac{n}{n-2}$},\\
 \left\{0 \right\}
 &  \mbox{ if  $n> 2$ and $ p \geq \frac{n}{n-2}$},
\end{cases}
$$
Recall, see Proposition \ref{prop-Pi_p}, that $\mathscr{Y}_{n,p,1} = \mathscr{Y}_{n,p}$.  Clearly, if $e\in \mathscr{Y}_{n,p,\lambda}$, then $\Delta_{\mathbb{R}^n\setminus \{0\}} e= \lambda e$. Therefore we have the following result.
\begin{lemma}\label{L5.5}
Assume  that a densely defined and closed linear operator $A$ satisfies conditions  \eqref{E5.1} and  \eqref{E5.2}. If  $\lambda>0$ belongs to the resolvent set  $\rho(A)$, then $\mathscr{Y}_{n,p,\lambda }\cap D(A)=\{0\}$.
\end{lemma}

Let us equip $D(A)$ with the graph norm.  Let
\begin{equation}\label{eqn-j_A}
j_A \colon D(A) \to W^{\Delta,p}(\mathbb{R}^n\setminus\{0\})
\end{equation}
be the natural imbedding map, which is well defined because of \eqref{E5.7}. Moreover, by \eqref{E5.5}, it follows  that $j_A$ is continuous. Similarly, by \eqref{eqn-7.17}, the imbedding map
\begin{equation}\label{eqn-j_A}
j_{A^\ast} \colon D(A^\ast) \to W^{\Delta,q}(\mathbb{R}^n\setminus\{0\})
\end{equation}
is well defined and continuous. Let us recall that  $\Pi_p\colon  W^{\Delta,p}(\mathbb{R}^n\setminus\{0\}) \to  W^{2,p}(\mathbb{R}^n)$ and
$\xi_p \colon  \wdp \to \mathbb{C}$ or $\xi_p\colon \wdp \to \mathbb{C}^{n+1}$, see Definitions \ref{def-Pi_p} and \ref{def-xi_p}.

Assume that $1\in \rho(A)$.  It is well  known that the  linear operator
$$
I-\Delta \colon W^{2,p}(\mathbb{R}^n) \to L^p(\mathbb{ R}^n)
$$
 is an isomorphisms. By assumption \eqref{E5.4}, also the  linear operator
$$
(I+A)^{-1} \colon  L^p(\mathbb{ R}^n) \to D(A)
$$
 is an  isomorphisms. Hence    \[(I+A^\ast)^{-1} \colon  L^q(\mathbb{ R}^n) \to D(A^\ast)\] is also a linear isomorphism, see Lemma 10.2 in \cite{Pazy_1983}. Therefore   we deduce that the following linear maps
\begin{equation}\label{E5.9}
\begin{aligned}
\gamma_{A} &:=  j_A \circ (I+A)^{-1}\circ (I-\Delta) \colon W^{2,p}(\mathbb{R}^n)  \to W^{\Delta,p}(\mathbb{R}^n\setminus\{0\}),
\\ \gamma_{A^\ast} &:=  j_{A^\ast} \circ (I+A^\ast)^{-1}\circ (I-\Delta) \colon W^{2,q}(\mathbb{R}^n)  \to W^{\Delta,q}(\mathbb{R}^n\setminus\{0\}),
\end{aligned}
\end{equation}
are  not only well defined but   bounded as well.

\begin{lemma}\label{L5.6}
We have
\begin{align}
u&= \gamma_{A}\circ \Pi _pu \quad \mbox{ for every } u\in W^{\Delta,p}(\mathbb{R}^n\setminus\{0\}),\\
v&= \gamma_{A^\ast}\circ \Pi_q v\quad  \mbox{ for every } v\in W^{\Delta,q}(\mathbb{R}^n\setminus\{0\}).
\end{align}
\end{lemma}
\begin{proof}
It is enough to prove the first part of  our lemma since the proof of the second part is fully analogous. Since $D(A)$ is  a dense subset of $W^{\Delta,p}(\mathbb{R}^n\setminus\{0\})$ we would  deduce that  $u= \gamma_{A}\circ \Pi _pu$ for every $u \in D(A)$. For this aim let us choose and fix  $u\in D(A)\subset W^{\Delta,p}(\mathbb{R}^n\setminus\{0\})$. We will show that
$$
(I+A)^{-1}(I-\Delta)\Pi _pu = u \mbox{ in } L^p(\mathbb{R}^n)
$$
or equivalently that
\begin{equation}\label{kjh}
(I-\Delta)\Pi_p u= (I+A)u \mbox{ in } L^p(\mathbb{R}^n).
\end{equation}
By Proposition \ref{prop-Pi_p} we have  $u-\Pi_pu \in \mathscr{Y}_{n,p}$.  Then \eqref{kjh} follows from the following three observations:
\begin{trivlist}
\item[(i)] According to Proposition \ref{P5.4}   $A u= -\Delta_{\mathbb{R}^n\setminus \{0\}}u$.
\item[(ii)] For any $R\in \mathscr{Y}_{n,p}$, $\Delta_{\mathbb{R}^n\setminus\{0\}}R=R$.
\item[(iii)] $\Delta_{\mathbb{R}^n\setminus\{0\}}\Pi_p u=\Delta \Pi_p u$, see identity \eqref{eqn-Range Pi_p} in Proposition \ref{prop-Pi_p}.
\end{trivlist}
In fact we have
\begin{align*}
(I+A)u&= u - \Delta_{\mathbb{R}^n\setminus\{0\}}u = u -\Delta_{\mathbb{R}^n\setminus\{0\}}(u-\Pi_pu +\Pi_pu)\\
&= u - (u-\Pi_p u)-   \Delta_{\mathbb{R}^n\setminus\{0\}}\Pi_p u \\
&= \Pi_p u -\Delta \Pi_p u.
\end{align*}
\end{proof}
Let us denote by $\Delta_{n,p}$ the Friedrichs  Laplace operator on $\mathbb{R}^n$ with domain $W^{2,p}(\mathbb{R}^n)$.
\begin{lemma}\label{L7.5}
Assume that $p\ge q$ and  $W^{\Delta,p}(\mathbb{R}^n\setminus\{0\})= W^{2,p}(\mathbb{R}^n)$. If $A= -\Delta_{n,p}$ then $A^\ast = -\Delta_{n,q}$.
\end{lemma}
\begin{proof} The case of $W^{\Delta,q}(\mathbb{R}^n\setminus\{0\})= W^{2,q}(\mathbb{R}^n)$ is obvious. Assume that
$$
W^{\Delta,q}(\mathbb{R}^n\setminus \{0\})= \mathbb{C}G_n \oplus  W^{2,q}(\mathbb{R}^n).
$$
Then $\frac{n}{n-1}\le q<\frac{n}{n-2}$, $p> \frac{n}{2}$, $W^{2,p}(\mathbb{R}^n)\hookrightarrow \mathscr{C}(\mathbb{R}^n)$ and by the Green formula, see Theorem \ref{T4.2} e) and g),
$$
\left( a_0G_n+ f, \Delta u\right)= a_0\overline{u(0)} + \left( \Delta f, u\right).
$$
Thus $a_0G_n+f\in D(A^\ast)$ if and only if
$$
D(A)=W^{2,p}(\mathbb{R}^n)\ni u\mapsto a_0\overline{u(0)} \in \mathbb{C}
$$
has continuous extension to $L^p(\mathbb{R}^n)$, which holds if and only if $a_0=0$.

Consider now the last case $1<q<\frac{n}{n-1}$, and $p>2$. Then $W^{2,p}(\mathbb{R}^n) \hookrightarrow \mathscr{C}^1(\mathbb{R}^n)$,
$$
W^{\Delta,q}(\mathbb{R}^n\setminus \{0\})= \mathbb{C}G_n \oplus  \bigoplus_{k=1}^n \mathbb{C}D_k G_n\oplus  W^{2,q}(\mathbb{R}^n),
$$
and by the Green formula, see Theorem \ref{T4.2} , we have
$$
\left( a_0G_n+\sum_{k=1}^n a_k D_k G_n +  f, \Delta u\right)= a_0\overline{u(0)} + \sum_{k=1}^n a_k \overline {D_ku(0)}+ \left( \Delta f, u\right).
$$
Hence
$$
a_0G_n+\sum_{k=1}^n a_k D_k G_n +  f\in D(A^\ast)
$$
 if and only if $a_0=a_1=\ldots =a_n=0$.
\end{proof}

\bigskip

\begin{definition}\label{def-tau_beta}
Given $\beta \in \mathbb{C} \cup \{\infty\}$, we define a linear map $\tau_{\beta} : W^{\Delta, p}({\mathbb{R}^n}\setminus\{0\}) \to \mathbb{C}$ by
\begin{equation}\label{eqn-tau_beta}
\mbox{ if } u= c_0G_n+f, \;\; f \in W^{2, p}({\mathbb{R}^n} \mbox{ then }
\tau_{\beta}u:= \begin{cases} c_0-\beta f(0),
 &\mbox{ if } \beta\not=\infty,
\\
-f(0), &\mbox{ if } \beta=\infty.
\end{cases}
\end{equation}
The map $\tau_\beta$ will be called the boundary operator.
\end{definition}

\begin{theorem} \label{Th5.7}
Assume that $(A,D(A))$ satisfies \eqref{E5.1} to \eqref{E5.4}. Assume that $1\in \rho(A)$. The following assertions holds true:
\begin{enumerate}
\item[(i)] If either $n\ge 4$ or $n=3$ and $p\not\in (\frac{3}{2},3)$, then
$$
\mbox{$D(A)=W^{2,p}(\mathbb{R}^n)$ and  $D(A^\ast)=W^{2,q}(\mathbb{R}^n)$,}
$$
and $Au=\Delta u$ if $u \in D(A)$ and  $Av=\Delta v$ if $v \in D(A^\ast)$.
\item[(ii)] If $n=3$ and $p\in (\frac{3}{2},3)$,  or $n=2$, then there is a $\beta\in \mathbb{C}\cup\{\infty\}$ such that
\begin{align*}
D(A)&=\left\{u\in W^{\Delta,p}(\mathbb{R}^n\setminus\{0\})\colon \tau_{\beta} u=0\right\}, \\
D(A^\ast)&=\left\{v\in W^{\Delta,q}(\mathbb{R}^n\setminus\{0\})\colon \tau_{\overline \beta} v=0\right\}.
\end{align*}
\item[(iii)] If $n=1$, then there are complex numbers $\beta^i_j$, $i,j=0,1$ such that
\end{enumerate}
\begin{align*}
D(A)&= \left\{(\beta^0_0 f(0)+\beta^0_1 f'(0))G_1 +  (\beta^1_0 f(0)+\beta^1_1 f'(0))G'_1 + f\colon f\in W^{2,p}(\mathbb{R})\right\},\\
D(A^\ast)&= \left\{(\overline{\beta^0_0} f(0)-\overline {\beta^1_0} f'(0))G_1 +  (\overline \beta^1_0 f(0)+\overline {\beta^1_1} f'(0))G'_1 + f\colon f\in W^{2,q}(\mathbb{R})\right\}.
\end{align*}
\end{theorem}
\begin{proof} The first part of the theorem follows from Lemma \ref{L7.5}.  We are proving the second part. Assume that  $q\le p$. Hence
$$
W^{\Delta,p}(\mathbb{R}^n)\setminus\{0\})= \mathbb{C}G_n \oplus W^{2,p}(\mathbb{R}^n).
$$
By Theorem \ref{T3.10}, the map  $\xi_{0,p}$ defined by
\begin{align}\label{eqn-xi_0,p}
\xi_{0,p}  \colon  W^{\Delta,p}(\mathbb{R}^n\setminus\{0\})\ni c_0G_n+f\mapsto c_0\in \mathbb{C}
\end{align}
is  bounded and linear.

On the other hand,
by Lemma \ref{L5.6}, $\gamma_A\circ\Pi_p$ is the identity operator on $W^{\Delta,p}(\mathbb{R}^n\setminus\{0\})$, where  $\gamma_{A}$ is the linear map  defined  in \eqref{E5.9}. Hence we infer that
\begin{equation}\label{E5.10}
u=c_0G_n+f\in D(A)\Longrightarrow c_0=\xi_{0,p}\circ \gamma_{A}\circ \Pi_p(u)=\xi_{0,p}\circ \gamma_{A}(f).
\end{equation}
Now, recall, see \eqref{E5.1}, that $\mathscr{ C}_0^\infty (\mathbb{R}^n\setminus\{0\}) \subset D(A)$. Therefore, by \eqref{E5.10},  $\xi_{0,p}\circ\gamma_{A}(f)=0$ for any  $f\in \mathscr{ C}_0^\infty (\mathbb{R}^n\setminus\{0\})$. Since $\mathscr{ C}_0^\infty (\mathbb{R}^n\setminus\{0\})$ is dense in $\wzpo$, we have, by continuity,  that $\xi_{0,p}\circ \gamma_{A} (f)=0$ for any $f\in \wzpo$. Recall that either $n=2$ and $p\ge 2$ or $n=3$ and  $2\le p<3$. Therefore, $W^{2,p}(\mathbb{R}^n) \hookrightarrow \mathscr{ C}(\mathbb{R}^n)$, and  by Proposition \ref{prop-density}, $\xi_{0,p}\circ\gamma_{A}(f)=0$ for every $f\in W^{2,p}(\mathbb{R}^n)$ such that $f(0)=0$. Hence, we have the following fact
\begin{equation}\label{E5.11}
 \exists\, \beta\in \mathbb{C}\colon \forall f \in  W^{2,p}(\mathbb{R}^n) \;\; \xi_{0,p}\circ \gamma_{A}(f)= \beta f(0).
\end{equation}

\bigskip
We are going to show that
\begin{align*}
D(A)&= \left\{u= c_0G_3+f \in W^{\Delta ,p}(\mathbb{R}^3\setminus\{0\})\colon c_0=\beta f(0)\right\},\\
D(A^\ast)&= \left\{v= b_0G_3+g \in W^{\Delta ,q}(\mathbb{R}^3\setminus\{0\})\colon b_0=\overline{\beta} g(0)\right\}.
\end{align*}
 Recall that either $n=2$ and $p\ge 2$ or $n=3$ and  $2\le p<3$, $2/3 < q\le 2$.  Therefore, $W^{2,q}(\mathbb{R}^n) \hookrightarrow \mathscr{ C}(\mathbb{R}^n)$, and the right hand side of the last identity above is well-defined.

So far we knew that if $c_0G_n+ f\in D(A)$, then $c_0= \beta f(0)$. Therefore,  we have to exclude the case of $D(A)= W^{2,p}_0(\mathbb{R}^n)$.  Suppose that by contradiction  $D(A)= W^{2,p}_0(\mathbb{R}^n)$.  Since   $G_n \in L^q(\mathbb{R}^n)$ the Green's formula \eqref{eqn-4.101} yields
\begin{align}
W^{2,p}_0(\mathbb{R}^n) &\ni u \mapsto \langle Au, G_3\rangle = \langle - \Delta u, G_n\rangle=\langle u, G_n\rangle.
\end{align}
This  implies that $G_n\in D\left(A^\ast\right)$ which contradicts Lemma \ref{L5.5} with $\lambda =1$. For  by  definition $v\in W^{\Delta,q}(\mathbb{R}^3\setminus\{0\})$ belongs to $D\left(A^\ast\right)$ if and only if there is a constant $C$ such that
$$
\left\vert (A u,v)\right\vert \le C\,\left\vert u\right\vert_{L^p(\mathbb{R}^3)}, \qquad \forall\, u\in D(A).
$$
We have by assumption that $u=f$, where $f\in W^{2,p}(\mathbb{R}^3)$ is such that $f(0)=0$. Applying the Green formula we obtain
$$
\left \vert ( A u,G_n)\right\vert = \left\vert \left( \Delta f,G_n \right)\right\vert =\left \vert \left(f,G_n \right)\right\vert \le \vert G_n\vert_{L^q(\mathbb{R}^n)}\vert f\vert_{L^p(\mathbb{R}^n)},
$$
since by Lemma \ref{L3.4}, $G_n\in L^q(\mathbb{R}^n)$.

We have shown that there is a space $\mathscr{X}$ such that
$$
W^{2,p}_0(\mathbb{R}^n)\subsetneqq \mathscr{X}\subset W^{2,p}(\mathbb{R}^n),
$$
and
\begin{equation}\label{E5.13}
D(A)= \left\{u= \beta f(0)G_n+f \colon f\in \mathscr{X}\right\}.
\end{equation}
To see that
\begin{equation}\label{E5.14}
D\left(A^\ast\right)\supset \left\{v= \overline{\beta} g(0)G_n+g\colon g\in W^{2,q}(\mathbb{R}^n)\right\}
\end{equation}
take $g\in  W^{2,q}(\mathbb{R}^n)$, $v= \overline{\beta} g(0)G_n+g$. Then for every $u = \beta f(0)G_n+ f$, $f\in \mathscr{X}$, we have
\begin{align*}
\left( Au,v\right)&= - \left(\Delta_{\mathbb{R}^n\setminus\{0\}}u,v\right)= -\left(u,\Delta_{\mathbb{R}^n\setminus\{0\}} v\right)-\left( c_0\overline{g(0)} - f(0)\overline {b_0}\right)\\
&= \left(u, -\overline \beta g(0)G_n - \Delta g\right),
\end{align*}
because
$$
c_0\overline{g(0)} - f(0)\overline {b_0}= \beta f(0)\overline {g(0)}-f(0)\beta \overline{g(0)}=0.
$$
Since $\overline{\beta} g(0)G_n + \Delta g\in L^q(\mathbb{R}^3)$, we have
$$
\left\vert \left( Au, v\right)\right\vert \le \vert \overline{\beta} g(0)G_n + \Delta g\vert_{L^q(\mathbb{R}^N)}\vert u\vert _{L^p(\mathbb{R}^n)}.
$$
Therefore, we have \eqref{E5.14}. In fact,  if $n=3$, then  we have equality  in \eqref{E5.14}. For, we have
$$
W^{\Delta,q}(\mathbb{R}^3\setminus \{0\})= \mathbb{C}G_3\oplus W^{2,q}(\mathbb{R}^3)
$$
and using arguments leading to \eqref{E5.13}, we can find a subspace $\mathscr{Y}\subset W^{2,q}(\mathbb{R}^3)$ and a constant $\rho$ such that
$$
D(A^\ast )\subseteq \left\{v= \rho g(0)G_3+g \colon g\in \mathscr{Y}\right\}.
$$
Taking into account  assertion  \eqref{E5.14}, we conclude that $\rho=\overline{\beta}$ and $\mathscr{Y}=W^{2,q}(\mathbb{R}^3)$.

It remains to consider the case  $n=2$ and $p\ge 2$, From the proof of the second part we new \eqref{E5.13} and \eqref{E5.14}. Our aim is to show the reverse inclusion to the one in  \eqref{E5.14}.  For this aim  we adopt our earlier argument from the case of  dimension $n=3$.
Let us choose and fix $v\in D(A^\ast)$.
By Lemma \ref{L5.6} we infer that  $v= \gamma_{A^\ast} \Pi_qv$. Let $\xi_{i,q}$ be the linear map defined by
\begin{equation}\label{eqn-xi_i,q}
\xi_{i,q}\colon W^{\Delta,q}(\mathbb{R}^2\setminus\{0\})\ni c_0G_2+ c_1D_1G_2+ c_2 D_2 G_2+ g\mapsto c_i\in \mathbb{C}.
\end{equation}
Because by \eqref{eqn-7.17}, $D(A^\ast) \subset W^{\Delta,q}(\mathbb{R}^2\setminus\{0\})$, $v$ can be written in the form
\[ v=  c_0G_2+ c_1D_1G_2+ c_2 D_2G_2+ g\dela{\in D(A^\ast).}
\]
 Thus,
$$
 c_i=\xi_{i,q}\circ \gamma_{A^\ast}(g), \;\; i=1,2,3.
$$
By Lemma \ref{L5.5}, $c_i\circ \gamma_{A^\ast} (g)=0$ for $g\in W^{2,q}_0(\mathbb{R}^2)$. Therefore, by Proposition \ref{prop-density},  for each $i$ there are $\beta^{i}_0, \beta^{i}_1, \beta^{i}_2$, such that
$$
\xi_{i,q}\circ\gamma_{A^\ast} (g)= \beta^i_0 g(0) + \beta ^{i}_1\frac{\partial g}{\partial x_1}(0)+ \beta ^{i}_1\frac{\partial g}{\partial x_2}(0).
$$
Therefore
$$
v=  c_0G_2+ c_1D_1G_2+ c_2 D_2G_2+ g\in D(A^\ast)
$$
implies that
$$
c_i=  \beta^i_{0} g(0)+ \beta^i_1D_1g(0)+  \beta^i_2D_2g(0).
$$
This proves that
$$
D\left(A^\ast\right)=\left\{v= \overline{\beta} g(0)G_2+g\colon g\in W^{2,q}(\mathbb{R}^2)\right\}.
$$
Since $A= \left(A^\ast\right)^\ast$, we see by the Green formula that
$$
D\left(A\right)=\left\{u= {\beta} f(0)G_2+f\colon f\in W^{2,p}(\mathbb{R}^2)\right\}.
$$

\bigskip
In the proof of the last part of the theorem one can adopt arguments from the proofs of the second part. Namely, by Proposition \ref{prop-density} and Lemmata  \ref{L5.5} and \ref{L5.6}, there are complex numbers $\beta^i_i$ and $\tilde \beta^i_j$ such that
$$
u=  c_0G_1+ c_1G'_1+ f\in D(A)\quad \text{and}\quad v=  b_0G_1+ b_1G'_1+ g\in D(A^\ast)
$$
implies that
$$
c_i=  \beta^i_{0} f(0)+ \beta^i_1f'(0)\quad \text{and}\quad  b_i=  \tilde \beta^i_{0} g(0)+ \tilde \beta^i_1g'(0).
$$
The desired conclusion follows from the Green formula
\begin{align*}
&\left(u,A^\ast v\right)-\left(Au,v\right)= c_0\overline{g(0)} - c_1 \overline{g'(0)} - f(0)\overline{b_0} +f'(0)\overline{b_1}\\
&= \left(\beta^0_{0} f(0)+ \beta^0_1f'(0)\right)\overline{g(0)} -\left(\beta^1_{0} f(0)+ \beta^1_1f'(0)\right)\overline{g'(0)}\\
&\quad  - f(0)\overline{\tilde \beta^0_{0} g(0)+ \tilde \beta^0_1g'(0)}+ f'(0)\overline{\tilde \beta^1_{0} g(0)+ \tilde \beta^1_1g'(0)}\\
&= f(0)\overline{g(0)} \left(\beta^0_0-\overline{\tilde \beta^0_{0}}\right) + f'(0)\overline{g(0)}\left( \beta^0_1 -\overline{\tilde \beta^1_1}\right)\\
&\quad + f(0)\overline{g'(0)} \left( -\beta^1_0 -\overline{\tilde \beta^0_1 }  \right)+ f'(0)\overline{g'(0)}\left(-\beta^1_1+\overline{\tilde \beta ^1_1} \right),
\end{align*}
which leads to the following identities
$$
\beta^0_0= \overline{\tilde \beta^0_{0}}, \quad \beta^0_1 =\overline{\tilde \beta^1_1}, \quad \beta^1_0 = -\overline{\tilde \beta^0_1},\quad \beta^1_1=\overline{\tilde \beta ^1_1}.
$$
\end{proof}
\section{Case of $n=2,3$}\label{S8}
In this section we assume that $n=2$ or $ n=3$. Moreover, we assume that
\[
\mbox{
$n=2$ and  $p\in (1,+\infty)$  or $n=3$ and  $p\in \left(\frac{3}{2},3\right)$.}
\]

\begin{lemma} \label{L8.1} Let  $u\in W^{\Delta,p}(\mathbb{R}^n\setminus\{0\})$ has the representation
\begin{equation}\label{E8.1}
u=c_0G_n+f, \qquad f\in W^{2,p}(\mathbb{R}^n).
\end{equation}
Then for every $\lambda >0$, u  can be written uniquely in the form
\begin{equation}\label{E8.2}
u=c_0^\lambda G_{n,\lambda}+ f^\lambda, \qquad f^\lambda\in W^{2,p}(\mathbb{R}^n).
\end{equation}
Moreover,
$$
c_0^\lambda = \begin{cases} c_0
&\text{if $n=2$},\\
\sqrt{\lambda}c_0 &\text{if $n=3$}
\end{cases}
$$
and
$$
f^\lambda (0)= \begin{cases}
f(0)+ c_0 \frac{\log\lambda}{2}&\text{if $n=2$,}\\
f(0) - c_0 \frac{1}{4\pi}\left[ 1-\sqrt{\lambda}\right]&\text{if $n=3$.}
\end{cases}
$$
\end{lemma}
\begin{proof} Assume first that $n=2$. Then, see  Lemma \ref{L3.25} (ii), for every $\lambda >0$,
$$
G_{2,\lambda}= G_2 + g^\lambda , \qquad g^\lambda\in W^{2,p}(\mathbb{R}^2).
$$
We have
$$
g^\lambda(0)= \left[ G_{2,\lambda}- G_2\right](0)=  \lim_{x\to 0} \left[\log \frac{1}{\vert \sqrt {\lambda}x\vert}- \log\frac{1}{\vert x\vert}\right] = -\frac{1}{2} \log \lambda.
$$
Therefore, if $u$ has representation \eqref{E8.1}, then
$$
u = c_0G_{2,\lambda} - c_0g^\lambda + f.
$$
Hence we infer that equality  \eqref{E8.2} holds with
$$
f^\lambda = f-c_0g^\lambda, \qquad f^\lambda(0)= f(0)+ c_0\frac{\log\lambda}{2}.
$$

Assume that $n=3$. Then, see Lemma \ref{L3.25} (iii),
$$
G_{3,\lambda}= a_\lambda G_3 + g^\lambda, \qquad g^\lambda\in W^{2,p}(\mathbb{R}^2).
$$
Since $G_3(x)= (4\pi \vert x\vert )^{-1}\e^{-\vert x\vert}$ and
$$
a_\lambda= \lim_{x\to 0}\frac{G_{3,\lambda}(x)}{G_3(x)}= \vert \lambda  \vert ^{-\frac{1}{2}}\lim_{x\to 0} \e^{-\vert \sqrt{\lambda}x\vert + \vert x\vert}=\frac{1}{\sqrt{\lambda}},
$$
we have
$$
G_{3,\lambda}= \frac{1}{\sqrt{\lambda}} G_{3} + g^\lambda,
$$
and
\begin{align*}
g^\lambda(0)&= \lim_{x\to 0} \left[ G_{3,\lambda}(x) - \frac{1}{\sqrt{\lambda}} G_{3}(x)\right]=\lim_{x\to 0} \left( 4\pi \vert \sqrt{\lambda}x\vert\right)^{-1} \left( \e^{-\vert \sqrt{\lambda}x\vert} -\e^{-\vert x\vert}\right)\\
&= \frac{1}{4\pi}\lim_{x\to 0} \left[ \frac{\e^{-\vert \sqrt{\lambda} x\vert} -1}{\vert \sqrt{\lambda}x\vert}+ \frac{1- \e^{-\vert x\vert}}{\vert x\vert} \frac{1}{\sqrt{\lambda}}\right] \\
&= \frac{1}{4\pi}\left[ \frac{1}{\sqrt{\lambda}}-1\right].
\end{align*}
Therefore, if $u$ has representation \eqref{E8.1}, then
$$
u= \sqrt{\lambda}c_0G_{n,\lambda}- \sqrt{\lambda}c_0 g^\lambda + f= \sqrt{\lambda}c_0G_{n,\lambda} + f^\lambda,
$$
and
$$
f^\lambda(0)= -\sqrt{\lambda}c_0 g^\lambda(0)+f(0)= f(0)- c_0 \frac{1}{4\pi}\left[ 1-\sqrt{\lambda}\right].
$$
\end{proof}
\begin{lemma}\label{L8.2}
Let $\lambda >0$. Assume that $u$ has representation \eqref{E8.2}. For $\beta \in \mathbb{C}$ define
$$
\tau^\lambda_\beta  u= c_0^\lambda - \beta f^\lambda (0).
$$
Assume that
\begin{equation}\label{Elambdabeta}
\beta\not = \begin{cases}
\frac{2}{\log \lambda}&\text{if $n=2$,}\\
-4\pi \sqrt{\lambda} \left( 1-\sqrt{\lambda}\right)^{-1}&\text{if $n=3$.}
\end{cases}
\end{equation}
Then
$$
\tau^\lambda _\beta u= 0 \qquad \Longleftrightarrow \quad \tau_{\beta_\lambda}u=0,
$$
where
$$
\beta_\lambda= \begin{cases}
\beta \left(1-\beta \frac{\log \lambda}{2}\right)^{-1}&\text{if $n=2$,}\\
 \beta \left[ \sqrt{\lambda}+ \frac{\beta}{4\pi}(1-\sqrt{\lambda})\right] ^{-1}&\text{if $n=3$}
\end{cases}
$$
or equivalently
$$
\beta= \begin{cases}
\beta_\lambda  \left(1+\beta_\lambda \frac{\log \lambda}{2}\right)^{-1}&\text{if $n=2$,}\\
\sqrt{\lambda} \beta_\lambda \left[ 1- \frac{\beta_\lambda}{4\pi}(1-\sqrt{\lambda})\right] ^{-1}&\text{if $n=3$.}
\end{cases}
$$
\end{lemma}
\begin{proof} Assume that $n=2$. By Lemma \ref{L8.1}, we have
\begin{align*}
\tau^\lambda_\beta u&= c_0^\lambda -\beta f^\lambda (0)= c_0- \beta \left[ f(0)+ c_0\frac{\log\lambda}{2}\right] =c_0\left(1-\beta \frac{\log \lambda}{2}\right)-\beta f(0)\\
&= \left(1-\beta \frac{\log \lambda}{2}\right)\left[ c_0- \beta \left(1-\beta \frac{\log \lambda}{2}\right)^{-1}f(0)\right].
\end{align*}
Assume that $n=3$, then
\begin{align*}
\tau^\lambda_\beta u&= c_0^\lambda -\beta f^\lambda(0)= \sqrt{\lambda}c_0 -\beta \left[ f(0) - c_0 \frac{1}{4\pi}\left(1-\sqrt{\lambda}\right)\right]\\
&= \left[ \sqrt{\lambda}+ \frac{\beta}{4\pi}(1-\sqrt{\lambda})\right] \left\{ c_0- \beta \left[ \sqrt{\lambda}+\frac{\beta}{4\pi}(1-\sqrt{\lambda})\right] ^{-1}f(0)\right\}.
\end{align*}
\end{proof}

Given  $\lambda >0$ and $\beta \in \mathbb{C}$, define
\begin{equation}\label{E8.3}
\begin{aligned}
D(A_{\lambda, \beta})&= \left\{u= \beta f^\lambda(0)G_{n,\lambda}+ f^\lambda \in W^{2,p}(\mathbb{R}^n)\right\} \\
&=\left\{ u \in W^{\Delta,p}(\mathbb{R}^n\setminus\{0\}) \colon \tau^\lambda_\beta u=0\right\},\\
A_{\lambda,\beta}u&=  -\lambda\beta f^{\lambda}(0)G_{n, \lambda}- \Delta_{n,p}f^\lambda.
\end{aligned}
\end{equation}
\begin{lemma}\label{L8.3}
Let $\lambda\in (0,+\infty)\setminus\{1\}$ and
\[
\beta =
\begin{cases}
\frac{2}{\log\lambda},  & \mbox{ if } n=2, \\
 -4\pi \sqrt{\lambda} \left( 1-\sqrt{\lambda}\right)^{-1} ,  & \mbox{ if }  n=3.
 \end{cases}
\]

Then
\begin{align*}
D(A_{\lambda,\beta})&= \left\{ c_0G_n +f \colon c_0\in \mathbb{C}, \ f\in W^{2,p}_0(\mathbb{R}^n\setminus\{0\})\right\},\\
 A_{\lambda,\beta}u&= -c_0G_n -\Delta f.
 \end{align*}
 In other words $A_{\lambda, \beta}=A_\infty$.
\end{lemma}
\begin{proof} Assume that $n=2$. Then
$$
u=c_0G_2+f= c_0G_{2,\lambda}+ c_0\left(G_2- G_{2, \lambda}\right) + f\in D(A_{\gamma, \frac{2}{\log \lambda}})
$$
if and only if
$$
c_0= \frac{2}{\log \lambda}c_0\left[ G_2-G_{2,\lambda}\right](0) + \frac{2}{\log \lambda } f(0)= c_0+ \frac{2}{\log \lambda } f(0).
$$
That is if and only if $f(0)=0$. Summing up, we have proved, as required,  that  \[D(A_{\lambda, \beta})= \mathbb{C}G_2\oplus W^{2,p}_0(\mathbb{R}^2\setminus\{0\}).\]
The same arguments work in the case of $n=3$.
\end{proof}

\begin{definition}\label{D8.4}
We define $(A_\infty, D(A_\infty))$ and the Friedrich  Laplacian $A_0$ by putting
\begin{align}\label{eqn-A_infty}
D(A_\infty)&= \left\{ c_0G_n +f \colon c_0\in \mathbb{C}, \ f\in W^{2,p}_0(\mathbb{R}^n)\right\},\\
 A_\infty u&= -c_0G_n -\Delta f \\
 D(A_0)&= W^{2,p}(\mathbb{R}^n)\\
 A_0u&=-\Delta_{n,p}u, \;\; u \in W^{2,p}(\mathbb{R}^n).
 \end{align}
 Moreover, if $\beta \in \mathbb{R}$, then we define an operator $A_\beta$ by
\begin{align}
\label{E6.6}
\begin{split}
D(A_\beta)&=\left\{ u=\beta f(0)G_n + f\colon f\in W^{2,p}(\mathbb{R}^n)\right\},\\
A_\beta u&=-\beta f(0)G_n -\Delta_{n,p} f.
\end{split}
\end{align}
\end{definition}
\begin{remark}\label{rem-def-A}
Using the boundary operators $\tau_\beta$ defined in Definition \ref{def-tau_beta}, we can write the above complicated definition of $A_\beta$ in the following simple way,
\begin{align}\label{eqn-A_beta-simple}
D(A_\beta)&:=\left\{ u \in W^{\Delta,p}(\mathbb{R}^n\setminus\{0\}): \tau_\beta(u)=0 \right\},\\
A_\beta u&=-\Delta u, \;\;\; u \in D(A_\beta),
\end{align}
where $\Delta u$ is the distributional Laplacian in $\mathbb{R}^n\setminus\{0\}$ of a distribution $u \in \mathscr{D}(\mathbb{R}^n\setminus\{0\})$.
\end{remark}

In the formulation of theorem below  $ K_0 $ is the Macdonald function defined earlier in formula \eqref{eqn-Macdonald function}.

\begin{theorem} \label{T8.5} Assume that $\beta \in \mathbb{R}$ and,    $n=2$ or $n=3$.  Then the following holds.
\begin{itemize}
\item[(i)] $\left(A_\beta, D(A_\beta)\right)$ satisfies \eqref{E5.1} to \eqref{E5.4},
\[\rho(A_\beta)\supset (0,+\infty)\setminus \Lambda_\beta\]
 where, $\Lambda_\beta=\emptyset$, or, in the cases listed below in \eqref{eqn-eigenvalue},    $\Lambda_\beta=\{\lambda_ \beta\}$, where
\begin{equation}\label{eqn-eigenvalue}
\lambda_{\beta} = \begin{cases} \e^{-\frac{2}{\beta}}&\text{if $n=2$, $\beta \in \mathbb{R}$,}\\
\left( 1- \frac{4\pi}{\beta}\right)^2&\text{if $n=3$ and   $\beta \in  \adda{ (-\infty,0) \cup } (4 \pi, +\infty\adda{]} $.}
\end{cases}
\end{equation}
Moreover, in the latter cases,  we have
 \begin{equation}\label{eqn-eigenvector}
 A_\beta { e}_{n, \beta}= \lambda_ {\beta} {{e}}_{n, \beta},
 \end{equation}
where  
\begin{equation}\label{eqn-eigenvector}
\begin{aligned}
{e}_{3, \beta}(x) &=\sqrt{\lambda_\beta} G_{3, \lambda_\beta}(x)=  \frac{1}{4\pi \vert x\vert} \e^{-\sqrt{\lambda_{\beta}} \vert x\vert}, \;\; x \in \mathbb{R}^3,
\\
{e}_{2, \beta}(x)&=G_{2,\lambda_\beta}(x)= \frac{1}{2\pi}K_0(\e^{-\frac{1}{\beta}} \vert x\vert ), \;\;  x \in \mathbb{R}^2.
\end{aligned}
\end{equation}
\item [(ii)] Any operator satisfying \eqref{E5.1} to \eqref{E5.4} is equal to $(A_\beta,D(A_\beta))$ for some $\beta \in \mathbb{C}\cup\{\infty\}$.

\item[(iii)] The resolvent sets of $A_0$ and $A_\infty$ satisfy $\rho(A_0)\supset (0,+\infty)$ and $\rho(A_\infty)\supset (0,+\infty)\setminus\{1\}$.
\end{itemize}
\end{theorem}
\begin{proof} Let $\beta\in \mathbb{C}\setminus\{0\}$ and $\lambda>0$ satisfy \eqref{Elambdabeta}.  Recall that $(A_{\lambda,\beta}, D(A_{\lambda,\beta}))$ is defined by \eqref{E8.3}.  For $u\in  D(A_{\lambda,\beta})$, we have
\begin{align*}
u&= \beta f^\lambda(0) G_{n,\lambda} + f^\lambda,\\
(\lambda I +A_{\lambda,\beta})u&=\lambda  \beta f^\lambda (0)G_{n,\lambda} + \lambda f^\lambda  - \lambda \beta f^\lambda (0)G_{n,\lambda} -\Delta_{n,p} f^\lambda = (\lambda -\Delta_{n,p})f^\lambda.
\end{align*}
Therefore
$$
\lambda I+A_{\lambda,\beta} = \left(\lambda I-\Delta_{n,p}\right)\circ \Pi_p,
$$
and $\lambda \in \rho(A_{\lambda,\beta})$  since $\lambda I-\Delta_{n,p}$ is a bijection from $W^{2,p}(\mathbb{R}^n)$ to $L^p(\mathbb{R}^n)$ and its inverse is bounded, and since $\Pi_p\colon D(A_{\lambda , \beta})\to W^{2,p}(\mathbb{R}^n)$ is bijective,
$$
\Pi_p^{-1}f = \lambda \beta f(0)G_{n,\lambda}+ f, \qquad f\in W^{2,p}(\mathbb{R}^n),
$$
and, by the Sobolev embedding there is a constant $c$ such that
$$
\vert \Pi_p^{-1}f \vert_{L^p(\mathbb{R}^n)}\le \left( c \vert \beta \vert \vert G_{n,\lambda} \vert _{L^p(\mathbb{R}^n)} +1\right)\| f\|_{W^{2,p}(\mathbb{R}^n)}.
$$
Since, see Lemma \ref{L8.2}, 
$$
A_{\beta}= \begin{cases}
A_{\lambda, \beta(1+\beta \frac{\log \lambda}{2})^{-1}}&\text{if $n=2$,}\\
A_{\lambda, \sqrt{\lambda} \beta \left[ 1- \frac{\beta}{4\pi}(1-\sqrt{\lambda})\right] ^{-1}}&\text{if $n=3$,}
\end{cases}
$$
except  $\beta =-\frac{2}{\log \lambda}$ in case of $n=2$, and $\beta = 4\pi (1-\sqrt{\lambda})^{-1}$ if $n=3$, we see that $\lambda \in \rho(A_\beta)$ unless $\lambda  = \lambda_\beta$. We will be  showing now that $G_{n,\lambda_\beta}\in D(A_{\beta})$. Assume first that $n=2$. Then by Lemma \ref{L3.25} (ii), for every $\lambda >0$,
\begin{equation}\label{E8.4}
G_{2,\lambda}= G_2 + f, \qquad f\in W^{2,p}(\mathbb{R}^2).
\end{equation}
Therefore $\tau_{\beta}G_{2,\lambda}=0$ for $\beta$ such that
$$
1=\beta f(0)= \beta \left[ G_{2,\lambda}- G_2\right](0)=  \beta \lim_{x\to 0} \left[\log \frac{1}{\vert \sqrt {\lambda}x\vert}- \log\frac{1}{\vert x\vert}\right] = -\frac{\beta}{2} \log \lambda.
$$
Therefore $G_{2,\lambda_\beta }\in D(A_\beta)$ and, with notation from \eqref{E8.4},
$$
(\lambda_\beta  + A_\beta) G_{2,\lambda_\beta} =\lambda_\beta G_{2,\lambda_\beta} - \beta f(0) G_{2} - \Delta_{n,p}f = \lambda _\beta G_{2,\lambda_\beta}-G_2 -\Delta\left(G_{2,\lambda_\beta}-G_2\right)=0.
$$
The case of $n=3$ can be treated in the same way. The first part of the theorem is proven.
\end{proof}

\begin{proof}[Proof of the 2nd part of the theorem]

Let $\lambda>0$. Then, see the prove of Theorem \ref{Th5.7}, any operator  $A$ satisfying  \eqref{E5.1} to \eqref{E5.4} is equal to $(A_{\gamma, \beta},D(A_{\gamma,\beta}))$ or  $-A$ is the Friedrichs Laplacian, that is $A=A_0$. Therefore $A=A_0$, or  is equal to $A_\beta$ for some $\beta \in \mathbb{C}\setminus\{0\}$, or, see Lemma \ref{L8.3},  $A=A_\infty$ if  $\beta =\frac{2}{\log\lambda}$ and  $n=2$ or   $\beta = -4\pi \sqrt{\lambda} \left( 1-\sqrt{\lambda}\right)^{-1}$ and $n=3$.
\end{proof}

\begin{proof}[Proof of the 3rd part of the theorem]
This proof is simple.  Indeed,  $A_\infty= A_{\gamma,\beta}$ for $\gamma\not =1$, and hence $\rho(A_\infty)\supset(0,+\infty)\setminus \{1\}$. Since $G_n\in D(A_\infty)$ we have $A_\infty G_n= -G_n$. The case of $A_0=-\Delta_{n,p}$ is obvious.
\end{proof}

\begin{remark}\label{rem-8.7}
If $n=2$, then for every  $\lambda >0$, $G_{n,\lambda}$ is a fundamental solution of the equation $\Delta u=\lambda u$, that is
$$
 \lambda G_{n,\lambda}-\Delta G_{n,\lambda}= \delta_0.
$$
Therefore, for every $f\in L^p(\mathbb{R}^n)$,
$$
G_{n,\lambda}\ast f= (\lambda -\Delta _{n,p})^{-1}f.
$$
A similar assertion does not hold  when $n=3$.
\end{remark}
\section{Heat evolutions and boundary problems}
In this section $n=2,3$. Recall that $(A_{\beta}, D(A_\beta))$, $\beta \in \mathbb{C}\cup \{\infty\}$, are operators defined in the previous sections. Each $A_\beta$ is an  operator on $L^p(\mathbb{R}^n)$ where $p\in (1,+\infty)$ if $n=2$ and $p\in (\frac{3}{2},3)$ if $n=3$.  Since each $A_\beta$ is  self-adjoint on $L^2(\mathbb{R}^n)$ with resolvent set $\rho(A_\beta)\supset (\lambda(\beta),+\infty)$, $-A_\beta$   generates a $C_0$ analytic semigroup $S_\beta$  on $L^2(\mathbb{R}^n)$. We will see later that $S_\beta$ is given be a kernel $\mathscr{G}_\beta(t,x,y)$, that is
$$
S_\beta (t)u(x)= \int_{\mathbb{R}^n} \mathscr{G}_\beta (t,x,y) u(y)\d y.
$$

Obviously, as $-A_0$ is the Friedrichs Laplacian, it is well defined on any $L^p(\mathbb{R}^n)$-space  $p\in (1,+\infty)$ also in dimension $n=3$.

Let us define the following \emph{generalized Dirichlet map}
\begin{equation}\label{eqn-Dirichlet map}
\Dir \colon \mathbb{R} \ni \gamma  \mapsto   \gamma G_n \in W^{\Delta,p}(\mathscr{O}).
\end{equation}
Note that the map $\Dir$ is independent from $\beta \in \mathbb{C}$. Let us recall that $\tau_\beta$ has been defined in \eqref{eqn-tau_beta}. The following result justifies the name we have used for the map $\Dir$.
\begin{proposition}\label{P91}
Assume that $\beta \in \mathbb{C}$. Then for every $\gamma \in \mathbb{C}$, the function $u=\Dir \gamma=\gamma G_n$ solves in the class $W^{\Delta,p}(\mathscr{O})$ the boundary problem
\begin{align*}
 \Delta_\mathscr{O} u(x)&= u(x)\quad  \text{for $x\in \mathscr{O}$,}\\
 \tau_\beta u&= \gamma.
 \end{align*}
\end{proposition}
\begin{proof} Obviously we have $\tau_\beta \gamma G_n= \gamma$. Fact that $\Delta_\mathscr{O} G_n=G_n$ follows from Lemma \ref{L3.3} or more directly Lemma \ref{L3.4}.
\end{proof}
\begin{remark}
Note that the Dirichlet problem
\begin{align*}
\Delta_{\mathscr{O}} u(x)&=\lambda u(x)\quad  \text{for $x\in \mathscr{O}$,}\\
\tau_\infty u&= \gamma,
\end{align*}
has no solution except $u\equiv 0$ in case of $\gamma =0$. Indeed, if $u=c_0G_n+ f$, $f(0)=-\gamma\not =0$ then
\begin{align*}
\Delta _{\mathscr{O}}u&=\lambda u
\end{align*}
if and only if $c_0=\lambda c_0$ and $\Delta_{n,p}f=f$. Since the Friedrichs Laplace operator $\Delta_{n,p}$ has empty point spectrum we have $f\equiv 0$. Therefore the boundary problem corresponding to $\beta=\infty$ is not well-posed.
\end{remark}

\begin{proposition}\label{P93}
Let $\beta \not =0$. Assume that:
\begin{itemize}
 \item[(i)]  For every  $t>0$  there exists  an integrable function $\eta \colon \mathscr{O} \mapsto (0,+\infty)$ such that
$$
0\le \frac{ \mathscr{G}_\beta(t,x,y) }{G_n(x)}\le \eta(y) \;\;\;\mbox{for all  $x\colon \vert x\vert \le 1$ and $y\in\mathscr{O}$}.
$$
\item[(ii)]  For every  $t>0$   the following  limit exists
$$
R_\beta(t,y):= \lim_{x\to 0} \frac{ \mathscr{G}_\beta(t,x,y) }{G_n(x)}.
$$
\end{itemize}
Then for every $t>0$ we have
$$
\left(I+ A_\beta\right)S_\beta  (t)\Dir\gamma  = \frac{\gamma}{\beta}R_\beta (t,\cdot).
$$
\end{proposition}
\begin{proof}
Let $v\in L^{ 2}(\mathbb{R}^n)$. Since the semigroup is analytic and self adjoint, and since $\Delta _\mathscr{O}G_n=G_n$, the Green formula yields
\begin{align*}
\left( \left(I+ A_\beta \right)S_\beta (t)\Dir \gamma, v\right)&= \gamma \left( G_n, (I+A_\beta)S_\beta(t)v\right)\\
&= \gamma \left( G_n, S_\beta(t)v\right)- \gamma \left( G_n, \Delta_\mathscr{O} S_\beta(t)v\right)\\
&= \gamma \left( \Delta_\mathscr{O} G_n, S_\beta(t)v\right)- \gamma \left( G_n, \Delta_\mathscr{O} S_\beta(t)v\right)\\
&= \gamma \overline{g(t,0)},
\end{align*}
where $g\in W^{\Delta,2}(\mathbb{R}^n)$ is such that
$$
S_\beta(t)v(y)= b_0(t)G_n(y)+ g(t,y).
$$
The proof will be completed as soon as we will show that
$$
g(t,0)= \int_{\mathbb {R}^n}R_\beta(t,y)v(y)\d y.
$$
We have
$$
b_0(t) = \lim_{x\to 0} \frac{S_\beta(t)v(x)}{G_n(x)}.
$$
Since, the semigroup is analytic $S_\beta (t)v\in D(A_\beta)$. Consequently, we have
$$
b_0(t)= \beta g(t,0).
$$
Therefore
$$
g(t,0)= \frac{1}{\beta} \lim_{x\to 0} \int_{\mathbb{R}^n}\frac{\mathscr{G}_\beta(t,x,y)}{G_n(x)}v(y)\d y
$$
and the desired conclusion follows from the Lebesgue dominated convergence theorem.
\end{proof}

If  $\beta =0$, then by \eqref{E6.6} we have $D(A_0)= W^{2,p}(\mathbb{R}^n)$. Therefore $-A_{0}$ is the Friedrich Laplacian  on  $\mathbb{R}^n$. In this case $p\in (1,+\infty)$ also if $n=3$. The corresponding semigroup $S:=S_{0}$ is the heat semigroup on $\mathbb{R}^n$ with Green kernel $G_0(t,x,y):= P(t,x-y)$, where
\begin{equation}\label{E91}
P(t,x):= P(t,\vert x\vert ):=\left(4\pi t\right)^{-\frac{n}{2}}\e^{-\vert x\vert ^2/4t}.
\end{equation}

\begin{proposition}\label{P94}
For every $t>0$ we have
$$
\left[\left(I+ A_0\right)S_0  (t)\Dir\gamma \right](x) = P(t,x)= S_0(t)\delta _0(x), \qquad x\in \mathscr{O}.
$$
\end{proposition}
\begin{proof}
We will follow  the proof of the previous proposition. Let $v\in L^{ 2}(\mathbb{R}^n)$. Since the semigroup is analytic and self adjoint, and since $\Delta _\mathscr{O}G_n=G_n$, the Green formula yields
\begin{align*}
\left( \left(I+ A_0\right)S_\beta (t)\Dir \gamma, v\right)&= \gamma \left( G_n, (I+A_0)S_0(t)v\right)\\
&= \gamma \left( G_n, S_0(t)v\right)- \gamma \left( G_n, \Delta_\mathscr{O} S_0(t)v\right)\\
&= \gamma \left( \Delta_\mathscr{O} G_n, S_0(t)v\right)- \gamma \left( G_n, \Delta_\mathscr{O} S_0(t)v\right)\\
&= \gamma \overline{g(t,0)},
\end{align*}
where $g\in W^{\Delta,2}(\mathbb{R}^n)$ is such that
$$
S_0(t)v(y)= b_0(t)G_n(y)+ g(t,y).
$$
Clearly if $\beta =0$, then $b_0\equiv 0$, and  $g(t,y)= S_0(t)\widetilde v(y)$. Therefore
$$
g(t,0)= \int_{\mathbb{R}^n}P(t,y)v(y)\d y.
$$
\end{proof}
\subsection{Stochastic boundary problem}
Let $W=\left(W(t):t\geq 0\right)$ be a real valued Wiener process defined on a probability space $(\Omega,\mathfrak{F},\mathbb{P})$. We are concerned with the boundary problem
\begin{equation}\label{E92}
\begin{aligned}
\frac{\partial u}{\partial t}&= \Delta u,\\
\tau_\beta u&= \frac{\d W}{\d  t},\\
u(0)&=u_0.
\end{aligned}
\end{equation}
where we assume that $\beta\geq 0$. Let us recall that $\tau_\beta$ is the boundary operator defined earlier.

Taking into account \cite{Brzezniak-Goldys-Peszat-Russo}, \cite{Goldys-Peszat},   its mild formulation has the form
$$
\begin{aligned}
u(t)&= S_\beta (t)u_0+ \left(I+A_\beta\right) \int_0^t S_\beta(t-s) \Dir \d W(s) \\
&= S_\beta (t)u_0+ \int_0^t \left(I+A_\beta\right) S_\beta(t-s) G_n \d W(s),  \qquad t \geq 0.
\end{aligned}
$$

\begin{corollary}\label{C95}
Assume that $\beta \not =0$. Then under that assumptions of Proposition \ref{P93} we have
 $$
 u(t,y)=[S_\beta (t)u_0](y)+\frac 1\beta  \int_0^t R_\beta (t-s, y)\d W(s), \qquad t \geq 0, \; y \in \mathbb{R}^n\setminus\{0\}.
 $$
\end{corollary}
\subsection{Case $\beta =0$ for $n=2,3$} Recall, see Proposition \ref{P94} that in this case we have
$$
\left(I+A_0\right)\int_0^t S_0(t-s)\mathscr{D} \d W(s)= \int_0^t S_0(t-s)\delta_0 \d W(s),
$$
where $S_0$ is the heat semigroup with the kernel $P$ defined by \eqref{E91}.

\begin{proposition}
Let $\beta =0$.  If $n=2$ then for any $p\in (1,2)$ the stochastic boundary problem \eqref{E92} is well posed and defines gaussian Markov family in $L^p(\mathbb{R}^2)$-space. If  $n=3$ then the stochastic boundary problem \eqref{E92} is well posed and defines gaussian Markov family in $L^p(\mathbb{R}^3)$-space only for  $p\in (1,\frac{3}{2})$.
\end{proposition}
\begin{proof}
Assume that $p \in (1,+\infty)$.  Since for $t>s\geq 0$, 
$$
\left( S_0(t-s)\delta_0\right)(y)= (4\pi (t-s))^{-\frac{n}{2}}\e^{-\frac{\vert y\vert ^2}{4(t-s)}}, \;\;  y \in \mathbb{R}^n\setminus\{0\}, 
$$
the problem is well-posed in $L^p(\mathbb{R}^n)$  if and only if, see \cite{Brz+Ver_2012},  for all, or equivalently some, $t>0$,  $I(t,p)<+\infty$, where
$$
I(t,p):= \int_{\mathbb{R}^n} \left( \int_0^t (4\pi s)^{-n} \e^{-2\frac{\vert y\vert ^2}{4s}}\d s\right)^{p/2}\d y.
$$
Note that $I(t,p)<+\infty$ if and only $J(p)<+\infty$ where
\begin{align*}
J(p)&:=\int_{\{y\in \mathbb{R}^n\colon \vert y\vert \le 1\}} \left( \int_0^{+\infty} (4\pi s)^{-n} \e^{-s-\frac{\vert \sqrt{2}y\vert ^2}{4s}}\d s\right)^{p/2}\d y =\int_{\{y\in \mathbb{R}^n\colon \vert y\vert \le 1\}} \left(  G_{2n}( \sqrt{2}y )\right)^{p/2}\d y.
\end{align*}
Since, see \eqref{eqn-G_n-asymptotics}, for small  $y$,
$$
G_{2n}(y)\approx \frac{\Gamma(n-1)}{4\pi} \vert y\vert ^{-2n+2},
$$
we have
$$
\int_{\{y\in \mathbb{R}^n\colon \vert y\vert \le 1\}} \vert y\vert ^{-np+p}\d y \sim C \int_0^1 r^{-np+p+n-1}\d r.
$$
Therefore $I(t,p)<+\infty$ if and only if $-np+p+n>0$.
\end{proof}
\begin{remark}
Since
$$
\sup_{t>0} I(t,p) = \int_{\mathbb{R}^n} \left( \int_0^{+\infty} s^{-n}\e ^{-\frac{\vert y\vert ^2}{4s}}\d s\right)^{p/2} \d y = C \int_{\mathbb{R}^n} \vert y\vert ^{p(n-1)} \d y  =+\infty,
$$
the invariant measure for the Markov family defined by \eqref{E92} does not exist.
\end{remark}
\begin{remark}
Let
$$
X(t,x):= \int_0^t S_0(t-s)\delta_0 \d W(s)(x)=\int_0^t \int_{\mathbb{R}^n} P(t-s,x) \d W(s), \qquad x\in \mathbb{R}^n\setminus \{0\}.
$$
Clearly $X(t,\cdot )$ is a gaussian random field on $\mathscr{O}$ , with mean zero and covariance
\begin{align*}
\Gamma(t,x,y)&:= \mathbb{E}\, X(t,x)X(t,y)= \int_0^{t}  P(s,x)P(s,y)\d s\\
&= \left(4\pi \right)^{-n} \int_0^t s^{-n} \e^{-\frac{\vert x\vert ^2 + \vert y\vert ^2}{4s}}\d s, \;\;  x.y \in \mathbb{R}^n\setminus\{0\}.
\end{align*}
The limiting ``covariance'' is given by 
$$
\Gamma (x,y)=\left(4\pi \right)^{-n}\int_0^{+\infty}s^{-n} \e^{-\frac{\vert x\vert ^2 + \vert y\vert ^2}{4s}}\d s= C_n \left( \vert x\vert ^2+\vert y\vert ^2\right)^{-n+1},
$$
where 
$$
C_n:= \left(4\pi \right)^{-n}\int_0^{+\infty}s^{-n} \e^{-\frac{1}{4s}}\d s. 
$$
\end{remark}

\begin{proposition}
Let $\beta =0$ and $\ell >n/2$. Then the stochastic boundary problem \eqref{E92} is well posed and defines a gaussian Markov family in the Sobolev space  $H^{-l}(\mathbb{R}^n)$-space. Moreover, this family  has an   invariant (gaussian) measure $\mu$ on $H^{-l}(\mathbb{R}^n)$ if and only if $n=3$.
\end{proposition}
\begin{proof} Since $(I-\Delta)^{-l/2}$ is an isometric isomorphism  between Hilbert spaces $H^{-l}(\mathbb{R}^n)$ and  $L^2(\mathbb{R}^n)$, the desired conclusion follows because there exist constants $C,C^\prime>0$ such that  for every $t>0$
\begin{align*}
\mathbb{E}\left\vert (I-\Delta)^{-l/2} \int_0^t S(t-s) \delta_0\d W(s)\right\vert _{L^2(\mathbb{R}^n)}^2
&=
\int_0^{t} \int_{\mathbb{R}^n} \left( (I-\Delta)^{-l/2} P(s,\cdot)(x)\right)^2 \d s\d x\\
&= C \int_{\mathbb{R}^n} \int_0^{t} (1+ \vert x\vert^2)^{-l} \e^{-2s\vert x\vert ^2}\d s\d x\\
&\le  Ct  \int_{\mathbb{R}^n} (1+ \vert x\vert^2)^{-l}\d x\\
&\le  C't  \int_0^{+\infty} (1+r^2)^{-l} r^{n-1} \d r <+\infty.
\end{align*}
Next
\begin{align*}
\sup_{t>0} \mathbb{E}\left\vert (I-\Delta)^{-l/2} \int_0^t S(t-s) \delta_0\d W(s)\right\vert _{L^2(\mathbb{R}^n)}^2
&= C \int_{\mathbb{R}^n} \int_0^{+\infty} (1+ \vert x\vert^2)^{-l} \e^{-2s\vert x\vert ^2}\d s\d x\\
&= C \int_{\mathbb{R}^n} (1+ \vert x\vert^2)^{-l} \frac{1}{2\vert x\vert ^2}\d x\\
& = C' \int_0^{+\infty} (1+r^2)^{-l} r^{-2+ n-1} \d r.
\end{align*}
The integral converges outside the unit ball. It converges inside if and only if $n-2>0$.
\end{proof}
\subsection{The case of $\beta\not =0$, the first approach}As in the case of $\beta=0$, we can treat boundary value problem  on some  Sobolev space of negative order. In fact,  let us define  $H^{-l}_\beta(\mathscr{O})$ as a closure of $L^2(\mathscr{O})$ equipped with the norm
$$
\left\vert \psi\right\vert_{H^{-l}_\beta(\mathscr{O})}^2 := \left\vert (I-A_\beta)^{-l/2}\psi\right\vert ^2_{L^2(\mathscr{O})}.
$$
\begin{proposition}
Let $\beta \not =0$ and $l>1$. Then the stochastic boundary problem \eqref{E92} is well posed and defines a gaussian Markov family in $H^{-l}_\beta(\mathbb{R}^n)$-space.
\end{proposition}
\begin{proof} We have
\begin{align*}
\mathbb{E} \left\vert \int_0^t  \left(I- A_\beta\right)S_\beta (t-s)\Dir \d W(s)\right\vert ^2_{H^{-l}_\beta (\mathscr{O})}&= \int_0^t \left\vert \left(I- A_\beta\right)^{1-l/2}S_\beta(s)G_n\right\vert ^2_{L^2(\mathscr{O})}\d s\\
&\le C \int_0^t s^{-2+l}\d s<+\infty.
\end{align*}
\end{proof}
\begin{remark}
Let us note that the Markov family defined by the boundary problem \eqref{E92} has  an  invariant probability measure in $H^{-l}_\beta (\mathscr{O})$ if and only if
\begin{equation}\label{eqn-infty}
\sup_{t>0} \mathbb{E}\left\vert  \int_0^t\left(I- A_\beta\right) S_\beta (t-s)\Dir \d W(s)\right\vert ^2_{H^{-l}_\beta (\mathscr{O})}<+\infty.
\end{equation}
\end{remark}
\begin{remark}
We have
\begin{align*}
\sup_{t>0} \mathbb{E}\left\vert  \int_0^t\left(I- A_\beta\right) S_\beta (t-s)\Dir \d W(s)\right\vert ^2_{H^{-l}_\beta (\mathscr{O})}&=\sup_{t>0}  \mathbb{E}\left\vert \int_0^t R_\beta(t-s)\d W(s)\right\vert ^{2}_{H^{-l}_\beta(\mathbb{R}^n)}\\
&= \int_{\mathbb{R}^n}\int_0^{+\infty}  \left[ (I-\Delta)^{-l/2} R_\beta (s,\cdot)(x)\right]^2\d s\d x.
\end{align*}
\end{remark}

\section{Stochastic boundary problem on $\mathbb{R}^3 \setminus \{0\}$; case $\beta \not =0$}
\label{sec-SDE-3}

 Taking into account Corollary \ref{C95}, Proposition \ref{P93} and Lemma \ref{L102}, we can rewrite \eqref{E92} as follows
 \begin{equation}\label{E7.3}
 u(t,y)=S_\beta (t)u_0(y)+\frac 1\beta  \int_0^t R_\beta (t-s, y)\d W(s), \;\;t \geq 0.
 \end{equation}

\begin{proposition}\label{P7.1}
Let $\beta \not =0$ and $p\in (\frac{3}{2},3)$. Then the stochastic boundary problem \eqref{E92} is well posed and defines a gaussian Markov family in the $L^p(\mathbb{R}
^3)$-space.
\end{proposition}
\begin{proof}
We have alreaday proved that if $p\in (\frac{3}{2},3)$ then the $S_\beta$ is a C$_)$-semigroup on the space $L^p(\mathbb{R}
^3)$.

Given $p>1$, by applying the Burkholder inequality  we infer that
 $$
 \mathbb{E} \int_{\mathbb {R}^3} \left\vert \int_0^t R_\beta (t-s, y)\d W(s)\right\vert ^p \d y <+\infty
$$
 if and only if
 \begin{equation}\label{E7.4}
 \int_{\mathbb{R}^3} \left( \int_0^t  R^2_\beta (t-s,y)\d s\right)^{p/2}\d y<+\infty.
 \end{equation}
 Finally the estimate from part (ii) of Lemma \ref{L102} yields  that \eqref{E7.4} holds if and only if
 \begin{align*}
 &\int_{\mathbb{R}^3} \left( \int_0^t \frac{(8\pi s)^2}{ \vert y\vert ^2} P^2(s,y) \d s\right)^{p/2}\d y\\
 &\quad = \int_{\mathbb{R}^3} \left( \int_0^t \frac{(8\pi s)^2}{ \vert y\vert^2} \left( 4\pi s\right)^{-3} \e^{-\frac{\vert y\vert^2}{2s}} \d s\right)^{p/2}\d y:= I(p,t)<+\infty.
 \end{align*}
 Clearly $I(p,t)<+\infty$ for any or all $t$ if and only if $J(p)<+\infty$, where
 \begin{align*}
 J(p)&:=\int_{\mathbb{R}^3} \vert y\vert ^{-p}  \left( \int_0^{+\infty}  \left(4\pi s\right)^{-1} \e^{-s - \frac{\vert \sqrt{2}y\vert^2}{4s}}  \d s\right)^{p/2} \d y \\
 & = \int_{\mathbb{R}^3} \vert y\vert ^{-p}  \left( G_2(\sqrt{2}y)\right)^{p/2} \d y,
 \end{align*}
 where $G_2$ has been defined in \eqref{eqn-G_n}. 
 Since  $G_2$ decays exponentially at infinity, and
 $$
 G_2(y)\approx  \frac{1}{2\pi} \log \frac{1}{\vert y\vert } \mbox{ as }y \approx 0
 $$
 $J(p)<+\infty$ if and only if  $-p+2>-1$, that is $p<3$. Obviously $p>3/2$ by assumption. 
 \end{proof}

\dela{Note that, by Lemma \ref{L102}(ii),
\begin{align*}
&\lim_{t\to+\infty} \mathbb{E} \int_0^t R_\beta(s,y)\d W(s) \int_0^t R_\beta(s,z)\d W(s)= \int_0^{+\infty} R_\beta(s,y)R_\beta(s,z)\d s\\
&\ge c^2 (8\pi)^2 \int_0^{+\infty} \frac{s^2}{\vert y\vert \vert z \vert } P(s,y)P(s,z)\d s  =+\infty,
\end{align*}
as
$$
\int_0^{+\infty} s^{-1} \e^{-\frac 1{s}}\d s =+\infty.
$$
\coma{Why is this the end of the proof?}
}

\appendix

\section{Comparison with earlier approaches, especially the one described in \cite{ABD_1995-a}}\label{Section7}

In this section, $n=2,3$ and  $\mathscr{O}:=\mathbb{R}^3\setminus\{0\}$. Moreover, if $n=3$, we assume that $\frac{3}{2}< p < 3$. From now on  $q>1$ is such that  $\frac{1}{p}+\frac{1}{q}= 1$. Note that if $\frac{3}{2}<p<3$, then also $\frac{3}{2}<q<3$.

\begin{remark}\label{rem-comparison}
\begin{proof}[Case $n=3$]
Using the notation from \cite{ABD_1995-a}, if $\alpha \in (-\infty,+\infty]$, then the operator $A_\alpha$ from \cite{ABD_1995-a} is equal to the  operator $A_\beta$ from the present paper, if and only if
\begin{equation}\label{E6.9}
\beta=8\pi\Gamma^\alpha(\coma{1})=\begin{cases} 32\pi^2(\sqrt{\coma{1}\dela{\pi}}+4\pi \alpha)^{-1} & \mbox { if } \alpha \in (-\infty,+\infty) \setminus\{\frac{-1}{4\pi} \}  \\
+ \infty &\mbox { if }  \alpha=\frac{-1}{4\pi},\\
 0 &\mbox { if } \alpha=+\infty;
\end{cases}
\end{equation}
i.e.
\begin{equation}\label{E6.10}
\alpha=\begin{cases} \frac{8\pi}\beta-\frac1{4\pi}  & \mbox { if } \beta  \in (-\infty,+\infty) \setminus\{0\}, \\
+\infty &\mbox { if } \beta=0, \\
-\frac1{4\pi}  &\mbox { if } \beta=+\infty.
\end{cases}
\end{equation}
Indeed, if $u \in D(A_\alpha)$, then there exists $h \in L^2(\mathbb{R}^n)$ such that $u=(\lambda+A_\alpha)^{-1}h$. Thus, by formula (2.2) from \cite{ABD_1995-a} we infer that
\[
u(x)=[R_\lambda \ast h](x)+ \Vert R_\lambda\Vert_{L^2}^{-2} \Gamma^\alpha(\lambda)  R_\lambda(x) \int_{\mathbb{R}^n} R_\lambda(y) h(y) \, \d y, \;\; x \in \mathbb{R}^2,
\]
where $R_\lambda$ is a solution to $(\lambda-\Delta)R_\lambda=\delta_0$ in $\mathbb{R}^n$. Note that $R_1=G_n$.

Note that $R_\lambda \ast h \in W^{2,2}(\mathbb{R}^n)$ and, since $ R_\lambda$ is a symmetric function,
\[
 \int_{\mathbb{R}^n} R_\lambda(y) h(y) \, \d y=  \int_{\mathbb{R}^n} R_\lambda(0-y) h(y) \, \d y=[R_\lambda \ast h](0).
\]
Therefore,
\[
u(x)=[R_\lambda \ast h](x)+ \Vert R_\lambda\Vert_{L^2}^{-2} \Gamma^\alpha(\lambda) R_\lambda(x) [R_\lambda \ast h](0)  \;\; x \in \mathbb{R}^2.
\]

Since $n=3$, by formula \eqref{eqn-L^2-norms of G_3} in  Lemma \ref{lem-L^2-norms of G_2 and G_3},  $\Vert R_1\Vert_{L^2}^{-2}=\Vert G_3\Vert_{L^2}^{-2}=8\pi $ and so
\[
u(x)=[R_1 \ast h](x)+ 8\pi  \Gamma^\alpha(\lambda) R_1(x) [R_1 \ast h](0)=[G_2 \ast h](x)+ 8\pi  \Gamma^\alpha(1) G_3(x) [G_3 \ast h](0).
\]
Hence we infer that $u \in D(A_\beta)$ with $\beta$ given by the first part of \eqref{E6.9}.

In particular $\beta=0$ corresponds to $\alpha=\infty$, $\beta=\infty$ corresponds to $\alpha=-\frac1{4\pi}$,   and vice versa. The choice we made that $\beta=\infty$ when  $\alpha=-\frac1{4\pi}$ is arbitrary, i.e.
\begin{equation}\label{eqn-beta=infty}
  \lim_{\alpha \to -\frac1{4\pi}+} 4\pi(\sqrt{\coma{1}\dela{\pi}}+4\pi \alpha)^{-1}.
\end{equation}
Note also that the eigenvalue $-\lambda_{\beta}$, see  \eqref{eqn-eigenvalue}, satisfies
\begin{equation}\label{eqn-eigenvalue-d=3}
-\lambda_{\beta} = -(4 \pi \alpha )^2.
\end{equation}

\end{proof}

\begin{proof}[Case $n=2$] Note that by  formula \eqref{eqn-L^2-norms of G_2} in Lemma \ref{lem-L^2-norms of G_2 and G_3},   $\Vert R_1\Vert_{L^2}^{-2}=\coma{2}$. Thus, we infer
\[
u(x)=[R_1 \ast h](x)+ \coma{2}  \Gamma^\alpha(\lambda) R_1(x) [R_1 \ast h](0)=[G_2 \ast h](x)+ \coma{2}  \Gamma^\alpha(1) G_2(x) [G_2 \ast h](0).
\]

Hence, by the definition \eqref{E6.6} of the operator $A_\beta$,   we infer that $u \in D(A_\beta)$ with
\[
\beta= \coma{2}  \Gamma^\alpha(1)
\]
By formula (2.7) in \cite{ABD_1995-a} we deduce that
\begin{equation}\label{eqn-b(alpha)-n=2}
\beta= \coma{8\pi} \alpha^{-1}.
\end{equation}
We conclude, that using the notation from \cite{ABD_1995-a}, if $n=2$  $\alpha \in (-\infty,+\infty]$, then the operator $A_\alpha$ from \cite{ABD_1995-a} is equal to the  operator $A_\beta$ from the present paper, if and only if condition \eqref{eqn-b(alpha)-n=2} holds.
\end{proof}

Let us finish this remark by  pointing out two calculations type errors in the paper \cite{ABD_1995-a}. Firstly, in formula (2.13) in \cite{ABD_1995-a}  one should have  $-\e^{-\alpha}$ and not, as it has been printed, $-4\e^{-2\alpha}$,  i.e. the correct version is
\begin{equation}\label{eqn-2.13-ABD}
\sigma_p (A^\alpha) =
\begin{cases}
\bigl\{ - \dela{4} \e^{-\dela{2} \alpha } \bigr\} , & \mbox{ for } \alpha \in \mathbb{R},\;\;\; n = 2 ,  \\
\bigl\{ - (4\pi \alpha)^2 \bigr\} , & \mbox{ for } \alpha < 0,\;\;\; n = 3,  \\
\emptyset  ,  & \mbox{ otherwise}. \\
\end{cases}
\tag{(2.13)}
\end{equation}

 In formula (2.15) in \cite{ABD_1995-a} instead of  $G_{4\e^{-2\alpha}}$ we should have $G_{\e^{-\alpha}}$. Thus the correct  version of formula  (2.15) in \cite{ABD_1995-a} is
\begin{equation}\label{eqn-2.15-ABD}
\tag{(2.15)}
\begin{aligned}
\Psi^{\alpha}(x)&= G_{\e^{-\alpha}}(x) \mbox{ for } n=2, \; \alpha \in \mathbb{R},
\\
\Psi^{\alpha}(x) &= G_{(4 \pi \alpha )^2 } (x), \; \mbox{ for } n=2, \; \alpha<0.
\end{aligned}
\end{equation}
There is also another  misprint in  \cite{ABD_1995-a}. Formula (2.7) therein    states that  for $n=3$
\[\Gamma^\alpha(\lambda)=4\pi(\sqrt{\pi}+4\pi \alpha)^{-1}.
\]
The correct version of that formula should be
\begin{equation}\label{eqn-2.7-ABD_1995-a}
\Gamma^\alpha(\lambda)=4\pi(\sqrt{\lambda}+4\pi \alpha)^{-1}.
\end{equation}
The correct formula was used later on, see the bottom of page 227 of \cite{ABD_1995-a}. The correct formula was also used in paper \cite{ABD_1995-b}, see formula (1.3) in \cite{ABD_1995-b}.

\end{remark}

From now we will assume that $n=3$ and
    \begin{equation}\label{eqn-beta=not=0}
{    \beta \in \mathbb{R} \setminus{0}}.
    \end{equation}
Let $S_\beta=\bigl(S_\beta(t)\bigr)_{t\geq 0}$ be the $C_0$-analytic semigroup generated by $-A_\beta$. It is known, see \cite{ABD_1995-b},  then $S_\beta$ is $C_0$ on  $L^p(\mathbb{R})$-space only if $p\in (\frac{3}{2}, 3)$. Moreover, $S_\beta(t)$ for $t>0$,  is an integral operator with  the integral kernel $\mathscr{G}_\beta(t,x,y)$, given by
\begin{align}\label{E6.13}
\mathscr{G}_\beta(t,x,y)&= P(t,x-y)+ \frac{2t}{\vert x\vert \vert y\vert} P(t,\vert x\vert +\vert y\vert )\\&-  \frac{8\pi \alpha t}{\vert x\vert \vert y\vert}\int_0^{+\infty} \e^{-4\pi \alpha u}P(t,u+\vert x\vert +\vert y\vert )\d u, \qquad x,y \in \mathbb{R}^3,
\nonumber
\end{align}
where $\alpha$ is given by formula \eqref{E6.10}. Let us emphasize here that although the case $\alpha=\frac{-1}{4\pi}$ plays a special r\^ole in  \eqref{E6.9} and \eqref{E6.10}, this is no longer the case of the formula \eqref{E6.13} above.

\bigskip
Note that
$$
\lim_{\alpha \to +\infty} \frac{8\pi \alpha t}{\vert x\vert \vert y\vert}\int_0^{+\infty} \e^{-4\pi \alpha u}P(t,u+\vert x\vert +\vert y\vert )\d u= \frac{2t}{\vert x\vert \vert y\vert} P(t,\vert x\vert +\vert y\vert ),
$$
and hence
$$
\lim_{\beta  \to 0+} \mathscr{G}_{\beta}(t,x,y)= \mathscr{G}_{0}(t,x,y)= P(t,x-y).
$$

\dela{\textbf{Question} Is it the above true for $\lim_{\alpha \to -\infty}$ and consequently for  $\lim_{\beta  \to 0-}$? NIE ROZUMIEM $\alpha \to -\infty$, to $\beta =-\frac{1}{4p}$!!
}

\begin{lemma}\label{L102} Assume that \eqref{eqn-beta=not=0}  holds.
 The the following assertions are satisfied.  \\
 $(i)$  For every  $t>0$  there exists  an integrable function $\eta \colon \mathscr{O} \mapsto (0,+\infty)$ such that
$$
0\le \frac{ \mathscr{G}_\beta(t,x,y) }{G_3(x)}\le \eta(y) \;\;\;\mbox{for all  $x\colon \vert x\vert \le 1$ and $y\in\mathscr{O}$}.
$$
$(ii)$  For every  $t>0$  there exists the following  limit exists
\begin{align}
\label{eqn-R_beta}
\begin{split}
R_\beta(t,y)&:= \lim_{x\to 0} \frac{ \mathscr{G}_\beta(t,x,y) }{G_3(x)}\\
&=  \frac{8\pi t}{\vert y\vert} P(t, y ) - \frac{32\pi^2 \alpha t}{\vert y\vert}\int_0^{+\infty} \e^{-4\pi \alpha u}P(t,u +\vert y\vert )\d u.
\end{split}
\end{align}
$(iii)$
Moreover, for every $T\in (0,+\infty)$ there is a constant $c=c(\beta,T)>0$ such that
\begin{align}
\label{eqn-R_beta-2}
c \frac{8\pi t}{\vert y\vert} P(t, y) \le  R_\beta (t,y) \le \frac{8\pi t}{\vert y\vert} P(t, y ), \qquad t\in (0,T],\ y\in \mathscr{O}.
\end{align}
\end{lemma}
\begin{proof}[Proof of part (i)] Let us choose and fix $t>0$ and $x,y \in \mathbb{R}^3$. Since
\begin{align*}
\frac{8\pi \alpha t}{\vert x\vert \vert y\vert}\int_0^{+\infty} \e^{-4\pi \alpha u}P(t,u+\vert x\vert +\vert y\vert )\d u&\le \frac{8\pi \alpha t}{\vert x\vert \vert y\vert}\int_0^{+\infty} \e^{-4\pi \alpha u}\d u P(t,\vert x\vert +\vert y\vert )\\
&\le \frac{2t}{\vert x\vert \vert y\vert} P(t,\vert x\vert +\vert y\vert ),
\end{align*}
we infer that
\begin{equation}\label{eqn-9.201}
\mathscr{G}_\beta(t,x,y)\ge P(t,x-y)\ge 0.
\end{equation}

Next let us assume that additionally  $\vert x\vert \le 1$ and $y\in\mathscr{O}$. Then, by \eqref{E6.13} we infer that
\begin{align}\label{eqn-9.202}
\begin{split}
0\le \frac{ \mathscr{G}_\beta(t,x,y) }{G_3(x)} &\le 4\pi \vert x\vert \e^{\vert x\vert} \left( P(t,x-y)+ \frac{2t}{\vert x\vert \vert y\vert} P(t,\vert x\vert +\vert y\vert ) \right)\\
&\le 4\pi \e \left( \sup_{z\colon \vert z\vert \le 1} P(t,z-y)+ \frac{2t}{\vert y\vert} P(t,y)\right)\\
&\le 4\pi \e \left(  \left(4\pi t\right)^{-\frac{3}{2}}\mathds{1}_{\{y\colon \vert y\vert \le 1\}} + P(t,\vert y\vert -1)+ \frac{2t}{\vert y\vert} P(t,y)\right),
\end{split}
\end{align}
what completes the proof of part (i). \end{proof}

\begin{proof}[Proof of part (ii)]

The first claim of $(ii)$ is obvious. In view of \eqref{eqn-R_beta}, the right inequality in \eqref{eqn-R_beta-2} is also obvious. In order to prove inequality in \eqref{eqn-R_beta-2} we choose and  fix $T>0$
and then we note that
$$
0\le P(t,u+\vert y\vert )\le P(t,y)\e^{-\frac{u^2}{ 4T}},
$$
and hence
\begin{equation}\label{eqn-9.101}
\int_0^{+\infty} \e^{-4\pi \alpha u}P(t,u +\vert y\vert )\d u \le P(t,y) \int_0^{+\infty} \e^{\coma{-}4\pi \vert \alpha\vert  u-\frac{u^2}{4T}}\d u.
\end{equation}
Consequently, it is enough to observe that
\begin{equation}\label{eqn-9.102}
c(\alpha,T):= 1- 4 \pi \alpha \int_0^{+\infty} \e^{\coma{-}4\pi \vert \alpha\vert  u-\frac{u^2}{4T}}\d u >0.
\end{equation}
\end{proof}

\section{Density in Sobolev spaces}

This section is based on the fifth section of Chapter 3 of the monograph \cite{Adams+Fournier_2003}, pp. 70--77.

\begin{definition}
Assume that $n \in \mathbb{N}^\ast$,  $m>0$ and  $p \in (1,+\infty)$. Let $q$ be the conjugate number to $p$, i.e. $q\in (1,+\infty)$ such that $\frac1p+\frac1{q}=1$. \dela{\coma{In the book I think only natural numbers $m$ are considered. It would be interesting to see how this works for general $m$.}} A closed subset $F \subset \mathbb{R}^n$ is called $(m,q)$-\emph{polar} if and only if
\begin{equation}\label{eqn-mp'-polar}
W^{-m,q}(\mathbb{R}^n) \cap \mathscr{D}^\prime(F) =\{ 0\},
\end{equation}
where $\mathscr{D}^\prime(F)$ is the set of all distributions on $\mathbb{R}^n$ whose support is contained in the set $F$, i.e.
\begin{equation}\label{eqn-D'(F)}
\mathscr{D}^\prime(F) :=\bigl\{ \Lambda \in \mathscr{D}^\prime( \mathbb{R}^n) \colon \supp \Lambda \subset F \bigr\}
\end{equation}
and the space $W^{-m,q}(\mathbb{R}^n)$, usually defined as the dual space to $W^{m,p}(\mathbb{R}^n)$,  is the set of all distributions $\Lambda \in \mathscr{D}^\prime( \mathbb{R}^n) $ of the form
\begin{equation}\label{eqn-W-m,p'}
\Lambda=\sum_{0 \leq \vert \alpha \vert \leq m} (-1)^{\vert \alpha \vert}D^{\alpha}v_\alpha, \;\;  v_\alpha \in L^{q}(\mathbb{R}^n)
\end{equation}
equipped with the norm
\begin{equation}\label{eqn-W-m,p-norm}
\Vert \Lambda \Vert_{W^{-m,q}}^{q} =\inf \Bigl\{ \sum_{0 \leq \vert \alpha \vert \leq m} \vert v_\alpha \vert_{L^{q}(\mathbb{R}^n)}^{q} \colon  \;\;\eqref{eqn-W-m,p'} \mbox{ holds }\Bigr\}.
\end{equation}
In particular, a closed subset $F \subset \mathbb{R}^n$ is called $(1,p)$-\emph{polar} if and only if
\begin{equation}\label{eqn-mp'-polar}
  W^{-1,q}(\mathbb{R}^n) \cap \mathscr{D}^\prime(F) =\{ 0\}.
\end{equation}
Equivalently,  a closed subset $F \subset \mathbb{R}^n$ is called $(1,p)$-polar if and only for every
$$
\Lambda=v_0+\sum_{j=1}^n\Dj v_j,\quad  v_i \in L^{p}(\mathbb{R}^n),\  i=0,\cdots,n,
$$
$\supp{\Lambda}\subset F$ implies  $\Lambda=0$.
\end{definition}

The following result is obvious.

\begin{proposition} \label{prop-embedding-polar}
Assume that $n \in \mathbb{N}^\ast$,  $m,k>0$ and  $p,p' \in (1,+\infty)$. Let $q,q^\prime$ be the conjugate number to $p,p^\prime$. If
\[
W^{m,p}(\mathbb{R}^n) \subset W^{k,p^\prime}(\mathbb{R}^n),
\]
e.g. if
\[
m-\frac{n}{p} \geq k- \frac{n}{p^\prime},
\]
then every $(m,q)$-polar set $F$ is also a $(k,q^\prime)$-polar.
\end{proposition}

\begin{example}\label{example-singleton polar set}
Suppose that $F$ is a singleton, i.e. $F=\{a\}$, for some $a \in \mathbb{R}^n$. If
$\Lambda \in \mathscr{D}^\prime(F)$, then by \cite{Rudin-FA_1973}, there exists a finite set of numbers $c_\alpha$, $\alpha \in \mathbb{N}^n$, such that
\[
\Lambda=\sum_{\alpha}c_\alpha D^\alpha \delta_a,
\]
where $\delta_a$ is the Dirac delta distribution at $a$. Hence, $F$ is not an $(m,q)$-polar set if and only if $\delta_a \in  W^{-m,q}(\mathbb{R}^n)$. But, by the Sobolev Embedding Theorem (at least when $m$ is a natural number)  the latter conditions holds iff  the space $W^{m,p}(\mathbb{R}^n)$ is a subset of $\mathscr{C}_b(\mathbb{R}^n)$, i.e. iff $m-\frac{n}{p}>0$.
\end{example}

Therefore, the following holds.
\begin{corollary}
\label{cor-singleton polar set}
A singleton subset of $\mathbb{R}^n$ is an $(m,q)$-polar set  iff $m-\frac{n}{p} \leq 0$, i.e. if and only if $m+\frac{n}{q} \leq n$. In particular, a singleton subset of $\mathbb{R}^n$ is a $(1,p)$-polar set if and only if $1+\frac{n}{p} \leq n$. Thus,  a singleton subset of $\mathbb{R}^1$ is a $(1,p)$-polar set if and only if $p=1$.  A singleton subset of $\mathbb{R}^n$,   $n \geq 2$,  is a $(1,p)$-polar set if and only if
 \[
 p \geq \frac{n}{n-1}=1+ \frac{1}{n-1}.
 \]
\end{corollary}

The following result is a consequence of Corollary \ref{cor-singleton polar set} and \cite[Theorem 3.28]{Adams+Fournier_2003}.

\begin{theorem}\label{thm-density}
Assume that $a \in \mathbb{R}^n$. Then the set $\mathscr{C}_0^\infty (\mathbb{R}^n \setminus\{a\})$ is dense in $W^{m,p}(\mathbb{R}^n)$ iff  $p \leq \frac{n}{m}$.
\end{theorem}

Theorems 3.33 and 3.37 deal with the density of the set $\mathscr{C}_0^\infty (\mathbb{R}^n \setminus\{a\})$   in the space $W^{m,p}(\mathbb{R}^n \setminus\{a\})$.


\bibliographystyle{plain}

\end{document}